\documentclass[11pt,a4paper]{amsart}
\usepackage[latin1]{inputenc}
\usepackage[all]{xy}
\usepackage{amssymb, amsmath, amscd, geometry}
\usepackage{graphicx}
\usepackage{setspace}
\usepackage[normalem]{ulem}

\numberwithin{equation}{section}
\input xy
\xyoption{all}

\usepackage{latexsym}
\usepackage[usenames]{color}
\usepackage{epic} 
\usepackage{mathrsfs}
\usepackage{euscript}
\usepackage{rotating}
\usepackage{verbatim}
\usepackage{tikz-cd}

\usepackage{ulem}

\newtheorem{lemma}{Lemma}[section]
\newtheorem{theorem}[lemma]{Theorem}
\newtheorem{corollary}[lemma]{Corollary}

\newtheorem{prop}[lemma]{Proposition}
\newtheorem{cor}[lemma]{Corollary}

\newtheorem*{theorem*}{Theorem}

\theoremstyle{definition}
\newtheorem{definition}{Definition}[section]
\newtheorem{remark}{Remark}[section]

\newcommand{\tcr}{\textcolor{red}}

\newcommand{\Z}{{\mathbb Z}}      
      \newcommand{\C}{{\mathbb C}}

\newcommand{\h}{{H}}

\newcommand{\iny}{\hookrightarrow} 

\newcommand{\tr}{\operatorname{tr}}

\newcommand{\id}{\operatorname{id}}

\usepackage{hyperref}
\DeclareMathOperator{\CB}{\text{CB}}
\DeclareMathOperator{\Span}{span}

\newcommand{\oline}[1]{\overline{#1}}

\newcommand{\wt}[1]{\widetilde{#1}}
\newcommand{\norm}[1]{\left\lVert #1 \right\rVert}
\newcommand{\abs}[1]{\left\lvert #1 \right\rvert}
\newcommand{\e}{\epsilon}
\newcommand{\g}{\gamma}
\newcommand{\cbnorm}[1]{\left\lVert #1 \right\rVert_{\text{cb}}}
\newcommand{\ket}[1]{\vert #1 \rangle}
\newcommand{\bra}[1]{\langle #1 \vert}
\newcommand{\om}{\omega}

\renewcommand{\l}{\lambda}
\newcommand{\s}{\sigma}

\newcommand{\wh}[1]{\widehat{#1}}
\newcommand{\ip}[2]{\langle #1 | #2 \rangle}
\renewcommand{\d}{\delta}

\newcommand{\vp}{\varphi}

\newcommand{\op}[2]{\ket{#1}{\bra{#2}}}
\renewcommand{\sp}[2]{\langle #1, #2 \rangle} 
\renewcommand{\mp}[2]{\langle \langle #1, #2 \rangle \rangle} 

\makeatletter
\DeclareRobustCommand\smileotimes{\mathbin{\mathpalette\smile@otimes\relax}}
\newcommand{\smile@otimes}[2]{%
  \vbox{
    \ialign{##\cr
      \hidewidth$\m@th#1{}_\smile$\kern-\scriptspace\hidewidth\cr
      \noalign{\nointerlineskip\kern-1pt}
      $\m@th#1\otimes$\cr
    }%
  }%
}
\makeatother

\makeatletter
\DeclareRobustCommand\frownotimes{\mathbin{\mathpalette\frown@otimes\relax}}
\newcommand{\frown@otimes}[2]{%
  \vbox{
    \ialign{##\cr
      \hidewidth$\m@th#1{}_\frown$\kern-\scriptspace\hidewidth\cr
      \noalign{\nointerlineskip\kern-1pt}
      $\m@th#1\otimes$\cr
    }%
  }%
}
\makeatother

\title{No-signaling values of quantum games--an operator algebra perspective}

\author{Roy Araiza} 
\address{Department of Mathematics \& Illinois Quantum Information Science and Technology Center\\University of Illinois at Urbana-Champaign\\1409 W. Green St. Urbana, IL 61891. USA}
\email{raraiza@illinois.edu}

\author{Marius Junge}
\address{Department of Mathematics \& Illinois Quantum Information Science and Technology Center\\University of Illinois at Urbana-Champaign\\1409 W. Green St. Urbana, IL 61891. USA}
\email{junge@math.uiuc.edu}

\author{Carlos Palazuelos}
\address{Instituto de Ciencias Matem\'aticas (ICMAT)\\Departamento de An\'alisis Matem\'atico y Matem\'atica Aplicada \\
Facultad de Ciencias Matem\'aticas \\ Universidad Complutense de Madrid \\
Madrid 28040. Spain}
\email{cpalazue@ucm.es}

\thanks{MJ is partially supported by the NSF grants DMS 2247114 and Raise-TAG183917. C.P. is partially supported by Grant PID2023-146758NB-I00 funded by MICIU/AEI/10.13039/501100011033, by TEC-2024/COM-84-QUITEMAD-CM funded by the Comunidad de Madrid, and by Grant CEX2023-001347-S funded by MICIU/AEI/10.13039/501100011033. C.P. gratefully acknowledges financial support for this publication from the Fulbright Program, sponsored by the U.S. Department of State and the U.S.-Spain Fulbright Commission.}

\addtocounter{tocdepth}{-1}

\begin{document}


\maketitle

\begin{abstract}
The aim of this work is to study two-prover quantum games (i.e., games with quantum inputs and outputs) from an operator-algebraic and operator-space point of view. We characterize several notions of the value of such games by formulating them in terms of tensor norms in the category of operator spaces. The main results of the paper concern the description of the so-called no-signalling value of these games, for which we not only provide a precise operator-space formulation, but also establish close connections between this study and some problems in operator algebras. In particular, we show how the recent counterexample to Grothendieck's theorem for operator spaces given in \cite{Ara} can be understood as a direct consequence of results in quantum information theory. We also obtain new upper bounds on the gap between the no-signalling value and the quantum value of two-prover quantum games, improving the best previously known estimates.
\end{abstract}

\section{Introduction}

The interaction between nonlocal games or, equivalently, the theory of Bell inequalities in quantum information and the fields of operator algebras and operator spaces/systems has been extremely fruitful over the last decades \cite{palazuelos2016survey}. In its origins, this interaction mostly consisted of the development of mathematical results to solve problems within the framework of quantum information (see, for example, \cite{PerezWPVJ08tripartite, JungePPVW10, JungeP11low, HMPS19, LMPRSSTW11, PSSTW16}). However, in recent years, the application of techniques from nonlocal games to solve problems within the framework of operator algebras has become increasingly common (\cite{ReVi14, JNVWY, Ara}). This symbiotic relationship between both fields has motivated the study of nonlocal games in much more general contexts (\cite{ToTu2020, BrHaToTu23, BrHaTo23, cooney2015rank, La}), where the mathematics required to describe the corresponding objects are much more sophisticated. Perhaps the most natural extension of nonlocal games is the consideration of their quantum version, the so called quantum games, precisely the subject of study in the present work.

A natural way to understand a (classical) two-prover nonlocal game  $\mathcal G$ is by realizing it as a pointwise positive element $\xi_\mathcal G$ in the tensor product of certain vector valued  $L_p$ spaces: $$\ell_1^{|I_A|}(\ell_\infty^{|O_A|})\otimes \ell_1^{|I_B|}(\ell_\infty^{|O_B|}),$$where $I_A$, $I_B$ are the sets of questions for the players of the game, say Alice and Bob, respectively, and $O_A$, $O_B$ are the sets of answers. It turns out that the different values of these games; that is, the best winning probability of the game when the players are restricted to a certain type of strategies, is reflected in the fact of considering different tensor norms of $\xi_\mathcal G$. In this way, given a game $\mathcal G$, its classical and quantum values can be shown \cite{palazuelos2016survey}  to be the injective norm $\|\xi_{\mathcal G}\|_\epsilon$  in the category of Banach spaces and the minimal norm $\|\xi_{\mathcal G}\|_{min}$ in the category of operator spaces.  In fact, due to their relevance as a set of strategies encompassing all those in which no information is exchanged between the players (such as classical and quantum ones), no-signalling strategies play a prominent role in this context and, consequently, so does the value of a game when the players are restricted to this type of strategy. In this case, we speak of the no-signalling value of $\mathcal G$ and, as was shown in \cite{amr2019optimal}, it can be expressed as
\begin{align}\label{Eq: NS_Classical_Games}
\|\xi_{\mathcal G}\|_{NOS}=\inf\Big\{\|\xi_1\|_{\ell_1^{|I_A|}(\ell_\infty^{|O_A|}(\ell_1^{|I_B|}(\ell_\infty^{|O_B|})))}+\|\xi_2^T\|_{\ell_1^{|I_B|}(\ell_\infty^{|O_B|}(\ell_1^{|I_A|}(\ell_\infty^{|O_A|})))}:\, \xi_{\mathcal G}=\xi_1+\xi_2\Big\},
\end{align}where the superscript $T$ denotes the flip operation between the registers $I_AO_A$ and $I_BO_B$. This value is a symmetrization of the norm $\|\xi_{\mathcal G}\|_{\ell_1^{|I_A|}(\ell_\infty^{|O_A|}(\ell_1^{|I_B|}(\ell_\infty^{|O_B|})))}$
which precisely describes the value of the game $\mathcal G$ under strategies where only the flow of information from Alice to Bob is restricted, but not in the opposite direction. This is usually referred to as the \emph{one-way no-signalling value} of $\mathcal G$ and motivates understanding the no-signalling value as the \emph{symmetrized one-way} value of the game.

The importance of nonlocal games in various fields of quantum information, such as cryptography or communication complexity problems, as well as the relationship between these games and different problems in functional analysis, has motivated the study of quantum games in recent years. These games, where the interaction between the players and the referee occurs via quantum states, exhibit a much greater richness than their classical counterparts, while at the same time allowing a similar mathematical treatment to a certain extent. Indeed, as we will show, not only given a quantum game $\mathcal G$, it can be identified with a positive tensor $\xi_\mathcal G$ in the noncommutative analog of the space previously considered; that is, $$S_1(I_A; S_\infty(O_A))\otimes S_1(I_B,S_\infty(O_B)),$$where now $I_A, I_B$ y $O_A, O_B$ refer to the Hilbert spaces that describe the questions and answers of the game, but again, different values of the game $\mathcal G$ will be identified with different tensor norms of $\xi_\mathcal G$. However, quantum games, far from being a mere generalization of their classical counterparts, present a much greater mathematical richness. In particular, there is an alternative and natural way to describe quantum games through a tensor $\eta_\mathcal G$ in  $$C_D\otimes S_1(I_A,O_A)\otimes S_1(I_B,O_B),$$ where we have $\xi_\mathcal G=\eta_\mathcal G^*\eta_\mathcal G$ and where $C_D$  is a column space of certain dimension $D$, in such a way that the different values of the game $\mathcal G$ can also be described by certain norms of $\eta_\mathcal G$ in this last space. We will refer to $\xi_{\mathcal G}$ as the \emph{full tensor} of the game and to $\eta_{\mathcal G}$ as the \emph{half tensor}.

The first part of our work was developed in parallel, albeit independently, with \cite{Crann}, where a similar approach was used to describe various values of games, such as classical, quantum, commuting quantum, and so on. In this way, one can see that, in complete analogy with classical games, the classical, quantum and one-way no-signalling value of a quantum game coincide  with $\|\xi_{\mathcal G}\|_{\epsilon}$, $\|\xi_{\mathcal G}\|_{min}$ and $\|\xi_{\mathcal G}\|_{S_1(I_A; S_\infty(O_A; S_1(I_B; S_\infty(O_B))))}$ respectively, where $\xi_{\mathcal G}$ is the full tensor of the game in $S_1(I_A; S_\infty(O_A))\otimes S_1(I_B,S_\infty(O_B))$. Moreover, one can also describe the second and third values we just mentioned by using the half tensor of the game, in such a way that the quantum value of $G$ is 
$\|\eta_{\mathcal G}\|^2_{C_D(S_1(I_A,O_A)\otimes_{\min} S_1(I_B,O_B))}$, while  $\|\eta_{\mathcal G}\|^2_{C_D(S_1(I_A,O_A)\otimes_{h} S_1(I_B,O_B))}$ is the one-way no-signalling value, 
where here $\otimes_h$ denotes the Haagerup tensor norm.

While we do discuss some of these results in our paper, the main focus of our study is on the no-signalling value of a quantum game.  There are two main reasons for the interest in this value. On the one hand, the relevance of its physical significance and its interpretation in terms of information-theoretic protocols call for a mathematical description which, unlike the situation for the previous values, has not so far been as clearly expressed in terms of tensor norms. On the other hand, as we will show in this work, there exists a deep connection between the study of the no-signalling value of a game and certain open problems in operator algebras. The depth of these connections stems from specific quantum subtleties that arise when moving from classical to quantum games, which are reflected, for example, in the greater difficulty of handling the symmetrization of the one-way no-signalling value when working with noncommutative $L_p$-spaces.   In fact, the mathematical study required in this context is motivated by the recent proof of the failure of a matricial version of Grothendieck's theorem for operator spaces \cite{Ara}, which we will demonstrate is closely tied to the study of the no-signalling value of quantum games.

In this work, we will show that understanding the no-signalling value as a symmetrized version of the one-way no-signalling values, together with the interpretation of these in terms of information transmission protocols between Alice and Bob, is key to providing an adequate mathematical description. As part of this process, we will also see that, in the quantum case, the natural value to consider from the perspective of operator algebras is the sub-no-signalling value of $\mathcal G$.

In this regard, the main result of our paper can be stated as a description of the sub-no-signalling value of a quantum game.
\begin{theorem}\label{thm: Main_I}
Let $G$ be a quantum games which is identified with a tensor $\xi_G\in S_1(I_A; S_\infty(O_A))\otimes S_1(I_B,S_\infty(O_B))$ and with a tensor $\eta_G\in C_D\otimes S_1(I_A,O_A)\otimes S_1(I_B,O_B).$ Then, if we denote by $\omega_{SNOS}(G)$ the sub-no-signalling value of $G$, we have
\begin{enumerate}
\item [a)] $\omega_{SNOS}(G)=\inf\{\|\xi_1\|_{S_1(I_A; S_\infty(O_A; S_1(I_B; S_\infty(O_B)))}+\|\xi_2^T\|_{S_1(I_B; S_\infty(O_B; S_1(I_A; S_\infty(O_A)))}\}$,  where the infimum runs over all positive elements $\xi_1$, $\xi_2$ such that $\xi_G\leq \xi_1+\xi_2$ and the superscript $T$ refers to the suitable rearrangment of the tensor.
\item [b)] $\omega_{SNOS}(G)=\|\eta_G\|^2_{C_D((S_1(I_A,O_A)\otimes_{\mu}S_1(I_B,O_B))},$ where $\otimes_{\mu}$ refer to the symmetrized Haagerup tensor norm in the category of operator spaces.
\end{enumerate}
\end{theorem}

Item a) is actually a generalization to the quantum case of the corresponding result for classical games already stated in Eq. (\ref{Eq: NS_Classical_Games}), where the noncommutative spaces are replaced by commutative ones. However, the situation in the noncommutative case is much more complex. In particular, while the consideration of sub-no-signalling strategies is merely a tool in the commutative case, since it can be shown that the sub-no-signalling value   $\omega_{SNOS}(G)$ and the no-signalling value  $\omega_{NOS}(G)$ coincide in that case, this equality is not at all clear in the noncommutative case. However, it can be shown that for any bipartite quantum game $G$ we have
\begin{align}\label{eq. Equiv_NS_SNS}
\omega_{NOS}(G)\leq \omega_{SNOS}(G)\leq 8\cdot \omega_{NOS}(G).
\end{align} In particular, Theorem \ref{thm: Main_I} applies to the value $\omega_{NOS}(G)$, up to a constant.  In addition, we believe that the description of $\omega_{SNOS}(G)$ provided in item a) above can be of independent interest to tackle some other problems in the field such as parallel repetition theorems for quantum games, upper/lower bounds for the distance between nonsingalling correlations and quantum correlations, etc.

The comparison between $\omega_{NOS}(G)$ and $\omega_{SNOS}(G)$ is closely related to the problem of lifting families of commuting partial isometries and families of commuting contractions, and the proof of Equation~(\ref{eq. Equiv_NS_SNS}) is rather technical. Since the results presented in this work do not depend on this estimate, and including its proof would considerably lengthen the paper, we have decided not to include it here. The complete proof will be presented in a forthcoming paper.

However, the significance of Theorem \ref{thm: Main_I} mainly lies in item b), which connects the problem of determining certain values for quantum games with some versions of Grothendieck's theorem in the category of operator spaces. Indeed, a direct consequence of Theorem \ref{thm: Main_I} is the following.
\begin{theorem}\label{thm: Main_II}
The following assertions are equivalent:
\begin{enumerate}
\item[1.] There exists a constant $c_1$ such that for every quantum game $G$ it holds that $$\omega_{SNOS}(G)\leq c_1\omega^*(G).$$
\item[2.] There exists a constant $c_2$ such that for every tensor $t\in C\otimes S_1\otimes S_1$ that is obtained as a half tensor of a quantum game it holds that
$$\|t\|_{C(S_1\otimes_{\mu} S_1)}\leq c_2\|t\|_{C(S_1\otimes_{min} S_1)}.$$
\end{enumerate}

Moreover, if $c_1$ is the best constant satisfying item 1. and  $c_2$ is the best constant satisfying item 2., then $$c_1=c_2^2.$$
\end{theorem}

Since we know of some examples of (even classical) games showing that item 1 in the previous theorem fails, we conclude that item 2 must also be false, providing a negative answer to the open question posed by Blecher in  \cite{Blecher92} (see also \cite{Pi02}) and  which can be referred to as the existence of a Grothendieck theorem in the category of operator spaces. In this way, the approach presented in this work allows for a deeper understanding of the recent resolution of that problem in \cite{Ara}. Indeed, the present work is strongly motivated by the article \cite{Ara}, as it was in that earlier study that the deep connections between certain problems in operator algebras and the no-signalling value of quantum games first began to emerge. It is important to note, though,  that the main result in \cite{Ara} is much stronger than the one followed from this work, as it provides a counterexample to item 2 above when the spaces $S_1$ are replace by its commutative analog $\ell_1$, a case for which the approach followed in the present paper does not offer any insight. In any case, we find it instructive - and for this reason we include it in this work - to show how a nonlocal game can be translated into a concrete tensor in the space $C \otimes S_1^n \otimes S_1^n$ that provides a counterexample to item 2 of the previous theorem and, consequently, to the matricial version of Grothendieck's theorem.

The connections revealed in the previous paragraphs and particularly the resolution of the problem regarding Grothendieck's theorem, motivate us to examine this problem from a more quantitative perspective. Hence, in the final part of this work, we establish some upper bounds for the largest possible value of the constant $c_2$ above, as a function of the rank of the tensor  $t$. Specifically, we can prove:
\begin{theorem}\label{thm: Main_III}
    Given any two-prover quantum game $\mathcal G(\ket{\psi}, P)$ with corresponding half tensor $\eta_\mathcal G$ we have the following inequality:\[
\norm{\eta_{\mathcal G}}_{C_D (S_1(I_A,O_A) \otimes_\mu S_1(I_B,O_B))} \leq K\sqrt{\operatorname{rank}(\eta_\mathcal G)}\norm{\eta_\mathcal G}_{C_D(S_1(I_A,O_A) \otimes_{min} S_1(I_B,O_B))}, 
    \]where $K$ is a universal constant.
    
    Hence, $\omega_{SNOS}(\mathcal G)\leq K^2 \operatorname{rank}(\eta_\mathcal G)\omega^*(\mathcal G)$ for every two-prover quantum game $\mathcal G$.
    \end{theorem}
    
The previous theorem improves upon the upper bounds established in \cite{Crann} in several respects (see Section~\ref{Sec: Grothendieck's Theorem for Operator Spaces} for further discussion). In a forthcoming paper, we will show that this theorem is optimal up to logarithmic factors. Indeed, additional examples of quantum games, which may exhibit more extreme behaviors and further illustrate the optimality of these bounds, will also be presented there. In that work, we aim to demonstrate the utility of operator space techniques in this context.

\section{Preliminaries}

\subsection{Operator spaces}

In this section we briefly recall the basic notions on operator spaces and their tensor products used throughout the paper; see \cite{Ef00, Pi03} for detailed references. 

An \emph{operator space} is a closed subspace \(X \subset \mathcal B(H)\), where \(\mathcal B(H)\) denotes the bounded operators on a Hilbert space \(H\). Such an inclusion induces matrix norms \(\|\cdot\|_n\) on \(M_n(X)\) via \(M_n(X)\subseteq M_n(\mathcal B(H))\simeq \mathcal B(H^{\oplus n})\). By Ruan's theorem \cite{Ru88}, a complex Banach space \(X\) equipped with norms  \(\|\cdot\|_n\) on \(M_n(X)\) satisfying Ruan's axioms arises in this way; we then say that \(X\) has an \emph{operator space structure} (o.s.s.).

A linear map \(T:X\to Y\) between operator spaces is \emph{completely bounded} if 
\[
\|T\|_{cb}:=\sup_n\|\id_n\otimes T:M_n(X)\to M_n(Y)\|<\infty.
\]
If \(\|T\|_{cb}\le1\), \(T\) is \emph{completely contractive}.  It is a \emph{complete isomorphism} (resp. \emph{complete isometry}) when both \(T\) and \(T^{-1}\) are completely bounded (resp. contractive).

Given an operator space $X$, then $X^*$ has a privileged operator space structure with norms on $M_n(X^*)$ given via the identification \[M_n(X^*)  =  \text{CB}(X, M_n),\] where the latter denotes the space of completely bounded maps $T: X \to M_n$ endowed with the norm $\|\cdot\|_{cb}$.  Every $C^*$-algebra $A$ has a natural operator space structure. Therefore, by duality, $\ell_1$ and the space of trace class operators $S_1$ also have a natural operator space structure, as they are the duals of $c_0$ and the space of compact operators $S_\infty$, respectively. More generally, given two complex Hilbert spaces $\h$ and $K$, we will denote by $S_p(\h, K)$ the spaces of $p$-Schatten class operators from $H$ to $K$. In particular, we will write $S_p(\h)$ when $\h=K$. One can endow the spaces $S_p(\h, K)$ with an o.s.s. by interpolating between $S_1(\h, K)$ and $S_\infty(\h, K)$. When, working in the operator space category, the duality between $S_p(\h, K)$ and $S_{p'}(\h, K)$, when $p$ and $p'$ are conjugate indices, is given \cite{Ef00} by 
\begin{align}\label{duality OS}
\langle A, B\rangle=\tr(A^TB).
\end{align}

The theory of vector-valued noncommutative $L_p$ spaces was developed by Pisier in \cite{Pi98}. Briefly, given an operator space \(X\), once operator space structures have been assigned to \(S_\infty(H; X)\) and, by duality, also to \(S_1(H; X)\), one can define an operator space structure on \( S_p(H; X) = (S_\infty(H; X), S_1(H; X))_{1/p} \),  for every \(1 \leq p \leq \infty\) by interpolation. In particular, we can define the operator spaces \(S_p(H; S_q(K))\) for all \(p,q\) in such a way that many of the properties of the classical spaces are preserved. For instance, the duality relations 
\begin{align}\label{Eq. Dual_p,q}
S_p(H; S_q(K))^* = S_{p'}(H; S_{q'}(K)),
\end{align}
hold, where (\ref{duality OS}) denotes the associated dual action.

According to \cite[Lemma 1.7]{Pi98} given $x\in S_\infty(H; X)$, where $X$ is any given operator spaces and $1\leq p\leq \infty$, 
\begin{align}\label{Ec:S_infty(X)}
\|x\|_{S_\infty(H;X)}=\sup\Big\{\big\|(A\otimes 1_{X})x(B\otimes 1_{X})\big\|_{S_{p}(H;X)}: A,B\in B_{S_{2p}(H)}\Big\}.
\end{align}

Also, \cite[Theorem 1.5]{Pi98} gives us, for every $1\leq p<\infty$, that
\begin{align}\label{Ec:S_1(X)}
\|x\|_{S_p(H;X)}=\inf \big\{\|A\|_{S_{2p}(H)}\|Z\|_{S_\infty(H; X)}\|B\|_{S_{2p}(H)}\big\},
\end{align}where the last infimum runs over all representations of the form $x=\big(A\otimes 1 \big)Z\big(B\otimes1 \big)$.

Throughout this work we shall frequently refer to the \textit{column} and \textit{row} Hilbert operator spaces.  
Given a Hilbert space $H$, its column operator space $H_c$ is defined through the canonical identification
\[
H \simeq \mathcal{B}(\mathbb{C}, H).
\]
In the finite-dimensional case $H = \ell_2^n$, we simply write $C_n$.  
The row operator space is defined by duality as $H_r = (H_c)^*$. More precisely, since the dual space  $H^*$ is naturally identified with the conjugate Hilbert space $\overline{H}$, we have $(H_r)^* = (\overline{H})_c$ and $(H_c)^* = (\overline{H})_r$.  For $H = \ell_2^n$, we denote $R_n$ the corresponding row space and we write $R_n^* = C_n$ and $C_n^* = R_n$.

The ket-bra notation will be extremely useful in this work. Recall that given a Hilbert space $H$, we use the notation $\lvert \xi \rangle \in H$ to denote elements of the space, whereas we use $\langle \xi \rvert \in H^*$ to denote elements of the dual. In this sense, if we consider $H = \mathbb{C}^n$, we naturally have $\lvert \xi \rangle = \sum_{i=1}^n \alpha_i \lvert i \rangle \in C_n,$
while their duals in $R_n$ are the row vectors $\langle \xi \rvert = \sum_{i=1}^n \overline{\alpha_i} \langle i \rvert$. Here, $(|i\rangle)_{i=1}^n$ denotes the canonical (or computational) basis of $\C^n$. This notation is also very useful for denoting rank-one operators. Indeed, given two vectors $\xi$, $\eta$, we denote by $\lvert \xi \rangle \langle \eta \rvert : \mathbb{C}^n \to \mathbb{C}^n$ the operator defined by $\lvert \xi \rangle \langle \eta \rvert ( \lvert z \rangle)= \langle \eta \rvert z \rangle \lvert \xi \rangle$.

\subsubsection{Tensor norms}

An abstract framework for operator space tensor products was established by Blecher and Paulsen in \cite{BlPa91}, where they introduced the \textit{projective} and \textit{injective} operator space tensor norms.

Given operator spaces $X$ and $Y$, and $w \in M_n(X \otimes Y)$, the \textit{operator space projective norm} of $w$ is defined as
\[
\|w\|_{\wedge}
:= \inf\{\|\alpha\|\, \|x\|\, \|y\|\, \|\beta\|\},
\]
where the infimum ranges over all $p,q,n \in \mathbb{N}$, $x \in M_p(X)$, $y \in M_q(Y)$, and $\alpha \in M_{n,pq}$, $\beta \in M_{pq,n}$ satisfying $w = \alpha(x \otimes y)\beta$.  
The corresponding operator space is denoted by $X \hat{\otimes} Y$.  
As suggested by its name, the projective tensor product preserves (complete) quotients, and serves as the predual of the space $\mathrm{CB}(X, Y^*)$, yielding the complete isometry
\[
(X \hat{\otimes} Y)^* = \mathrm{CB}(X, Y^*).
\]

The \textit{injective} (or \textit{minimal}) tensor product is defined via inclusion into the minimal C$^*$-algebra tensor product.  
If $X \subset \mathcal{B}(H)$ and $Y \subset \mathcal{B}(K)$, then there is a completely isometric embedding
\[
X \otimes_{\min} Y \subset \mathcal{B}(H \otimes_2 K),
\]
so that the minimal norm is preserved under (complete) isometries.  
For $w \in M_n(X \otimes Y)$, this norm satisfies
\[
\|w\|_{M_n(X \otimes_{\min} Y)}
= \sup \big\|
(\id \otimes T \otimes S)(w)
\big\|_{M_n(\mathcal{B}(H \otimes K))},
\]
where the supremum is taken over all Hilbert spaces $H,K$ and complete contractions $T \!:\! X \to \mathcal{B}(H)$, $S \!:\! Y \to \mathcal{B}(K)$.

In contrast with the projective case, we have the completely isometric embedding
\begin{align}\label{Eq: min_cb}
X \otimes_{\min} Y \subseteq \mathrm{CB}(X^*, Y).
\end{align}
Hence, the projective and injective norms are dual to each other, and for finite-dimensional operator spaces $X, Y$, one has the completely isometric identifications
\[
(X \hat{\otimes} Y)^* = X^* \otimes_{\min} Y^*,
\qquad
(X \otimes_{\min} Y)^* = X^* \hat{\otimes} Y^*.
\]

Given operator spaces $X \subset \mathcal{B}(H)$ and $Y \subset \mathcal{B}(K)$, the \textit{Haagerup tensor norm} of $w \in M_n(X \otimes Y)$ is defined by
\begin{align}\label{Ec. Haagerup_norm}
\|w\|_{M_n(X \otimes_h Y)} := \inf \big\{ \|u\|_{M_{n,r}(X)} \, \|v\|_{M_{r,n}(Y)} \big\},
\end{align}
where the infimum is taken over all $r \in \mathbb{N}$ and matrices $u, v$ satisfying 
$w_{ij} = \sum_{k=1}^r u_{ik} \otimes v_{kj}$ for all $i,j$.  
The resulting operator space, denoted $X \otimes_h Y$, is called the \textit{Haagerup tensor product}.  
Although $\otimes_h$ is injective, projective, and self-dual, it is well known that it is not commutative.

These tensor norms behave well under tensoring of linear maps.  
If $\alpha$ denotes any of the projective, injective, or Haagerup norms, and $T_i : X_i \to Y_i$ are completely bounded maps for $i=1,2$, then
\[
T_1 \otimes T_2 : X_1 \otimes_\alpha X_2 \longrightarrow Y_1 \otimes_\alpha Y_2
\]
is completely bounded, with
\[
\|T_1 \otimes T_2\|_{cb} = \|T_1\|_{cb} \, \|T_2\|_{cb}.
\]
This property is known as the \textit{metric mapping property} in the category of operator spaces.  
By iterating this construction, one can extend the above tensor norms to tensor products involving more than two operator spaces.  
All three norms are associative, and their fundamental properties---such as projectivity, injectivity, and duality---remain valid and behave consistently under these extensions. It is worth noting that, unlike the projective and minimal norms, the Haagerup norm fails to be commutative.

The tensor norms introduced above are closely connected when considering the column and row operator spaces.  
Indeed, for every operator space $X$, one has the complete isometries 
\begin{align}\label{Eq: indetity_min_h_1}
C \otimes_{\min} X = C \otimes_h X 
\quad \text{and} \quad 
X \otimes_{\min} R = X \otimes_h R.
\end{align} Given an operator space $X$, one usually writes $C \otimes_{\min} X=C(X)$
By duality, the corresponding relations
\begin{align}\label{Eq: indetity_min_h_2}
R \,\hat{\otimes}\, X = R \otimes_h X
\quad \text{and} \quad
X \,\hat{\otimes}\, C = X \otimes_h C
\end{align}
also hold completely isometrically.

Given any operator space $X$, there are two completely isometric identifications which will play an important role in this work \cite[Corollary 5.10 \& 5.11]{Pi03}\footnote{Although the cited reference only discusses the case of $\h=K$, the proof for the statement about rectangular matrices is completely analogous.}. The first one is given by the map $J:\h_c\otimes_h X\otimes_h (K^*)_r\rightarrow S_\infty(K, \h)\otimes_{min} X$, defined by
\begin{align}\label{Haagerup K}
J\Big(\sum_{i,j}e_{i,1}\otimes x_{i,j}\otimes e_{1,j}\Big)=\sum_{i,j} e_{i,j}\otimes x_{i,j}.
\end{align} The second one is given by the map $K:(\h^*)_r\otimes_h X\otimes_h K_c\rightarrow S_1(K, \h)\hat{\otimes} X$, defined by
\begin{align}\label{Haagerup T}
K\Big(\sum_{i,j}e_{1,i}\otimes x_{i,j}\otimes e_{j,1}\Big)=\sum_{i,j} e_{i,j}\otimes x_{i,j}.
\end{align} It is instructive to see the above identifactions using bra-ket notation: 
\begin{align*}
&J\big(\ket{\psi} \otimes x \otimes \bra{\eta}\big)=\sum_{i,j}\op{\psi}{\eta}\otimes x,\\&
K\big(\bra{\eta}\otimes x\otimes \ket{\psi}\big)=\sum_{i,j}\op{\overline{\eta}}{\overline{\psi}}\otimes x.
\end{align*}

As mentioned above, the non-commutativity of the Haagerup tensor norm naturally motivates the definition of the \textit{symmetrized Haagerup tensor product}, denoted by $\otimes_{\mu}$.  
For $w \in M_n(X \otimes Y)$, its norm is given by
\begin{align}\label{Def_mu_norm}
\|w\|_{\mu}
:= \inf \big\{
\|w_1\|_{M_n(X \otimes_h Y)}
+ \|w_2^T\|_{M_n(Y \otimes_h X)} :\,
w = w_1 + w_2
\big\},
\end{align}
where, for $w = \sum_i a_i \otimes x_i \otimes y_i$, we set
$w^T = \sum_i a_i \otimes y_i \otimes x_i$. This can be rephrased in the following way (see \cite[Theorem~1]{Oik99}):  
Given operator spaces $X$ and $Y$, consider the map
\begin{align}\label{mu_quotient}
Q : (X \otimes_h Y) \oplus_1 (Y \otimes_h X) \longrightarrow X \otimes_{\mu} Y
\end{align}
defined on the direct sum of the linear tensor products by 
\[
Q(u \oplus v) = u + v^T.
\]
Then, $Q$ extends to a complete metric surjection from 
$(X \otimes_h Y) \oplus_1 (Y \otimes_h X)$ onto $X \otimes_{\mu} Y$.

Given an operator space $X$, its universal unital operator algebra $OA(X)$ (see \cite[Chapter~6]{Pi03} for a detailed construction) is characterized by the following universal property:  
for every complete contraction $T : X \to \mathcal{B}(H)$, there exists a unique unital homomorphism 
$\pi : OA(X) \to \mathcal{B}(H)$ extending $T$, where $X$ is naturally embedded in $OA(X)$.

This algebra plays a key role since, for operator spaces $X$ and $Y$, one has (see \cite[Lemma~5]{Oi99})
\begin{align*}
X \otimes_{\mu} Y
\subset OA(X) \otimes_{\max} OA(Y)
\quad \text{completely isometrically,}
\end{align*}
where $\otimes_{\max}$ denotes the maximal tensor product of unital operator algebras.  
Consequently, for any $w \in M_n(X \otimes Y)$,
\begin{align}\label{Definition_mu_norm}
\|w\|_{M_n(X \otimes_{\mu} Y)}
= \sup
\big\|
(\id \otimes T \odot S)(w)
\big\|_{M_n(\mathcal{B}(H))},
\end{align}
where the supremum runs over all complex Hilbert spaces $H$ and all complete contractions
$T : X \to \mathcal{B}(H)$ and $S : Y \to \mathcal{B}(H)$
with commuting ranges.  
Here, $T \odot S$ denotes the bilinear map defined by
$T \odot S(x \otimes y) = T(x)S(y)$.

In 1991, Blecher \cite{Blecher92} conjectured a fully functorial version of the \emph{operator space Grothendieck theorem}: for arbitrary C$^*$-algebras $A$ and $B$, it should hold that
\begin{equation}
A^* \otimes_{\min} B^* \approx A^* \otimes_\mu B^*, \label{Blecher conjecture}
\end{equation}
completely isomorphically. In a groundbreaking result by Pisier and Shlyakhtenko \cite{Pi02}, complemented by the work \cite{HaMu08}, the authors proved that Equation~\eqref{Blecher conjecture} holds isomorphically. In particular, it follows from \cite[Theorem 1.1]{HaMu08} (see \cite[Proposition~18.2]{pisier2012grothendieck} for a restatement of this result) that, for any $z \in S_1 \otimes S_1$, one has
\begin{align}\label{Eq: GT_OS}
\|z\|_{S_1 \otimes_{\mu} S_1} \leq 2\|z\|_{S_1 \otimes_{\min} S_1}.
\end{align}

Blecher's conjecture, in its most general form, remained open until it was disproved in \cite{Ara}. We will return to this question in Section~\ref{Sec: Grothendieck's Theorem for Operator Spaces} of the present work.

\subsection{Completely positive maps and the Choi matrix}

Let us recall some basic facts concerning completely positive maps and their Choi representations.  

Given two finite-dimensional Hilbert spaces $H$ and $K$,  a linear map $\varphi : S_\infty(H) \to S_\infty(K)$ is said to be \emph{completely positive (cp)} if, for every $n \in \mathbb{N}$, the amplified map
\[
\id_{n} \otimes \varphi : M_n(S_\infty(H)) \longrightarrow M_n(S_\infty(K))
\]
is positive; that is,
\[
(\id_{n} \otimes \varphi)(X) \ge 0
\quad \text{whenever } X \ge 0.
\]

Given a linear map $\varphi : S_\infty(H) \to S_\infty(K)$, its associated \emph{Choi matrix} is defined as
\begin{align}\label{map_tensor_I}
C_\varphi = \sum_{i,j} |i\rangle\langle j| \otimes \varphi_{i,j}
\;\in\; S_1(H) \otimes S_\infty(K),
\end{align}
where $(|i\rangle)_i$ is any orthonormal basis of $H$ and $\varphi_{i,j} := \varphi(|i\rangle\langle j|)$ for all $i,j$. Note that $\lvert i \rangle \langle j \rvert$ is identified with the matrix having a $1$ in the $(i,j)$ entry and $0$ elsewhere.

It is well known that the map $\varphi$ is completely positive if and only if its Choi matrix $C_\varphi$ is positive semidefinite.  Moreover, one has the reconstruction formula
\begin{align}\label{map_tensor_II}
\varphi(\rho)
=\tr_{S_\infty(H)}
\Big(
C_\varphi \cdot (\rho^{T} \otimes 1_{S_\infty(K)})
\Big),
\end{align}
where $\rho^T$ denotes the transpose of $\rho$ with respect to the chosen basis and $\tr_{S_\infty(H)}$ is the partial trace over $S_\infty(H)$.

Equations (\ref{map_tensor_I}) and (\ref{map_tensor_II}) establish an identification between completely positive maps $\varphi:S_\infty(\h)\rightarrow S_\infty(K)$ and positive elements in $S_\infty(\h)^*\otimes S_\infty(K)$ and, moreover, it leads to a (completely) isometric identification 
\begin{align}\label{identification cb-min}
CB(S_1(\h), S_1(K))=S_\infty(\h; S_1(K)).
\end{align}

Let us also recall that $\varphi:S_\infty(\h)\rightarrow S_\infty(K)$ is completely positive if and only if $\varphi^\dag: S_\infty(K)\rightarrow S_\infty(\h)$ is completely positive. Here, $\varphi^\dag$ is the adjoint map of $\varphi$ defined by means of the duality
\begin{align}\label{duality Hilbert}
\langle\langle A,\varphi^\dag (B)\rangle\rangle=\langle\langle \varphi(A), B\rangle\rangle\hspace{0.5 cm}\text{for every $A\in S_\infty(\h),\, B\in S_\infty(K)$},
\end{align}where $\langle\langle X,Y\rangle\rangle=\tr(X^\dag Y)$ is the standard inner product. 

Moreover, in this case, $\varphi$ is trace preserving if and only if $\varphi^\dag$ is unital.

\begin{remark}\label{rmk: notation_tensors}
 Since quantum channels are completely positive and trace preserving maps, every quantum channel can be naturally seen as a map $\varphi:S_1(\h)\rightarrow S_1(K)$, in the sense that it defines a (complete) contraction between these spaces. According to Equation (\ref{identification cb-min}) and the duality $S_\infty(\h)^*=S_1(\h)$ explained before, we will write $$C_\varphi\in S_\infty(\h)\otimes S_1(K).$$ On the other hand, the adjoint map of a quantum channel can be naturally seen as a map $\varphi^\dag: S_\infty(\h)\rightarrow S_\infty(K)$, which will be identified with a tensor $$C_{\varphi^\dag}\in S_1(\h)\otimes S_\infty(K).$$ Note that in finite dimension, the spaces $S_\infty(\h)$ and $S_1(\h)$ are linearly isomorphic, but the previous notation indicates the natural spaces of these tensors when considering the right norms. Obviously, the remark also applies when considering rectangular matrices $S_\infty(\h, K)$.
\end{remark}

Note that 
\begin{align}\label{Eq. trace_adjoint}
\tr(\varphi(A)B)=\langle\langle B^\dag,\varphi(A)\rangle \rangle=\langle\langle \varphi^\dag(B^\dag),A\rangle \rangle=\tr((\varphi^\dag(B^\dag))^\dag A)=\tr(A(\varphi^\dag(B^\dag))^\dag).
\end{align}

This inequality allows us to relate the duality $\langle \cdot ,\cdot \rangle$ with $ \varphi^\dag$. Indeed, it holds that $$\langle B,\varphi(A)\rangle=\tr (\overline{\varphi^\dag(\bar{B})}A^T)=\langle \overline{\varphi^\dag(\bar{B})},A\rangle.$$

The previous calculations become more clear when $\varphi$ is completely positive. In this case, we have $(\varphi^\dag(B^\dag))^\dag=\varphi^\dag(B)$ and one obtains 
\begin{align}
\tr(\varphi(A)\, B)=\tr(A \, \varphi^\dag(B)). \label{eq: duality for completely positive}
\end{align}

\begin{remark}\label{rem:Choi_adjoint}
Throughout the manuscript, we will often establish corresponding characterizations between the Choi matrix of a linear map and the Choi matrix of its adjoint operator. It is easy to see that if $F$ denotes the flip operator, then
\[
C_{\varphi^\dag} = F(\overline{C_\varphi}).
\]
Note that when $\varphi$ is completely positive, $C_\varphi$ is a positive semidefinite element, so we have
\[
C_{\varphi^\dag} = F(C_\varphi^T).
\]
This means that, up to conjugation---which will be irrelevant in this work---the tensors $C_{\varphi^\dag}$ and $C_\varphi$ are essentially the same, once the tensor indices are rearranged accordingly.

	Thus, in the remainder of this work, whenever we compute a norm, the indices $I_A$, $O_A$, $I_B$, and $O_B$ will determine the specific rearrangement of the tensor we are using.
	\end{remark}

In this work, we will frequently deal with linear maps that are completely contractive and/or completely positive.  
The following two results provide, respectively, useful characterizations of these classes of maps. They can be found in  \cite[Theorem 8.4]{paulsen2002completely} and  \cite[Corollary 2.27]{Watrous} respectively. 
\begin{theorem}\label{thm: cb decomposition as sum}
Let $H$ and $K$ be two finite dimensional Hilbert spaces and $\Phi: S_\infty(K) \to S_\infty(\h)$ be a linear completely bounded map. Then there exist operators $A_i \in S_\infty(K, H)$ and $B_i \in S_\infty(H,K), i=1,\dotsc,m$ such that \[
\Phi(x) = \sum_{i=1}^m A_i x B_i,\hspace{0.3 cm} x\in S_\infty(K),
    \] where $m \leq \dim{H}\dim{K}$. Moreover, \[
\cbnorm{\Phi} = \inf \Big\{  \Big\|\sum_{i}A_iA_i^\dag\Big\|^{\frac{1}{2}} \Big\|\sum_i B_i^\dag B_i\Big\|^{\frac{1}{2}} \Big\},
    \] where the infimum is taken over all families of operators $(A_i)_i, (B_i)_i$ such that $\Phi(\cdot) = \sum_i A_i(\cdot)B_i$ as above.
\end{theorem}

\begin{theorem}[Stinespring dilation Theorem]\label{Stinespring}
Let $\Phi:S_1(\h)\rightarrow S_1(K)$ be a completely positive and trace preserving, where $\h$ and $K$ are finite dimensional Hilbert spaces. Then, there exists a Hilbert space $W$ and a linear isometry $V:\h\rightarrow K\otimes W$ such that $\Phi(\rho)=\tr_{W}(V\rho V^\dag)$ for every $\rho \in S_1(\h)$. By dualizing the previous statement, we obtain that $\Phi^{\dag}(\eta)=V^\dag (\eta\otimes 1_{S_\infty(W)})V$ for every $\eta\in S_\infty(K)$.
\end{theorem}

\begin{remark}\label{Rem: Stines_general}
Note that, given the isometry $V:\h\rightarrow K\otimes W$ in Theorem $\ref{Stinespring}$ one can obtain a contraction $U:H\otimes W\to K\otimes W$ defining, by linearity, $U(|h\rangle\otimes |0\rangle)=V(|h\rangle)$ and $U(|h\rangle\otimes |i\rangle)=0$ for every $i\neq 0$, where here $(|i\rangle)_{i=0}^m$ is an orthonormal basis of $W$. In this case, we have that  $\Phi(\rho)=\tr_{W}\big(U(\rho\otimes |0\rangle\langle 0|) U^\dag\big)$  for every $\rho \in S_1(\h)$. Indeed, the operator $U$ defined above is a partial isometry. More generally, if $U$ is allowed to be an arbitrary contraction, then linear maps $\Phi$ admitting such a representation are precisely the completely positive and completely contractive ones.
\end{remark}

\section{no-signalling quantum strategies and symmetrized one-way communication}\label{sec: ns strategies and symmetrized one-way}
In order to assist exposition, when introducing particular quantum strategies we will motivate with its classical counterpart when possible.
\begin{definition}
    A tuple of nonnegative real numbers $\{P(a,b|x,y)\}_{x,y,a,b}^{I_A,I_B,O_A,O_B}$ is called a \emph{correlation or strategy} if $\sum_{a,b} P(a,b|x,y) = 1$ holds for every pair of inputs $x \in I_A, y \in I_B$. A \emph{quantum strategy} is a completely positive and trace preserving (cptp) map $\varphi: S_1(I_AI_B) \to S_1(O_AO_B)$.
\end{definition}

\begin{remark}
Note that in the first part of the previous definition, $I_A, I_B, O_A, O_B$ are sets over which the indices $x, y, a, b$ range, respectively. In contrast, in the definition of a quantum strategy, $I_A, I_B, O_A, O_B$ denote Hilbert spaces. Moreover, we write $I_A I_B = I_A \otimes I_B$, $O_A O_B = O_A \otimes O_B$. This contraction in the tensor notation will remain constant throughout this work. To simplify, for Hilbert spaces $H$ and $K$, we will simply write $HK$. Furthermore, this notation extends to vectors, so that we usually write $|\xi\rangle \otimes |\eta\rangle = |\xi\rangle |\eta\rangle$ or, even more compactly, $|\xi\eta\rangle$ when we need to shorten the notation due to long formulas. In particular, when working with elements of the canonical basis, we write $|i\rangle \otimes |j\rangle = |i\rangle |j\rangle = |ij\rangle$.
\end{remark}

Note that $\varphi$ is a quantum strategy if and only if $C_\varphi$ is positive semidefinite and
\[
\tr_{O_A O_B} C_\varphi = 1_{S_\infty(I_A I_B)}.
\]
In order to avoid cumbersome notation, at times when it is clear from context we will omit the index sets for player's questions and answers. 

\begin{definition}
    Given a strategy $\{P(a,b|x,y)\}_{x,y,a,b}$ we will call $P$ \emph{no-signalling} if there exist conditional probability distributions $\{P_1(a|x)\}_{x,a}, \{P_2(b|y)\}_{y,b}$ such that \begin{align}
        P_1(a|x) &= \sum_b P(a,b|x,y),\quad \text{for every} \,\, x,y,a \,\,\text{and} \label{eq: NS B not to A}\\
        P_2(b|y) &= \sum_a P(a,b|x,y),\quad \text{for every} \,\, x,y,b. \label{eq: NS A not to B}
    \end{align}
\end{definition}
We will denote the set of strategies satisfying Equation~\eqref{eq: NS B not to A} by $\mathcal{NOS}_{B \not\to A}$, and we will let $\mathcal{NOS}_{A \not\to B}$ denote the set of strategies satisfying Equation~\eqref{eq: NS A not to B}. If a strategy $P$ is in $\mathcal{NOS}_{B \not\to A} \cap \mathcal{NOS}_{A \not\to B}$ then we will call it \emph{no-signalling}, and the set of all no-signalling strategies will be denoted by $\mathcal{NOS}$.

The following definition was first considered in \cite{duan2015nons}.
\begin{definition}
    A quantum strategy $\varphi: S_1(I_AI_B) \to S_1(O_AO_B)$ is called \emph{no-signalling} if it satisfies the following two conditions: \begin{align}
\tr_{S_1(O_A)} \vp(\rho \otimes \s) &= \tr_{S_1(O_A)}\vp (\rho^\prime \otimes \s),\quad \text{for every}\,\,\rho,\rho^\prime \in S_1(I_A), \s \in S_1(I_B) \label{eq: QNS A not to B} \\ \tr_{S_1(O_B)} \vp(\rho \otimes \s) &= \tr_{S_1(O_B)} \vp(\rho \otimes \s^\prime), \quad \text{for every}\,\,\rho \in S_1(I_A), \s,\s^\prime \in S_1(I_B). \label{eq: QNS B not to A}
    \end{align}
\end{definition} For consistency with the notation introduced earlier, we will denote the set of quantum strategies satisfying Equation~\eqref{eq: QNS A not to B} by $\mathcal{NOS}_{A \not\to B}$ and set of quantum strategies satisfying Equation~\eqref{eq: QNS B not to A} by $\mathcal{NOS}_{B \not\to A}$. A quantum strategy $\vp$ satisfying both Equation~\eqref{eq: QNS A not to B} and Equation~\eqref{eq: QNS B not to A} will be called \emph{(quantum) no-signalling} and the set of all no-signalling quantum strategies will be denoted $\mathcal{NOS}$. 

Our main goal in this section is to show that quantum no-signalling strategies are precisely those which arise from allowing Alice to send one-way quantum information to Bob, but not the other way around, and also from allowing Bob to send one-way quantum information to Alice, but not the other way around. In particular, we prove that the set of quantum no-signalling strategies coincide with the so-called \emph{symmetrized one-way} quantum strategies. This result was already proved in  \cite{piani2006properties,Eggeling2002semicasual}. We add here a new proof for completeness.

Consider a cptp map $\Phi: S_1(I_AI_B) \to S_1(O_AO_B)$ and let $\Phi^\dag: S_\infty(O_AO_B) \to S_\infty(I_AI_B)$ denote its adjoint, which is necessarily unital completely positive (ucp). We consider the following two conditions: \begin{align}
    &\Phi^\dag(1_{S_\infty(O_A)} \otimes S_\infty(O_B)) \subset 1_{S_\infty(I_A)} \otimes S_\infty(I_B), \label{eq: alternate QNS A not to B} \\
    &\Phi^\dag(S_\infty(O_A) \otimes 1_{S_\infty(O_B)}) \subset S_\infty(I_A) \otimes 1_{S_\infty(I_B)}. \label{eq: alternate QNS B not to A}
\end{align}

We will prove that a quantum strategy $\Phi: S_1(I_AI_B) \to S_1(O_AO_B)$ satisfies Equation~\eqref{eq: QNS A not to B} if and only if its adjoint $\Phi^\dag: S_\infty(O_AO_B) \to S_\infty(I_AI_B)$ satisfies Equation~\eqref{eq: alternate QNS A not to B}. Similarly, $\Phi$ satisfies Equation~\eqref{eq: QNS B not to A} if and only if $\Phi^\dag$ satisfies Equation~\eqref{eq: alternate QNS B not to A}. First, a useful lemma: 

\begin{lemma}\label{lem: vanishing lemma}
    Consider a matrix $Z \in M_n \otimes M_m$ such that $\tr(Z(A \otimes B)) = 0$ for all hermitian $A \in M_n, B \in M_m,$ where $\tr(A) = 0$. Then there exists $X \in M_m$ such that $Z = 1_{M_n} \otimes X$.
\end{lemma}

\begin{proof}
    Begin by choosing orthonormal bases $(A_i)_i \subset (M_n)_{s.a}, (B_j)_j \subset (M_m)_{s.a}$ such that $A_1 = 1_{M_n}, B_1 = 1_{M_m}$ see for instance \cite{devicente2006separability} for its existence. Here, the orthogonality is with respect to the inner product \[
\sp{x}{y} = \tr(xy).
    \] It must follow that for $i \neq 1 \neq j$, we have $\tr(A_i) = \tr(B_j) = 0$. In particular, we write $Z \in M_n \otimes M_m$ as $Z = \sum_{ij} Z_{ij}(A_i \otimes B_j), Z_{ij} = \sp{Z}{A_i \otimes B_j}$, where $Z_{i,j}$ are some complex numbers.
Suppose there exists a pair $(i,j)$ such that $Z_{ij} \neq 0, i \neq 1$. Then it follows $Z_{ij} = \sp{Z}{A_i \otimes B_j} \neq 0$, which contradicts our initial assumptions. Therefore, it must follow \[
Z = \sum_{j} Z_{1j}(A_1 \otimes B_j) = 1_{M_n} \otimes \Big(\sum_j Z_{1j}B_j\Big),
\] which finishes the proof. \qedhere
\end{proof}

We are now ready to prove the equivalence previously mentioned. 
\begin{prop}\label{prop: equivalence of QNS B not to A}
    Given a quantum channel $\Phi: S_1(I_AI_B) \to S_1(O_AO_B)$ with adjoint $\Phi^\dag: S_\infty(O_AO_B) \to S_\infty(I_AI_B)$, let $C_{\Phi^\dag} \in S_1(O_AO_B) \otimes S_\infty(I_AI_B)$ denote the Choi matrix of $\Phi^\dag$. The following conditions are equivalent: \begin{enumerate}
        \item $\tr_{S_1(O_B)}\Phi(\rho \otimes \s) = \tr_{S_1(O_B)} \Phi(\rho \otimes \s^\prime)$ for every state $\rho \in S_1(I_A)$ and $\s, \s^\prime \in S_1(I_B)$. 
        \item $\Phi^\dag(S_\infty(O_A) \otimes 1_{S_\infty(O_B)}) \subset S_\infty(I_A) \otimes 1_{S_\infty(I_B)}$.
        \item $\tr_{S_1(O_B)} C_{\Phi^\dag} = \wh{P} \otimes 1_{S_\infty(I_B)}$ for some positive semidefinite matrix $\wh{P} \in S_1(O_A) \otimes S_\infty(I_A)$ such that $\tr_{S_1(O_A)}\wh{P} = 1_{S_\infty(I_A)}$.
    \end{enumerate} The analogous result via interchanging the roles of Alice and Bob also holds. 
\end{prop}

Recall that, according to Remark \ref{rem:Choi_adjoint} the previous statement can also be expressed using  $C_{\Phi}$ instead of $C_{\Phi^\dag}$.

\begin{proof}
    We will prove $(1) \implies (2) \implies (3) \implies (1)$. \newline Assume $\tr_{S_1(O_B)} \Phi(\rho \otimes \s) = \tr_{S_1(O_B)} \Phi(\rho \otimes \s^\prime)$ for every state $\rho \in S_1(I_A)$ and $\s,\s^\prime \in S_1(I_B)$. This condition is equivalent to $\Phi$ satisfying \[
\tr_{S_1(O_B)} \Phi(a \otimes b) = 0,
    \] for every hermitian $b \in S_1(I_B)$ of trace zero, and $a \in S_1(I_A)$. Given $x \in S_\infty(O_A)$ we will prove \[\Phi^\dag(x \otimes 1_{S_\infty(O_B)}) = y \otimes 1_{S_\infty(I_B)},\] for some $y \in S_\infty(I_A)$. Via Lemma~\ref{lem: vanishing lemma} it suffices to prove \[
\tr_{S_\infty(I_AI_B)}(\Phi^\dag(x \otimes 1_{S_\infty(O_B)})\cdot(a \otimes b)) = 0,
    \]for all hermitian $a \in S_1(I_A), b \in S_1(I_B), \tr(b) = 0$. This is immediate by observing \begin{align*}
        \tr_{S_\infty(I_AI_B)}(\Phi^\dag(x \otimes 1_{S_\infty(O_B)})\cdot(a \otimes b)) &= \mp{\Phi^\dag(x \otimes 1_{S_\infty(O_B)})}{a \otimes b} \\&= \mp{x \otimes 1_{S_\infty(O_B)}}{\Phi(a \otimes b)} \\&= \tr_{S_1(O_AO_B)}((x \otimes 1_{S_\infty(O_B)})\cdot \Phi(a \otimes b)) \\
        &= \tr_{S_1(O_A)}(x \cdot \tr_{S_1(O_B)}(\Phi(a \otimes b))) \\ &= 0.
    \end{align*} Note that the last equality follows since $\tr_{S_1(O_B)} \Phi(a \otimes b)=0$, for every hermitian $b \in S_1(O_B)$ of trace zero, and $a \in S_1(I_A)$. An application of Lemma~\ref{lem: vanishing lemma} finishes $(1) \implies (2)$. \newline 
    Suppose now $\Phi^\dag(S_\infty(O_A) \otimes 1_{S_\infty(O_B)}) \subset S_\infty(I_A) \otimes 1_{S_\infty(I_B)}$. This implies  \[
\tr(C_{\Phi^\dag} \cdot (x \otimes 1_{S_\infty(O_B)} \otimes a \otimes b)) = 0,
    \] for every $x \in S_\infty(O_A)$, and self-adjoint $a \in S_\infty(I_A), b \in S_\infty(I_B), \tr(b) = 0$. We rewrite this as \[
\tr (\tr_{O_B} C_{\Phi^\dag} \cdot (x \otimes a \otimes b)) = 0.
    \]
Thus it follows for hermitian operators $w \in S_\infty(O_AI_A)$ and $b \in S_1(I_B), \tr(b) = 0$ we obtain \[
\tr (\tr_{O_B} C_{\Phi^\dag} \cdot (w \otimes b)) = 0,
\] which is equivalent to saying $\tr_{O_B} C_{\Phi^\dag} = \wh{P} \otimes 1_{S_\infty(I_B)}$ for some $\wh{P} \in S_1(O_A) \otimes S_\infty(I_A)$. Finally, we observe that completely positivity of $\Phi^\dag$ implies $C_{\Phi^\dag}$ is positive , and in particular $\tr_{S_\infty(O_B)} C_{\Phi^\dag}$ is positive . Along with the fact that $\tr_{O_AO_B}C_{\Phi^\dag} = 1_{S_1(I_AI_B)}$, implies $\wh{P}$ has the desired properties. This finishes the implication $(2) \implies (3)$. \newline 
Our final implication is $(3) \implies (1)$. Thus, assume $\tr_{S_1(O_B)} C_{\Phi^\dag} = \wh{P} \otimes 1_{S_\infty(I_B)}$ for some positive  $\wh{P} \in S_1(O_A) \otimes S_\infty(I_A)$, which satisfies $\tr_{S_1(O_A)} \wh{P} = 1_{S_\infty(I_A)}$. It follows that for every $x \in S_\infty(O_A)$, and self-adjoint $a \in S_1(I_A), b \in S_1(I_B)$, with $\tr(b) = 0$, we have \[
0 = \tr (\tr_{S_\infty(O_B)} C_{\Phi^\dag} \cdot (x \otimes a \otimes b)) = \tr(C_{\Phi^\dag} \cdot (x \otimes 1_{S_\infty(O_B)} \otimes a \otimes b)).
\] We now observe that $\tr(C_{\Phi^\dag} \cdot (x \otimes 1_{S_\infty(O_B)} \otimes a \otimes b))$ equals
\begin{align*}
 &= \tr_{S_\infty(I_AI_B)}((a \otimes b)\cdot \tr_{S_\infty(O_AO_B)}(C_{\Phi^\dag} \cdot(x \otimes 1_{S_\infty(O_B)} \otimes 1_{S_1(I_A)} \otimes 1_{S_1(I_B)}))) 
\\ &= \tr_{S_\infty(I_AI_B)} ((a \otimes b)\cdot \Phi^\dag(x^T \otimes 1_{S_\infty(O_B)}))
\\ &= \tr_{S_1(O_AO_B)}(\Phi(a \otimes b) \cdot (x^T \otimes 1_{S_1(O_B)})) \\ &= \tr_{S_1(O_A)}(\tr_{S_1(O_B)}\Phi(a \otimes b) \cdot x^T).
\end{align*} Thus we have deduced that for every $x \in S_\infty(O_A)$ and self-adjoint $a \in S_1(I_A), b \in S_1(I_B), \tr(b) = 0$, that \[
\tr_{S_1(O_A)}(\tr_{S_1(O_B)}\Phi(a \otimes b) \cdot x^T) = 0.
\] This is equivalent to $\tr_{S_1(O_B)}\Phi(a \otimes b) = 0$ for all self-adjoint $a \in S_1(I_A), b \in S_1(I_B), \tr(b) = 0$. Thus we have shown $(3) \implies (1)$, which finishes the proof. \qedhere 
\end{proof}
The equivalent conditions given in Proposition~\ref{prop: equivalence of QNS B not to A} describe the no-signalling condition from Bob to Alice, i.e., the conditions express that Bob cannot send quantum information to Alice through the channel. For this reason we have adopted the notation $\mathcal{NOS}_{B \not\to A}(I_{AB}, O_{AB})$, to denote the set of channels $\Phi: S_1(I_AI_B) \to S_1(O_AO_B)$ satisfying any of these conditions. As observed in \cite{piani2006properties,Eggeling2002semicasual}, such a quantum channel $\Phi: S_1(I_AI_B) \to S_1(O_AO_B)$ can be implemented by allowing Alice to send quantum information to Bob, but not the other way around. We now define quantum one-way channels and provide a more detailed explanation of the previous sentence.
\begin{definition}\label{defn: quantum one-way}
    A quantum channel $\Phi: S_1(I_AI_B) \to S_1(O_AO_B)$ is called \emph{quantum one-way from Alice to Bob}, if there exists a Hilbert space $K$ (the register to be sent from Alice to Bob) and quantum channels \[
\Phi_A: S_1(I_A) \to S_1(O_AK), \Phi_B: S_1(KI_B) \to S_1(O_B),
\] such that \begin{equation} 
\Phi = (\id_{S_1(O_A)} \otimes \Phi_B) \circ (\id_{S_1(I_B)} \otimes \Phi_A). \label{eq: QOW A to B}
\end{equation} Denote the set of channels of the form \eqref{eq: QOW A to B} by $\mathcal{OW}_{A \to B}(I_{AB},O_{AB})$.
\end{definition}

\begin{theorem}\label{thm: QNS equals QOW}
    Given a quantum channel $\Phi: S_1(I_AI_B) \to S_1(O_AO_B)$, then $\Phi \in \mathcal{NOS}_{B \not\to A}(I_{AB},O_{AB})$ if and only $\Phi \in \mathcal{OW}_{A \to B}(I_{AB}, O_{AB})$. Similarly, $\Phi \in \mathcal{NOS}_{A \not\to B}(I_{AB}, O_{AB})$ if and only if $\Phi \in \mathcal{OW}_{B \to A}(I_{AB}, O_{AB})$. 
\end{theorem}

\begin{proof}
    By Proposition~\ref{prop: equivalence of QNS B not to A} it follows $\Phi \in \mathcal{NOS}_{B \not\to A}(I_{AB}, O_{AB})$ if and only if \[
\Phi^\dag(S_\infty(O_A) \otimes 1_{S_\infty(O_B)}) \subset S_\infty(I_A) \otimes 1_{S_\infty(I_B)}.
    \] We first assume $\Phi \in \mathcal{OW}_{A \to B}(I_{AB}, O_{AB})$ and will show $\Phi^\dag$ satisfies the above equivalent condition. By our assumption, let $\Phi_A: S_1(I_A) \to S_1(O_AK), \Phi_B: S_1(KI_B) \to S_1(O_B)$, be two quantum channels, $K$ a Hilbert space, such that 
\[
\Phi = (\id_{S_1(O_A)} \otimes \Phi_B)\circ (\id_{S_1(I_B)} \otimes \Phi_A), \Phi^\dag = (\Phi_A^\dag \otimes \id_{S_\infty(I_B)}) \circ (\id_{S_\infty(O_A)} \otimes \Phi_B^\dag).
\] Given $x \in S_\infty(O_A)$ we observe \begin{align*}
    \Phi^\dag(x \otimes 1_{S_\infty(O_B)}) &= (\Phi_A^\dag \otimes \id_{S_\infty(I_B)}) \circ (\id_{S_\infty(O_A)} \otimes \Phi_B^\dag)(x \otimes 1_{S_\infty(O_B)}) \\&= (\Phi_A^\dag \otimes \id_{S_\infty(I_B)})(x \otimes 1_{S_\infty(K)} \otimes 1_{S_\infty(I_B)}) \\&= \Phi_A^\dag(x \otimes 1_{S_\infty(K)}) \otimes 1_{S_\infty(I_B)}.
\end{align*} This proves $\Phi \in \mathcal{NOS}_{B \not\to A}(I_{AB}, O_{AB})$. \newline Suppose now $\Phi \in \mathcal{NOS}_{B \not\to A}(I_{AB},O_{AB})$. This in turn implies we may define a completely positive and unital map $\Phi_A^\dag: S_\infty(O_A) \to S_\infty(I_A)$ such that for any $x \in S_\infty(O_A)$ one has \begin{equation} 
\Phi^\dag(x \otimes 1_{S_\infty(O_B)}) = \Phi_A^\dag(x) \otimes S_\infty(I_B).\label{eq: qow assumption}
\end{equation} We now perform Stinespring dilations (see Theorem \ref{Stinespring}) to both $\Phi$ and $\Phi_A$. To this end, we have linear isometries \[
V_{AB}: I_AI_B \to O_AO_BH_{AB}, V_A: I_A \to O_AH_A,
\] such that \[
\Phi(\rho) = \tr_{H_{AB}}(V_{AB}\rho V_{AB}^\dag), \Phi_A(\s) = \tr_{H_A}(V_A \s V_A^\dag), \rho \in S_1(I_AI_B), \s \in S_1(I_A).
\] Begin by writing $V_{AB}(\ket{x}\ket{y}) = \sum_{ab}\ket{a}\ket{b}\ket{\eta_{xy}^{ab}}, \ket{\eta_{xy}^{ab}} \in H_{AB}$ for every $x,y,a,b$. Similarly, write $V_A(\ket{x}) = \sum_a \ket{a}\ket{\xi_x^a}, \ket{\xi_x^a} \in H_A$ for every $x,a$. It then follows for every $a,a^\prime \in A$, that \[
\Phi^\dag(\op{a}{a^\prime} \otimes 1_{S_\infty(O_B)}) = \sum_{x,x^\prime,y,y^\prime,b} \ip{\eta_{x,y}^{a,b}}{\eta_{x^\prime,y^\prime}^{a^\prime,b}} \op{xy}{x^\prime y^\prime},  
\] and \[
\Phi_A^\dag(\op{a}{a^\prime}) = \sum_{x,x^\prime} \ip{\xi_x^a}{\xi_{x^\prime}^{a^\prime}} \op{x}{x^\prime}.
\] It follows via Equation~\eqref{eq: qow assumption} that \[
\sum_b \ip{\eta_{x,y}^{a,b}}{\eta_{x^\prime,y^\prime}^{a^\prime,b^\prime}} = \d_{y,y^\prime} \ip{\xi_x^a}{\xi_{x^\prime}^{a^\prime}},
\] for every $x,x^\prime,y,y^\prime,a,a^\prime$.
Define the Hilbert space $K:= \Span\{\ket{\xi_x^a}: x,a\} \subset H_A$, and define the mapping \[
\theta_B: I_BK \to H_{AB}O_B, \ket{y}\ket{\xi_x^a} \mapsto \sum_b \ket{\eta_{x,y}^{a,b}}\ket{b}, 
\]for every $x,y,a$, and extend via linearity. We first observe that $\theta_B$ is a well defined isometry. This follows since for every $(\l_{x,y}^a)_{x,y,a} \subset \mathbb C$, we have\begin{align*}
    \Big\|\theta_B(\sum_{x,y,a} \l_{x,y}^a \ket{y}\ket{\xi_x^a})\Big\|_{H_{AB}O_B}^2 &= \sum_{x,x^\prime,y,y^\prime,a,a^\prime,b,b^\prime} \oline{\l_{x,y}^a}\l_{x^\prime,y^\prime}^{a^\prime}\ip{\eta_{x,y}^{a,b}}{\eta_{x^\prime,y^\prime}^{a^\prime,b^\prime}} \ip{b}{b^\prime} 
    \\ &= \sum_{x,x^\prime,y,y^\prime,a,a^\prime} \sum_b \oline{\l_{x,y}^a}\l_{x^\prime,y^\prime}^{a^\prime} \ip{\eta_{x,y}^{a,b}}{\eta_{x^\prime,y^\prime}^{a^\prime,b^\prime}} \\ &= \sum_{x,x^\prime,y,a,a^\prime} \oline{\l_{x,y}^a}\l_{x^\prime,y}^{a^\prime} \ip{\xi_x^a}{\xi_{x^\prime}^{a^\prime}} \\&= \Big\| \sum_{x,y,a} \l_{x,y}^a \ket{y}\ket{\xi_x^a}\Big\|_{I_BK}^2.
\end{align*} This proves $\theta_B$ is an isometry. Furthermore, we immediately obtain \[
V_{AB} = (\id_{O_A} \otimes \theta_B) \circ(V_A \otimes \id_{I_B}).
\] To finish the proof, we consider the quantum channels \[
\tilde{\Phi}_A: S_1(I_A) \to S_1(O_AK), \sigma \mapsto V_A \sigma V_A^\dag, \sigma \in S_1(I_A),
\] and \[
\tilde{\Phi}_B: S_1(KI_B) \to S_1(O_B), \rho \mapsto \tr_{H_{AB}} \theta_B \rho \theta_B^\dag, \rho \in S_1(KI_B).
\] We then have \[
\Phi = (\id_{S_1(O_A)} \otimes \tilde{\Phi}_B) \circ (\tilde{\Phi}_A \otimes \id_{S_1(I_B)}),
\] which proves $\Phi \in \mathcal{OW}_{A \to B}(I_{AB}, O_{AB})$. \newline The second claim of the theorem follows in the same manner. This finishes the proof. \qedhere 
\end{proof}

Define the following sets \begin{align*}
    \mathcal{NOS}(I_{AB}, O_{AB}) & =\mathcal{NOS}_{A \not\to B}(I_{AB}, O_{AB}) \cap \mathcal{NOS}_{B \not\to A}(I_{AB}, O_{AB}), 
   \\ \mathcal{OW}_{sym}(I_{AB},O_{AB}) &= \mathcal{OW}_{A \to B}(I_{AB},O_{AB}) \cap \mathcal{OW}_{B \to A}(I_{AB},O_{AB}),  
\end{align*} as the set of \emph{no-signalling quantum channels} and the set of \emph{symmetrized one-way quantum channels}, i.e., those channels which can be implemented simultaneously by allowing Alice to send one-way quantum information to Bob (but not the other way around) and by allowing Bob to send one-way quantum information to Alice (but not the other way around). A consequence of Theorem~\ref{thm: QNS equals QOW} is the following corollary: 

\begin{corollary}
    It follows \[
\mathcal{NOS}(I_{AB}, O_{AB}) = \mathcal{OW}_{sym}(I_{AB}, O_{AB}).
    \]
\end{corollary}

\section{Games and Tensors}\label{Sec: Games and Tensors}
The main theme for this section will be to demonstrate how one may calculate various values of games (entangled, one-way, sub-one-way) via calculating the corresponding game tensor $\xi_\mathcal G$ with respect to different operator space tensor norms. For readability and completeness, we will begin with calculations in the one-player scenario. This will streamline and assist in the more complex two-player tensor calculations. 

\subsection{One-player} \label{subsec: one-player tensors}
A \emph{one-player quantum game} is given by the pair $\mathcal  G(\psi, P)$, where $\ket{\psi} \in RI$ is a state in the bipartite Hilbertian tensor product of the referee's Hilbert space $R$, and Alice's Hilbert space $I$, and $P \in S_\infty(RO)$ is a projection on the Hilbertian tensor product of the referee's Hilbert space and Alice's output space $O$. As in the previous section, we let $I,O,R$, etc. also denote the corresponding Hilbert spaces of the various registers.
In order to simplify notation and making a slight abuse of language, we will write $|\h|$ and $|A|$ to denote the dimension of the Hilbert space $\h$ and the cardinality of the set $A$ respectively.

Begin by writing the projection $P$ as $P = \sum_{t \in D} \op{\g_t}{\g_t}$, where $(\ket{\g_t})_{t \in D}$ is an orthonormal system for $RI$, and thus $\abs{D} \leq |R||I|$.

Fixing an orthonormal basis $(\ket{i})_{i}$ for $R$, we may write \[
\ket{\g_t} = \sum_{i=1}^{|R|} \ket{i}\ket{\g_t^i}, \,\,\text{and}\,\, \ket{\psi} = \sum_{i=1}^{|R|} \ket{i}\ket{\psi_i},
\] where $\ket{\g_t^i} \in O$ and $\ket{\psi_i} \in I$ for every $i=1,\cdots , |R|$ and $t \in D$. We then write the projection as \[
P = \sum_{t \in D} \op{\g_t}{\g_t} = \sum_{t,i,j} \op{i}{j} \otimes \op{\g_t^i}{\g_t^j} = \sum_{ij} \op{i}{j} \otimes x_i^*x_j \in S_\infty(R) \otimes S_\infty(O),
\] where for each $i=1,\cdots , |R|$, $x_i = \sum_t \op{t}{\g_t^i} \in S_\infty(O,D)$.  We also write \[
\psi:= \op{\psi}{\psi} = \sum_{i,j} \op{i}{j} \otimes \op{\psi_i}{\psi_j} \in S_1(R) \otimes S_1(I).
\]
Taking the transpose over the output Hilbert space $O$, we denote \[
P^{T_O} = \sum_{i,j} \op{i}{j} \otimes (x_i^*x_j)^{T} = \sum_{i,j} \op{i}{j} \otimes x_j^T\oline{x}_i = \sum_{i,j} \op{i}{j} \otimes y_j^*y_i \in S_\infty(R) \otimes S_\infty(O),
\] where for each $i=1,\cdots , |R|$, $y_i:= \oline{x}_i = \sum_t \op{t}{\oline{\g}_t^i} \in S_\infty(O,D)$.
\begin{definition}\label{def: full tensor, one-player}
    Given a one-player quantum game $\mathcal  G(\ket{\psi}, P)$ we define the \emph{full game tensor} $\xi_{\mathcal  G}$ as \[
\xi_{\mathcal  G} = \tr_R\left(\left(\psi \otimes 1_{S_\infty(O)} \right)\left(P^{T_O} \otimes 1_{S_1(I)} \right)  \right) = \sum_{i,j} \op{\psi_i}{\psi_j} \otimes y_i^*y_j \in S_1(I) \otimes S_\infty(O).
    \]
\end{definition}

Given a one-player quantum game $\mathcal  G$, then consider the following decomposition of the full tensor $\xi_{\mathcal  G} = \eta_{\mathcal  G}^* \eta_{\mathcal  G}$, where \[
\eta_{\mathcal  G} = \sum_{t,j} \op{1}{\psi_j} \otimes \op{t}{\oline{\g}_t^j} \in S_1(I)\otimes S_\infty(O,D)=S_1(I, S_\infty(O,D)) = R_I \otimes_h C_D \otimes_h R_O \otimes_h C_I.
\]
In particular, we consider the tensor 
\begin{align}\label{Eq._Half_Tensor}
\eta_{\mathcal  G} := \sum_{t,j} \ket{t} \otimes \bra{\oline{\g}_t^j} \otimes \ket{\oline{\psi}_j} \in C_D \otimes R_O \otimes C_I.
\end{align} Recall that the tensor $\bra{\oline{\g}_t^j} \otimes \ket{\oline{\psi}_j} \in R_O \otimes C_I$ is identified with the operator $\op{\g_t^j}{\psi_j} \in S_1(I,O)$ so that, given $S \in S_\infty(I,O)$, we have \[
\tr(S (\op{\g_t^j}{\psi_j})^T) = \bra{\oline{\g}_t^j}S\ket{\oline{\psi}_j}.
\]

\subsection{The honest value of a game}
Given a projection $P \in S_\infty(RO)$ and a state $\ket{\psi} \in RI$, where $\psi:= \op{\psi}{\psi} \in S_1(RI)$, then the \emph{honest value} of the quantum game $\mathcal  G(\ket{\psi}, P)$ is defined as \begin{align}
    \om(\mathcal  G):= \sup \tr(P \cdot (\id_{S_1(R)} \otimes \vp)(\psi)),
\end{align} where the supremum is taken over all strategies (cptp maps) $\vp: S_1(I) \to S_1(O)$. Our first lemma (whose extensions we will also prove in later sections) is to prove that we may reduce to complete contractions $\vp: S_1(I) \to S_1(O)$ when calculating the honest value. 

\begin{lemma}\label{lem: replace cc with cptp}
    Given a one-player quantum game $\mathcal  G(\ket{\psi}, P)$, let \[
\tilde{\om}(\mathcal  G):= \sup \tr(P \cdot (\id_{S_1(R)} \otimes \vp)(\psi)),
    \] where the supremum is taken over all complete contractions $\vp: S_1(I) \to S_1(O)$. Then $\tilde{\om}(\mathcal  G) =\om(\mathcal  G)$. 
\end{lemma}

\begin{proof}
Since every completely positive trace preserving map is a complete contraction it is immediate that $\om(\mathcal  G) \leq \tilde{\om}(\mathcal  G)$, so we only need to show $\tilde{\om}(\mathcal  G) \leq \om(\mathcal G)$. We begin with a complete contraction $\vp: S_1(I) \to S_1(O)$, and let $\vp^\dag: S_\infty(O) \to S_\infty(I)$ denote its adjoint which is also a complete contraction. Employing Theorem~\ref{thm: cb decomposition as sum} we may write $\vp^\dag(\cdot) = \sum_i A_i(\cdot)B_i$, where $(A_i)_i \subset S_\infty(O,I), (B_i)_i \subset S_\infty(I,O)$, and \[ \Big\| \sum_i A_iA_i^\dag \Big\|^{\frac{1}{2}} \Big\| \sum_i B_i^\dag B_i \Big\|^{\frac{1}{2}} \leq 1.\] We may assume without loss of generality that $\| \sum_i A_iA_i^\dag \| \leq 1, \| \sum_i B_i^\dag B_i \| \leq 1$.  In calculating the honest value of $\mathcal  G$ we write \begin{align*}
   & \tr(P \cdot(\id_{S_1(R)} \otimes \vp)(\psi)) = \tr([(\id_{S_\infty(R)} \otimes \vp^\dag)(P^\dag)]^\dag \cdot \psi) 
    \\ &= \tr \Big( \Big(\sum_i (\id_R \otimes  B_i^\dag)P(\id_R \otimes A_i^\dag)\Big) \cdot \psi\Big) 
    \\&= \sum_i \tr\Big(\psi^{\frac{1}{2}}(\id_R \otimes  B_i^\dag)PP(\id_R \otimes A_i^\dag)\psi^{\frac{1}{2}}\Big) 
    \\&\leq \Big(\sum_i \tr((\id_R \otimes  B_i^\dag)P(\id_R \otimes B_i)\psi)\Big)^{\frac{1}{2}}\Big(\sum_i \tr((\id_R \otimes A_i)P(\id_R\otimes A_i^\dag)\psi)\Big)^{\frac{1}{2}},
\end{align*} where we have applied the Cauchy-Schwarz inequality. Note we have also used the fact that $P$ is self-adjoint. Define the completely positive maps $\Phi_1, \Phi_2: S_\infty(O) \to S_\infty(I)$ via \[
\Phi_1(x) = \sum_i B_i^\dag x B_i,\hspace{0.3 cm} \Phi_2(x) = \sum_i A_i x A_i^\dag.
\] Continuing with our calculation we see \begin{align*}
     &\Big(\sum_i \tr[(\id_R \otimes  B_i^\dag)P(\id_R \otimes B_i)\psi]\Big)^{\frac{1}{2}}\Big(\sum_i \tr[(\id_R \otimes A_i)P(\id_R \otimes A_i^\dag)\psi]\Big)^{\frac{1}{2}} 
     \\&= \Big(\tr( (\id_{S_1(R)} \otimes \Phi_1 )(P))\Big)]^{\frac{1}{2}} \Big( \tr((\id_{S_1(R)} \otimes \Phi_2 )(P)) \Big)^{\frac{1}{2}}.
\end{align*} 

The normalization condition on the operators $A_i, B_i$ guarantees that the superoperators $\Phi_1, \Phi_2$ are complete contractions. Thus, we may assume $\vp$ is completely positive and completely contractive in the definition of $\tilde{\om}(\mathcal  G)$. \newline To conclude, given any completely positive complete contraction $S: S_\infty(O) \to S_\infty(I)$, it follows that $X:= 1_{S_\infty(I)} - S(1_{S_\infty(O)})$ is positive . Then, we can consider  the unital completely positive map \[
\tilde{S}: S_\infty(O) \to S_\infty(I), a \mapsto S(a) + \frac{\tr(a)}{|O|}X,
\] for every $a \in S_\infty(O)$. It follows \[
\tr( ( \id_{S_\infty(R)} \otimes S)(P) \cdot \psi) \leq \tr((\id_{S_\infty(R)} \otimes \tilde{S} )(P) \cdot \psi), 
\] where we have used the fact that both $P$ and $\psi$ are positive. This proves that we may take $\vp^\dag: S_\infty(O) \to S_\infty(I)$ to be unital completely positive which implies $\vp: S_1(I) \to S_1(O)$ is completely positive trace preserving, as desired. This proves $\tilde{\om}(\mathcal  G) \leq \om(\mathcal  G)$. \qedhere 
\end{proof}

The following lemma expresses the correspondence between the dual action of strategies on games and the duality in the corresponding Schatten spaces. 
\begin{lemma}\label{lem: duality of game tensors and strategies}
For every linear map $\vp: S_1(I) \to S_1(O)$, letting $C_\vp$ denote the corresponding Choi matrix, we have that \[
\sp{\xi_{\mathcal  G}}{C_\vp} = \tr( P \cdot (\id_{S_1(R)} \otimes \vp )(\psi)).
\]
\end{lemma}

\begin{proof}
    Recall that our full tensor takes the form \[
\xi_{\mathcal  G} = \sum_{i,j} \op{\psi_i}{\psi_j} \otimes y_i^*y_j \in S_1(I) \otimes S_\infty(O). 
    \] It follows \begin{align*}
\sp{\xi_{\mathcal  G}}{C_{\vp}} = \tr(\xi_{\mathcal  G}^T C_{\vp}) &= \tr \Big( \Big(\sum_{i,j} \op{\oline{\psi}_j}{\oline{\psi}_i} \otimes y_j^T\oline{y_i}\Big) \Big( \sum_{i^\prime,j^\prime} \op{i^\prime}{j^\prime} \otimes \vp(\op{i^\prime}{j^\prime})\Big) \Big)
\\&=  \sum_{i,j,i^\prime,j^\prime} \ip{\oline{\psi}_i}{i^\prime}\ip{j^\prime}{\oline{\psi}_j} \tr( y_j^T\oline{y_i}\vp(\op{i^\prime}{j^\prime}))   
\end{align*}

Representing the output of $\vp$ via its Choi matrix, we observe \begin{align*}
    (\id_{S_1(R)} \otimes \vp)(\psi) &= \tr_{S_1(I)} \Big[\Big( 1_{S_1(R)} \otimes \sum_{i,j} \op{i}{j} \otimes \vp(\op{i}{j})\Big) \Big(\sum_{i^\prime,j^\prime} \op{i^\prime}{j^\prime} \otimes {\op{\psi_{i^\prime}}{\psi_{j^\prime}}^T} \otimes 1_{S_1(O)}\big)\Big] 
    \\&= \tr_{S_1(I)} \Big[\Big( 1_{S_1(R)} \otimes \sum_{i,j} \op{i}{j} \otimes \vp(\op{i}{j})\Big)\Big (\sum_{i^\prime,j^\prime} \op{i^\prime}{j^\prime} \otimes \op{\oline{\psi}_{j^\prime}}{\oline{\psi}_{i^\prime}} \otimes 1_{S_1(O)}\Big)\Big] 
    \\&= \sum_{i,j,i^\prime,j^\prime} \ip{j}{\oline{\psi}_{j^\prime}}\ip{\oline{\psi}_{i^\prime}}{i} \op{i^\prime}{j^\prime} \otimes \vp(\op{i}{j}).
\end{align*}
Recalling that our projection takes the form $P = \sum_{\tilde{i},\tilde{j}} \op{\tilde i}{\tilde j} \otimes x_{\tilde i}^* x_{\tilde j} \in S_\infty(R) \otimes S_\infty(O)$, we conclude \begin{align*}
    \tr(P \cdot(\id_{S_1(R)} \otimes \vp)(\psi)) &= \tr\Big[ P \cdot \Big(\sum_{i,j,i^\prime,j^\prime} \ip{j}{\oline{\psi}_{j^\prime}}\ip{\oline{\psi}_{i^\prime}}{i} \op{i^\prime}{j^\prime} \otimes \vp(\op{i}{j}))\Big) \Big]
    \\&= \sum_{i,j,i^\prime,j^\prime,\tilde{i},\tilde{j}} \ip{j}{\oline{\psi}_{j^\prime}}\ip{\oline{\psi}_{i^\prime}}{i} \ip{\tilde{j}}{i^\prime}\ip{j^\prime}{\tilde{i}}\tr(x_{\tilde{i}}^*x_{\tilde{j}}\vp(\op{i}{j}))
    \\&= \sum_{i,j,i^\prime,j^\prime} \ip{j}{\oline{\psi}_{j^\prime}}\ip{\oline{\psi}_{i^\prime}}{i} \tr(x_{j^\prime}^*x_{i^\prime} \vp(\op{i}{j})) 
    \\&=\sum_{i,j,i^\prime,j^\prime} \ip{j}{\oline{\psi}_{j^\prime}}\ip{\oline{\psi}_{i^\prime}}{i} \tr(y_{j^\prime}^T\oline{y_{i^\prime}} \vp(\op{i}{j})).
\end{align*} This concludes the proof. \qedhere 
\end{proof}

Using the above lemma, we obtain the following characterization of $\omega(\mathcal{G})$ in terms of tensor norms.
\begin{theorem}\label{thm: value of game, one-player, full tensor}
    Given any one-player quantum game $\mathcal  G(\ket{\psi},P)$, it follows \[
\om(\mathcal  G) = \norm{\xi_{\mathcal  G}}_{S_1(I,S_\infty(O))}.
    \]
\end{theorem}

\begin{proof}
   By applying Lemma~\ref{lem: replace cc with cptp}, Lemma~\ref{lem: duality of game tensors and strategies}, and the duality (\ref{Eq. Dual_p,q}), we observe that \begin{align*}
        \norm{\xi_{\mathcal  G}}_{S_1(I,S_\infty(O))} &= \sup_{\norm{C_{\vp}}_{S_\infty(I,S_1(O))} \leq 1} \sp{\xi_{\mathcal  G}}{C_{\vp}} 
        \\&= \sup_{\cbnorm{\vp: S_1(I) \to S_1(O)} \leq 1} \tr( P \cdot (\id_{S_1(R)} \otimes \vp)(\psi))
        \\&= \om (\mathcal  G).
    \end{align*}
\end{proof}

Our next result provides a description of the value $\om(\mathcal  G)$ in terms of the half tensor $\eta_{\mathcal  G}$. It is inspired in the work \cite{cooney2015rank}, where a similar statement was proven for rank-one games, i.e., where the game projection $P$ takes the form $P = \op{\g}{\g}$.

\begin{theorem}\label{thm: value of game, one-player, half tensor}
    Given a one-player quantum game $\mathcal  G(\ket{\psi},P)$, we have \[
\om(\mathcal  G) = \norm{\eta_{\mathcal  G}}_{C_D(S_1(I,O))}^2.
    \]
\end{theorem}

\begin{proof}
Let us start by recalling that, by (\ref{Eq: indetity_min_h_1}),$$C_D(S_1(I,O))=C_D \otimes_{min} S_1(I,O)=C_D \otimes_{h} S_1(I,O).$$

We have \[
\om(\mathcal  G) = \sup \tr (P \cdot (\id_{S_1(R)} \otimes \vp)(\psi)),
\] where the supremum is taken over all completely positive and trace preserving maps $\vp: S_1(I) \to S_1(O)$. Applying Theorem \ref{Stinespring} (see also Remark \ref{Rem: Stines_general}), there exists a Hilbert space $E$, a unit vector $\ket{\xi} \in E$ and a contraction $U: IE \to OE$ such that \[
\vp(\rho) = \tr_E \big[U(\rho \otimes \op{\xi}{\xi})U^\dag\big],
\] and, conversely, every such map of this form yields a completely positive complete contraction from $S_1(I) \to S_1(O)$. 
Using Lemma~\ref{lem: replace cc with cptp} we have \begin{align*}
    \om(\mathcal  G) &= \sup_{\norm{U}_{S_\infty(IE,OE)} \leq 1, \ket{\xi} \in E} \tr_{S_1(RO)}( P \cdot \tr_E((\id_{R} \otimes U)(\psi \otimes \op{\xi}{\xi})(\id_R \otimes U^\dag)))
    \\&= \sup_{\norm{U}_{S_\infty(IE,OE)} \leq 1, \ket{\xi} \in E} \tr_{S_1(ROE)} ((P \otimes \id_{S_1(E)}) \cdot (\id_R \otimes U)(\psi \otimes \op{\xi}{\xi})(\id_R \otimes U^\dag)).
\end{align*} Expressing the projection as $P = \sum_{i,j,t} \op{i}{j} \otimes \op{\g_t^i}{\g_t^j},$ the state $\psi$ as $\psi = \sum_{i,j} \op{i}{j} \otimes \op{\psi_i}{\psi_j}$, and tracing out over the Referee Hilbert space $R$, we express the value as \begin{align*}
    \om(\mathcal  G) &= \sup_{\norm{U}_{S_\infty(IE,OE)} \leq 1, \ket{\xi} \in E} \tr_{S_1(OE)} \Big[ \sum_{i,j,t} (\op{\g_t^i}{\g_t^j} \otimes \id_{S_1(E)}) \cdot U(\op{\psi_j\xi}{\psi_i\xi})U^\dag \Big] 
    \\&= \sup_{\norm{U}_{S_\infty(IE,OE)} \leq 1, \ket{\xi} \in E} \sum_t \Big\|\sum_{j} (\bra{\g_t^j} \otimes \id_E)U(\ket{\psi_j \xi})\Big\|_{E}^2.
\end{align*} Let us now consider the quantity $\norm{\eta_{\mathcal  G}}_{C_D(R_O \otimes_h C_I)}=C_D \otimes_{min} S_1(I,O)$. In calculating the min norm we see \begin{align*}
    \norm{\eta_{\mathcal  G}}_{C_D \otimes_{min} S_1(I,O)} &= \sup_{\cbnorm{T: S_1(I,O) \to S_\infty(E)} \leq 1}  \Big\|(\id_{C_D} \otimes T)(\eta_{\mathcal  G})\Big\|_{C_D \otimes_{min} S_\infty(E)} 
    \\&= \sup_{\cbnorm{T: S_1(I,O) \to S_\infty(E)} \leq 1} \Big\|(\id_{C_D} \otimes T)\Big(\sum_{j,t} \ket{t} \otimes \bra{\oline{\g}_t^j} \otimes \ket{\oline{\psi}_j}\Big)\Big\|_{C_D \otimes_{min} S_\infty(E)} 
    \\&=  \sup_{\cbnorm{T: S_1(I,O) \to S_\infty(E)} \leq 1} \Big\|\sum_{j,t} \ket{t} \otimes T(\bra{\oline{\g}_t^j} \otimes \ket{\oline{\psi}_j})\Big\|_{C_D \otimes_{min} S_\infty(E)}.
\end{align*}

For each $t$ set $A_t:= \sum_j T(\bra{\oline{\g}_t^j} \otimes \ket{\oline{\psi}_j})$. Continuing we observe
\begin{align*}
    \norm{\eta_{\mathcal  G}}_{C_D \otimes_{min} S_1(I,O)} &= \sup_{\cbnorm{T: S_1(I,O) \to S_\infty(E)} \leq 1} \Big\|\sum_t \ket{t} \otimes A_t\Big\|_{C_D \otimes_{min} S_\infty(E)} 
    \\&=  \sup_{\cbnorm{T: S_1(I,O) \to S_\infty(E)} \leq 1} \Big\|\sum_t A_t^* A_t\Big\|_{S_\infty(E)}^{\frac{1}{2}} 
    \\ &=  \sup_{\cbnorm{T: S_1(I,O) \to S_\infty(E)} \leq 1, \ket{\xi} \in E} \Big(\sum_t \norm{A_t \ket{\xi}}_E^2 \Big)^{\frac{1}{2}}.
\end{align*}

Given such a complete contraction $T: S_1(I,O) \to S_\infty(E)$, we consider its associated tensor $\wh{T} \in S_\infty(I,O) \otimes S_\infty(E)$, satisfying $\|\wh{T}\|_{S_\infty(I,O) \otimes_{min} S_\infty(E)} \leq 1,$ where the correspondence is given via the completely isometric identification \[CB(S_1(I,O),S_\infty(E))=
S_\infty(I,O) \otimes_{min} S_\infty(E) = S_\infty(IE,OE),
\] where \[
T(\rho) = \tr_{S_1(O)}(\wh{T} \cdot(\rho^T \otimes 1_{S_\infty(E)})).
\] Therefore, taking the supremum over complete contractions $T: S_1(I,O) \to S_\infty(E)$ is equivalent to taking the supremum over contractions $U: IE \to OE$. Given vectors $\ket{\psi} \in I, \ket{\g} \in O, \ket{\xi} \in E$, the identification satisfies \begin{align}
    T(\bra{\oline{\g}} \otimes \ket{\oline{\psi}})\ket{\xi} = (\bra{\oline{\g}} \otimes \id_E)U(\ket{\oline{\psi} \xi}). \label{eq: useful identification, T to V}
\end{align} 

Thus, we observe 
\begin{align*}
    \norm{\eta_{\mathcal  G}}_{C_D \otimes_{min} S_\infty(E)}^2 &=  \sup_{\cbnorm{T: S_1(I,O) \to S_\infty(E)} \leq 1, \ket{\xi} \in E} \sum_t \norm{A_t \ket{\xi}}_E^2 
    \\&= \sup_{\norm{U}_{S_\infty(IE,OE)} \leq 1, \ket{\xi} \in E} \sum_t \Big\|\sum_j (\bra{{\oline{\g}_t^j}} \otimes \id_E)U(\ket{{\oline{\psi_j}} \xi})\Big\|^2 \\
   &=\sup_{\norm{U}_{S_\infty(IE,OE)} \leq 1, \ket{\xi} \in E} \sum_t \Big\|\sum_j (\bra{\oline{\g}_t^j} \otimes \id_E)\oline{U}(\ket{\oline{\psi}_j\oline{\xi}})\Big\|^2\\&=\sup_{\norm{U}_{S_\infty(IE,OE)} \leq 1, \ket{\xi} \in E} \sum_t \Big\|\sum_j (\bra{\g_t^j} \otimes \id_E)U(\ket{\psi_j \xi})\Big\|^2.
\end{align*} 

This shows $\om(\mathcal  G) = \norm{\eta_{\mathcal  G}}_{C_D(S_1(I,O))}^2$ which finishes the proof. 
\qedhere 
\end{proof}

\subsection{Two-player}\label{subsec: two-player tensors}
We now consider values of two-player quantum games. In particular, we focus on the \emph{entangled, one-way} and \emph{sub-one-way} values of $\mathcal  G$. 
Following Subsection~\ref{subsec: one-player tensors}, there are orthonormal systems $\{\ket{\g_t}\}_{t\in D}$ in $RO_AO_B$, and $\{\ket{i}\}$ in  $R$ such that our game projection $P \in S_\infty(RO_AO_B)$ takes the form \[
P = \sum_{ t} \op{\g_t}{\g_t} = \sum_{i,j,t} \op{i}{j} \otimes \op{\g_t^i}{\g_t^j} = \sum_{i,j} \op{i}{j} \otimes x_i^*x_j, x_i = \sum_{t \in D} \op{t}{\g_t^i} \in S_\infty(O_AO_B,D),
\] and $\abs{D} \leq \abs{R}\abs{O_A}\abs{O_B}$.
The vector $\ket{\psi} \in RI_AI_B$ is written $\ket{\psi} = \sum_i \ket{i}\ket{\psi_i}, \ket{\psi_i} \in I_AI_B$. Together we obtain the full game tensor \begin{align*}
    \xi_{\mathcal  G} = \tr_R \left( (\psi \otimes 1_{S_\infty(O_AO_B)})(P^{T_O} \otimes 1_{S_1(I_AI_B)}) \right) = \sum_{i,j} \op{\psi_i}{\psi_j} \otimes y_i^*y_j \in S_1(I_AI_B) \otimes S_\infty(O_AO_B),
    \end{align*} where $y_i = \sum_t \op{t}{\oline{\g}_t^i}$. Our half tensor $\eta_{\mathcal  G} \in C_D \otimes R_{O_AO_B} \otimes C_{I_AI_B}$, \[
\eta_{\mathcal  G}= \sum_{j,t} \ket{t} \otimes \bra{\oline{\g}_t^j} \otimes \ket{\oline{\psi}_j}. 
    \] 
    
 \subsubsection{Entangled value}

We now introduce the \emph{entangled value} $\om^*(\mathcal  G)$ of a two-player quantum game $\mathcal  G$:
\begin{definition}\label{defn: entangle value of game}
    Given a two-player quantum game $\mathcal  G(\ket{\psi}, P)$ then the \emph{entangled value}, $\om^*(\mathcal  G)$, of the game $\mathcal  G$ is defined as \begin{align}
        \om^*(\mathcal  G) = \sup \tr(P \cdot (\id_{S_1(R)}\otimes \vp_A \otimes \vp_B)(\psi \otimes \rho)),
    \end{align} where the supremum is taken over all Hilbert spaces $E_A,E_B$, all strategies $\vp_A: S_1(I_AE_A) \to S_1(O_A), \vp_B: S_1(I_BE_B) \to S_1(O_B)$ and all pure states $\rho\in S_1(E_AE_B)$. 
\end{definition}

Analogous to Lemma~\ref{lem: replace cc with cptp} we have the following: 

\begin{lemma}\label{lem: replace cc with cptp two player}
    Given a two-player quantum game $\mathcal  G(\ket{\psi}, P)$ let \[
\tilde{\om}^*(\mathcal  G) = \sup \tr(P \cdot (\id_{S_1(R)}\otimes \vp_A \otimes \vp_B)(\psi \otimes \rho)),
    \] where the supremum is taken over all Hilbert spaces $E_A$, $E_B$, all complete contractions $\vp_A: S_1(I_AE_A) \to S_1(O_A), \vp_B: S_1(I_BE_B) \to S_1(O_B)$ and all, not necessarily positive, $\rho \in S_1(E_AE_B), \norm{\rho}_{S_1(E_AE_B)} \leq 1$. Then \[
\tilde{\om}^*(\mathcal  G) = \om^*(\mathcal  G).
    \] 
\end{lemma}

\begin{proof}We only need to prove $\tilde{\om}^*(\mathcal  G) \leq \om^*(\mathcal  G)$.
    We follow the same methods as Lemma~\ref{lem: replace cc with cptp}. Via a convexity argument, we may assume $\rho$ has the form $\rho = \op{\xi}{\eta}$, where $\ket{\xi}, \ket{\eta} \in E_AE_B$ are unit vectors. Consider two complete contractions \[
    \vp_A: S_1(I_AE_A) \to S_1(O_A), \vp_B: S_1(I_BE_B) \to S_1(O_B).
    \]Using the fact that both $\vp_A^\dag: S_\infty(O_A) \to S_\infty(I_AE_A), \vp_B^\dag: S_\infty(O_B) \to S_\infty(I_BE_B)$ are complete contractions, we apply Theorem~\ref{thm: cb decomposition as sum} to both operators obtaining c.b. decompositions \begin{align*}
\vp_A^\dag(\cdot) &= \sum_i A_i(\cdot)B_i, A_i: O_A \to I_AE_A, B_i: I_AE_A \to O_A, \\ \vp_B^\dag(\cdot) &= \sum_j C_j(\cdot)D_j, C_j: O_B \to I_BE_B, D_j: I_BE_B \to O_B.
    \end{align*} We may assume without loss of generality that \[
\max\Big\{ \Big\|\sum_i A_iA_i^\dag\Big\|, \Big\|\sum_i B_i^\dag B_i\Big\|, \Big\|\sum_i C_iC_i^\dag\Big\|, \Big\|\sum_i D_i^\dag D_i\Big\| \Big\} \leq 1. 
    \] As shown in Lemma \ref{lem: replace cc with cptp}, one can show that $$\tr\left(P \cdot (\id_{S_1(R)}\otimes \vp_A\otimes \vp_B)(\psi \otimes \rho)\right)\leq \Delta_1\cdot \Delta_2,$$where 
    \begin{align*}
     &\Delta_1=\Big( \sum_{i,j} \tr \Big( (\id_R\otimes B_i^\dag \otimes D_j^\dag) P (\id_R\otimes B_i \otimes D_j ) (\psi \otimes \op{\eta}{\eta}) \Big) \Big)^\frac{1}{2},\\&
     \Delta_2=\Big( \sum_{i,j} \tr \Big( ( \id_R\otimes A_i \otimes C_j ) P (\id_R\otimes A_i^\dag \otimes C_j^\dag) (\psi \otimes \op{\xi}{\xi}) \Big) \Big)^\frac{1}{2}.
    \end{align*}

  We define the operators $\vp_{A,i}: S_\infty(O_A) \to S_\infty(I_AE_A), \vp_{B,i}: S_\infty(O_B) \to S_\infty(I_BE_B)$, for $i=1,2$, via \[
\vp_{A,1}(x) = \sum_i B_i^\dag x B_i, \vp_{A,2}(x) = \sum_i A_i x A_i^\dag, \vp_{B,1}(x) = \sum_j D_j^\dag x D_j, \vp_{B,2}(x) = \sum_j C_j x C_j^\dag. 
    \]
By taking the supremum over all strategies we deduce $$\tilde{\om}^*(\mathcal  G) \leq \Gamma_1\cdot \Gamma_2,$$where 
\begin{align*}
&\Gamma_1=\left( \tr \left( (\id_{S_1(R)}\otimes \vp_{A,1} \otimes \vp_{B,1})(P)) \cdot (\psi \otimes \op{\eta}{\eta})\right)\right)^\frac{1}{2},\\&
\Gamma_2=\left( \tr \left((( \id_{S_1(R)}\otimes \vp_{A,2} \otimes \vp_{B,2} )(P)) \cdot (\psi \otimes \op{\xi}{\xi})\right)\right)^\frac{1}{2}
\end{align*}

The operators $\vp_{A,i}, \vp_{B,i}$ are completely positive complete contractions. Therefore we may assume the operators $\vp_A, \vp_B$ are completely positive complete contractions in the definition of $\tilde{\om}^*(\mathcal  G)$. Moreover, we may assume $\rho$ is a state in the definition of $\tilde{\om}^*(\mathcal  G)$. \newline The final claim that we may take the $\vp_A,\vp_B$ to also be trace preserving is analogous to the proof in Lemma~\ref{lem: replace cc with cptp}. This finishes the proof. \qedhere 
\end{proof}

\begin{corollary}\label{cor: entangled value, full tensor}
    Given a two-player quantum game $\mathcal  G(\ket{\psi}, P)$ it follows \[
\om^*(\mathcal  G) = \norm{\xi_{\mathcal  G}}_{S_1(I_A; S_\infty(O_A)) \otimes_{min} S_1(I_B; S_\infty(O_B))}.
    \]
\end{corollary}
\begin{proof}
    For Hilbert spaces $I,O,E$, we have the completely isometric identification \begin{align*}
        &\CB(S_1(I;S_\infty(O),S_\infty(E)) = \CB(S_1(I) \frownotimes S_\infty(O), S_\infty(E)) = S_\infty(I) \otimes_{min} S_1(O) \otimes_{min}S_\infty(E) \\&= S_\infty(I) \otimes_{min} S_\infty(E) \otimes_{min} S_1(O) = S_\infty(IE) \otimes_{min} S_1(O) = \CB(S_1(IE), S_1(O)).
    \end{align*} To this end, given a linear operator $T \in \CB(S_1(I; S_\infty(O)), S_\infty(E))$, let $\wh{T} \in \CB(S_1(IE), S_1(O))$ denote the corresponding operator under the above identifications. 

    If we start by calculating the norm of the full tensor we see \begin{align*}
\norm{\xi_{\mathcal  G}}_{S_1(I_A; S_\infty(O_A)) \otimes_{min} S_1(I_B; S_\infty(O_B))} &=  \sup \norm{(T_A \otimes T_B)(\xi_\mathcal G)}_{S_\infty(E_AE_B)}
\\&= \sup \tr \left((T_A \otimes T_B)(\xi_{\mathcal  G})\rho^T \right),
    \end{align*} where the supremum is taken over all complete contractions $T_A: S_1(I_A; S_\infty(O_A)) \to S_\infty(E_A), T_B: S_1(I_B; S_\infty(O_B)) \to S_\infty(E_B)$, and elements $\rho \in S_1(E_AE_B)$ such that $\|\rho\|_{S_1(E_AE_B)}\leq 1$. Our reasoning above coupled with Lemma~\ref{lem: replace cc with cptp two player} yields the equality \[
\tr \left((T_A \otimes T_B)(\xi_{\mathcal  G})\rho^T \right) = \tr\left(P \cdot(\id_{S_1(R)} \otimes \wh{T}_A \otimes \wh{T}_B)(\psi \otimes \rho))\right),
    \] which proves the equality. \qedhere 
\end{proof}

Lemma~\ref{lem: replace cc with cptp two player} was proven independently in \cite{Crann} and generalizes the analogous result proven in \cite{palazuelos2016survey} in which it is shown that for a classical game $\mathcal G$, its quantum value is given as \[
\om^*(\mathcal G) = \norm{\mathcal G}_{\ell_1(I_A, \ell_\infty(O_A)) \otimes_{min} \ell_1(I_B, \ell_\infty(O_B))}.
\] It was also shown that the classical value of the classical game $\mathcal G$ is given as \[
\om(\mathcal G) = \norm{\mathcal G}_{\ell_1(I_A, \ell_\infty(O_A)) \otimes_{\e} \ell_1(I_B, \ell_\infty(O_B))}.
\] The previous calculations, specialized to the case where $E_A = E_B = \mathbb{C}$ (so that the state $\rho$ does not appear in the computations), lead to the classical value (or, in the case of quantum games, more appropriately, the non-entangled value) of the game $G$ (see also \cite{Crann}):
\[
\omega(\mathcal{G}) = \|\xi_{\mathcal{G}}\|_{S_1(I_A; S_\infty(O_A)) \otimes_{\varepsilon} S_1(I_B; S_\infty(O_B))}.
\]

Corollary~\ref{cor: entangled value, full tensor} allows us to calculate the entangled value of a game via the half tensor: 

\begin{theorem}\label{thm: entangled value, half-tensor}
    Given a two-player quantum game $\mathcal  G(\ket{\psi},P)$, we have \[
\om^*(\mathcal  G) = \norm{\eta_{\mathcal  G}}_{C_D(S_1(I_A,O_A) \otimes_{min} S_1(I_B,O_B))}^2.
    \]
\end{theorem}

\begin{proof}
    We have \[
\om^*(\mathcal  G) = \sup \tr(P \cdot(\id_{S_1(R)}\otimes \vp_A \otimes \vp_B)(\psi \otimes \rho))
    \] where the supremum is taken over all Hilbert spaces $E_A,E_B$, strategies $\vp_A: S_1(I_AE_A) \to S_1(O_A), \vp_B: S_1(I_BE_B) \to S_1(O_B)$, and (we may assume pure) quantum states $\rho = \op{\xi}{\xi} \in S_1(E_AE_B)$. Once again, we will prove this result in the same spirit as that of its one-player counterpart (see Theorem~\ref{thm: value of game, one-player, half tensor}).

    Applying Stinespring's theorem to our strategies, there exist Hilbert spaces $K_A,K_B$, unit vectors $\ket{\xi_A} \in K_A, \ket{\xi_B} \in K_B$, and contractions $U_A: O_AK_A \to I_AE_AK_A,U_B: O_BK_B \to I_BE_BK_B$ such that \[
\vp_A(\rho) = \tr_{K_A} U_A^\dag(\rho \otimes \xi_A)U_A,\,  \vp_B(\sigma) = \tr_{K_B} U_B^\dag(\sigma \otimes \xi_B)U_B,
\] where $\xi_A = \op{\xi_A}{\xi_A}, \xi_B = \op{\xi_B}{\xi_B}$, for every $\rho \in S_1(I_AE_A)$ and $\sigma \in S_1(I_BE_B)$. Moreover, every map of this form yields a completely positive complete contraction from $S_1(I_AE_A)\to S_1(O_A)$ and $S_1(I_BE_B) \to S_1(O_B)$, respectively. To simplify the notation, we will denote $O_{AB}=O_AO_B$, $K_{AB}=K_AK_B$ and $\xi_{AB}=\xi_A\xi_B$. Then, applying Lemma~\ref{lem: replace cc with cptp two player} we observe 
\begin{align*}
    &\om^*(\mathcal  G) = \sup \tr_{S_1(RO_{AB})}
    \left[ P \cdot \tr_{K_{AB}} ((\id_R \otimes U_A^\dag \otimes U_B^\dag)(\psi \otimes \op{\xi\xi_{AB}}{\xi\xi_{AB}})(\id_R \otimes U_A \otimes U_B)) \right]
\\ &= \sup \tr_{S_1(RO_{AB}K_{AB})} 
\Big[ (P \otimes \id_{K_{AB}})\cdot (\id_R \otimes U_A^\dag \otimes U_B^\dag)(\psi \otimes \op{\xi \xi_{AB}}{\xi \xi_{AB}})(\id_R \otimes U_A \otimes U_B)\Big],
\end{align*}where the supremum is taken over all Hilbert spaces $E_A$, $E_B$, $K_A$, $K_B$, all contractions $U_A: O_AK_A \to I_AE_AK_A, U_B: O_BK_B \to I_BE_BK_B$ and vectors $\norm{\ket{\xi}}_{E_AE_B}, \norm{\ket{\xi_A}}_{K_A}, \norm{\ket{\xi_B}}_{K_B}$ with norm lower than or equal one. Recalling that $P = \sum_{i,j,t} \op{i}{j} \otimes \op{\g_t^i}{\g_t^j}$ and $\psi = \sum_{i,j} \op{i}{j} \otimes \op{\psi_i}{\psi_j}$, and after tracing out the Referee Hilbert space,  we rewrite $\om^*(\mathcal  G)$ as 
\begin{align*}
    \sup \tr_{S_1(O_{AB}K_{AB})} 
     \Big(\sum_{i,j,t} (\op{\g_t^i}{\g_t^j} \otimes \id_{K_{AB}}) \cdot (U_A^\dag \otimes U_B^\dag)(\op{\psi_j\xi\xi_{AB}}{\psi_i\xi\xi_{AB}})(U_A\otimes U_B)\Big).
     \end{align*} Then, one can easily deduce that 
     \begin{align*}\om^*(\mathcal  G)= \sup \sum_t \Big\|\sum_{j} (\bra{\g_t^j} \otimes \id_{K_{AB}})\cdot (U_A^\dag \otimes U_B^\dag)(\ket{\psi_j \xi \xi_{AB}})\Big\|_{K_{AB}}^2,
      \end{align*}where the supremum is taken over all Hilbert spaces $E_A$, $E_B$, $K_A$, $K_B$, all contractions \[\norm{U_A}_{S_\infty(O_AK_A,I_AE_AK_A)}, \norm{U_B}_{S_\infty(O_BK_B,I_BE_BK_B)} \leq 1,\] and all vectors $\norm{\ket{\xi}}_{E_AE_B}, \norm{\ket{\xi_A}}_{K_A}, \norm{\ket{\xi_B}}_{K_B} \leq 1$.

It can be shown that the above value equals \[
\om^*(\mathcal  G) = \sup \sum_t \Big\|\sum_j (\bra{\g_t^j} \otimes \id_{H_AH_B})\cdot (V_A \otimes V_B)(\ket{\psi_j \eta})\Big\|_{H_AH_B}^2,
\] where the supremum is taken over all Hilbert spaces $H_A$, $H_B$, all contractions \[\norm{V_A}_{S_\infty(I_AH_A, O_AH_A)}, \norm{V_B}_{S_\infty(I_BH_B,O_BH_B)} \leq 1,\] and all vectors $\norm{\ket{\eta}}_{H_AH_B} \leq 1$. \newline 

If we now calculate the norm of the half tensor we see \begin{align*}
    \norm{\eta_{\mathcal  G}}_{C_D (S_1(I_A,O_A) \otimes_{min} S_1(I_B,O_B))} &= \norm{\eta_{\mathcal  G}}_{C_D \otimes_{min} (S_1(I_A,O_A) \otimes_{min} S_1(I_B,O_B))} 
    \\&= \sup \norm{ (\id_{C_D} \otimes T_A \otimes T_B)(\eta_{\mathcal  G})}_{C_D \otimes_{min} S_\infty(E_A) \otimes_{min} S_\infty(E_B)}
    \\&= \sup \norm{ (\id_{C_D} \otimes T_A \otimes T_B)(\eta_{\mathcal  G})}_{C_D \otimes_{min} S_\infty(E_AE_B)}
    \\&= \sup \Big\|\sum_{j,t} \ket{t} \otimes (T_A \otimes T_B)(\bra{\oline{\g}_t^j} \otimes \ket{\oline \psi_j}\Big\|_{C_D \otimes_{min} S_\infty(E_AE_B)},
\end{align*} where the supremum is taken over all complete contractions $T_A: S_1(I_A,O_A) \to S_\infty(E_A), T_B: S_1(I_B,O_B) \to S_\infty(E_B)$. As before, let us denote for each $t \in D$ the operator $A_t:= \sum_{j} (T_A \otimes T_B)(\bra{\oline{\g}_t^j} \otimes \ket{\oline{\psi}_j})$. As in Theorem~\ref{thm: value of game, one-player, half tensor}, we have \[
\norm{\eta_{\mathcal  G}}_{C_D \otimes_{min} (S_1(I_A, O_A) \otimes_{min} S_1(I_B,O_B))} = \sup \Big( \sum_t \norm{A_t\ket{\eta}}_{E_AE_B}^2\Big)^\frac{1}{2},
\] where the supremum is taken over all complete contractions $T_A: S_1(I_A,O_A) \to S_\infty(E_A), T_B: S_1(I_B,O_B) \to S_\infty(E_B)$ and unit vectors $\ket{\eta} \in E_AE_B$. In Theorem~\ref{thm: value of game, one-player, half tensor} we used Equation~\eqref{eq: useful identification, T to V} where, if given a complete contraction $T: S_1(I,O) \to S_\infty(E)$, we may take the supremum over contractions $V: IE \to OE$, where if $\ket{\psi} \in I, \ket{\g} \in O, \ket{\xi} \in E$ then \[
T(\bra{\oline{\g}} \otimes \ket{\oline{\psi}})\ket{\xi} = {(\bra{\oline{\g}} \otimes \id_E)V(\ket{\oline{\psi} \xi}}).
\] 
Using this for both $T_A$ and $T_B$, our previous expression becomes \[
\norm{\eta_{\mathcal  G}}_{C_D \otimes_{min} (S_1(I_A,O_A) \otimes_{min} S_1(I_B,O_B))}^2 = \sup \sum_t \Big\|\sum_j (\bra{\g_t^j} \otimes \id_{E_AE_B})(V_A \otimes V_B)(\ket{\psi_j\eta})\Big\|_{E_AE_B}^2 ,
\] where the supremum is over all Hilbert spaces $E_A$, $E_B$, all contractions $V_A: I_AE_A \to O_AE_A, V_B: I_BE_B \to O_BE_B$ and all unit vectors $\ket{\eta}_{E_AE_B}$.
This concludes the proof. \qedhere 
\end{proof}

\subsubsection{One-way value}\label{sec: one-way values of games}

We will now focus on the so-called \emph{one-way value} of a two-player quantum game $\mathcal  G$. This value of the game corresponds to Alice sending quantum information to Bob (but not vice versa) as part of their strategy, i.e., the supremum will be taken over strategies from the set $\mathcal {OW}_{A \to B}(I_{AB}, O_{AB})$, as defined in Section~\ref{sec: ns strategies and symmetrized one-way}. Analogously, we may consider the one-way value by allowing Bob to send quantum information to Alice, but not vice-versa. Formally: 
\begin{definition}\label{defn: value, one-way}
    Given a two-player quantum game $\mathcal  G(\ket{\psi}, P)$ then the \emph{one-way} value of $\mathcal  G$, $\om_{A \to B}(\mathcal  G)$ defined by allowing Alice to send quantum information to Bob (but not vice versa), is defined as \begin{align}
        \om_{A\to B}(\mathcal  G) = \sup \tr\left(P \cdot (\id_{S_1(R)}\otimes \vp)(\psi)\right),
        \end{align} where the supremum is taken over all $\vp \in \mathcal {OW}_{A \to B}(I_{AB}, O_{AB})$.  Equivalently, \begin{align}
            \om_{A \to B}(\mathcal  G) = \sup \tr\left(P \cdot (\id_{S_1(R)}\otimes \id_{S_1(O_A)}\otimes \vp_B) \circ (\id_{S_1(R)}\otimes \vp_A \otimes \id_{S_1(I_B)})(\psi)\right),
    \end{align} where the supremum is taken over all Hilbert spaces $K$, and quantum channels \[
\vp_A: S_1(I_A) \to S_1(O_AK),\,  \vp_B: S_1(KI_B) \to S_1(O_B),
    \] such that $\vp = (\id_{S_1(O_A)} \otimes \vp_B) \circ (\vp_A \otimes \id_{S_1(I_B)})$.
\end{definition}

\begin{remark}
To simplify notation we will writer $\id_{S_1(R)}\otimes \id_{S_1(O_A)}=\id_{S_1(RO_A)}$ and $\id_{S_1(R)}\otimes \id_{S_1(I_B)}=\id_{S_1(RI_B)}$ so that 
\begin{align}
            \om_{A \to B}(\mathcal  G) = \sup \tr\left(P \cdot (\id_{S_1(RO_A)}\otimes \vp_B) \circ (\id_{S_1(RO_A)}\otimes \vp_A )(\psi)\right).
    \end{align}
\end{remark}

As we saw for the honest value of a game, and the entangled value of a game, it follows that to calculate the one-way value of $\mathcal  G$, one may consider  complete contractions $\vp_A: S_1(I_A) \to S_1(O_AK), \vp_B: S_1(KI_B) \to S_1(O_B)$ such that $\vp = (\id \otimes \vp_B) \circ (\vp_A \otimes \id)$. Therefore, analogous to both Lemma~\ref{lem: replace cc with cptp} and Lemma~\ref{lem: replace cc with cptp two player}, we have: 

\begin{lemma}\label{lem: replace cc with cptp, two-player one-way}
    Given a two-player quantum game $\mathcal  G(\ket{\psi},P)$ consider the value \[
\tilde{\om}(\mathcal  G):= \sup \tr\left( P \cdot ( \id_{S_1(RO_A)} \otimes \vp_B) \circ (\id_{S_1(RI_B)}\otimes  \vp_A)(\psi)\right),
    \] where the supremum is taken over all complete contractions $\vp_A: S_1(I_A) \to S_1(O_AK)$ and $\vp_B: S_1(KI_B) \to S_1(O_B)$. Then \[
\tilde{\om}(\mathcal  G) = \om_{A\to B}(\mathcal  G).
    \] 
\end{lemma}

\begin{proof}We only need to prove the nontrivial inequality. We once again use our insight gained in the proof of Lemma~\ref{lem: replace cc with cptp}. Let us denote our complete contractions as $\vp_A: S_1(I_A) \to S_1(O_AK)$ and $\vp_B: S_1(KI_B) \to S_1(O_B)$, which in turn yield complete contractions $\vp_A^\dag: S_\infty(O_AK) \to S_\infty(I_A)$ and $\vp_B^\dag: S_\infty(O_B) \to S_\infty(KI_B)$, for Hilbert space $K$. Via Theorem~\ref{thm: cb decomposition as sum} there exist decompositions \begin{align*}
        \Phi_A^\dag(\cdot) = \sum_i A_i(\cdot)B_i, \, \Phi_B^\dag(\cdot) = \sum_j C_j(\cdot)D_j,
    \end{align*} where $A_i: O_AK \to I_A, B_i: I_A \to O_AK, C_j: O_B \to KI_B, D_j: KI_B \to O_B$, for every $i,j$. We calculate: 
    \begin{align}\label{Eq: one_way}
         &\tr\left(P \cdot( \id_{S_1(R)}\otimes \id_{S_1(O_A)} \otimes \vp_B)\circ(\id_{S_1(R)}\otimes \vp_A \otimes \id_{S_1(I_B)}))(\psi)\right)
        \\ &\nonumber= \tr\Big( \Big((\id_{S_\infty(R)}\otimes \vp_A^\dag \otimes \id_{S_\infty(I_B)}) \circ (\id_{S_\infty(R)}\otimes \id_{S_\infty(O_A)} \otimes \vp_B^\dag)(P^\dag)\Big)^\dag (\psi)\Big)
        \\&\nonumber= \sum_{i,j} \tr \Big( \Big(  (B_i^\dag \otimes \id_{I_BR})(D_j^\dag \otimes \id_{O_AR})P(C_j^\dag \otimes \id_{O_AR})(A_i^\dag \otimes \id_{I_BR}) \Big) \cdot \psi \Big)
       \\&\nonumber= \sum_{i,j} \tr \left( \psi^\frac{1}{2} \cdot \left(  (B_i^\dag \otimes \id_{I_BR})(D_j^\dag \otimes \id_{O_AR})PP(C_j^\dag \otimes \id_{O_AR})(A_i^\dag \otimes \id_{I_BR}) \right) \cdot \psi^\frac{1}{2} \right).
    \end{align}
    We define the following completely positive complete contractions for $i=1,2$: \begin{align*}
        \vp_{A,i}: S_\infty(O_AK) \to S_\infty(I_A), \vp_{B,i}: S_\infty(O_B) \to S_\infty(KI_B),
    \end{align*} by \begin{align*}
        \vp_{A,1}(x) = \sum_i B_i^\dag x B_i, \vp_{A,2}(x) = \sum_i A_i x A_i^\dag, \vp_{B,1}(y) = \sum_j D_j^\dag y D_j, \vp_{B,2}(y) = \sum_j C_j y C_j^\dag,
    \end{align*} for every $x \in S_\infty(O_AK)$ and $y \in S_\infty(O_B)$. Then, Cauchy-Schwartz inequality allows us to upper bound Equation (\ref{Eq: one_way}) by $\Delta_1\cdot\Delta_2$, where 
     \begin{align*}
         &\Delta_1=\tr \left( ( \vp_{A,1}\otimes \id_{S_\infty(RI_B)})(\vp_{B,1}\otimes \id_{S_\infty(RO_A)})(P)\cdot \psi \right)^\frac{1}{2}.\\&
         \Delta_2=\tr \left( ( \vp_{A,2}\otimes \id_{S_\infty(RI_B)})(\vp_{B,2}\otimes \id_{S_\infty(RO_A)})(P)\cdot \psi \right)^\frac{1}{2}.
                 \end{align*}

            Thus, this proves we may take our operators $\vp_A, \vp_B$ to be completely positive complete contractions in the definition of $\tilde{\om}(\mathcal  G)$. Finally, the claim that we may take $\varphi_A$ and $\varphi_B$ to also be trace-preserving is proved as in Lemma~\ref{lem: replace cc with cptp}, by applying the same argument to both $\varphi_{A,i}$ and $\varphi_{B,i}$. This finishes the proof. \qedhere 
    \end{proof}

We now arrive at the corresponding half-tensor calculation: 
\begin{theorem}\label{thm: one-way value A to B, half-tensor}
    Given a two-player quantum game $\mathcal  G(\ket{\psi}, P)$ it follows \begin{align*}
        \om_{A \to B}(\mathcal  G) = \norm{\eta_{\mathcal  G}}_{C_D ( S_1(I_B,O_B) \otimes_h S_1(I_A,O_A))}^2.
    \end{align*}
\end{theorem}

\begin{proof}
    By definition we have \[
\om_{A \to B}(\mathcal  G) = \sup \tr\left(P \cdot ( \id_{S_1(RO_A)} \otimes \vp_B)\circ (\id_{S_1(RI_B)}\otimes \vp_A)(\psi)\right),
    \] where the supremum is taken over all Hilbert spaces $K$ and all completely positive trace-preserving operators \[
\vp_A: S_1(I_A) \to S_1(O_AK), \vp_B: S_1(KI_B) \to S_1(O_B).
    \] Applying Stinespring's theorem to both channels there exist Hilbert spaces $E_A,E_B$, unit vectors $\ket{\xi_A} \in E_A, \ket{\xi_B} \in E_B$, and contractions $U_A: O_AKE_A \to I_AE_A, U_B: O_BE_B \to KI_BE_B$, such that \[
\vp_A(x) = \tr_{E_A} U_A^\dag(x \otimes \op{\xi_A}{\xi_A})U_A, \vp_B(y) = \tr_{E_B}U_B^\dag(y \otimes \op{\xi_B}{\xi_B})U_B,
    \] for every $x \in S_1(I_A)$ and every $y \in S_1(KI_B)$. Moreover, every map of this form defines a completely positive complete contraction. Let us denote $\tilde{U}_A:= U_A \otimes \id_{RI_BE_B}$ and $\tilde{U}_B:= U_B \otimes \id_{RO_AE_A}$. By Lemma~\ref{lem: replace cc with cptp, two-player one-way} we may express the one-way value as \begin{align*}
        \om_{A \to B}(\mathcal  G) &= \sup \tr \left[P \cdot (( \id_{S_1(RO_A)} \otimes \vp_B)\circ (\id_{S_1(RI_B)}\otimes \vp_A) (\psi) \right] 
        \\& = \sup \tr_{S_1(RO_AO_BE_AE_B)} \left[(P \otimes \id_{S_1(E_AE_B)}) \cdot (\tilde{U}_B^\dag \tilde{U}_A^\dag(\psi \otimes \op{\xi_A\xi_B}{\xi_A\xi_B})\tilde{U}_A\tilde{U}_B) \right],
    \end{align*} where the supremum is taken over all Hilbert spaces $K$, all contractions \[\norm{U_A}_{S_\infty(O_AKE_A,I_AE_A)}, \norm{U_B}_{S_\infty(O_BE_B,KI_BE_B)} \leq 1,\] and $\norm{\ket{\xi_A}}_{E_A}, \norm{\ket{\xi_B}}_{E_B} = 1$. Recalling \[
P = \sum_{i,j,t} \op{i}{j} \otimes \op{\g_t^i}{\g_t^j}, \psi = \sum_{i,j} \op{i}{j} \otimes \op{\psi_i}{\psi_j},
    \] and we denote $\wh{U}_A = U_A \otimes \id_{I_BE_B}, \wh{U}_B:= U_B \otimes \id_{O_AE_A}$. Rewriting with $\wh{U}_A, \wh{U}_B$ and tracing out the referee space we obtain \begin{align*}
    \om_{A \to B}(\mathcal  G) &= \sup \sum_{i,j,t} \tr_{S_1(O_AO_BE_AE_B)} \Big( ( \op{\g_t^i}{\g_t^j} \otimes \id_{S_1(E_AE_B)})(\wh{U}_B^\dag \wh{U}_A^\dag(\op{\psi_j\xi_A\xi_B}{\psi_i\xi_A\xi_B})\wh{U}_A\wh{U}_B) \Big)
\\ &= \sup \sum_t \Big\|\sum_j (\bra{\g_t^j} \otimes \id_{E_AE_B})(\wh{U}_B^\dag \wh{U}_A^\dag\ket{\psi_j\xi_A\xi_B}\Big\|_{E_AE_B}^2
\\&= \sup \sum_t \Big\|\sum_j (\bra{\g_t^j} \otimes \id_{E_AE_B})(U_B^\dag \otimes \id_{O_AE_A})(U_A^\dag \otimes \id_{I_BE_B}) \ket{\psi_j\xi_A\xi_B}\Big\|_{E_AE_B}^2,
        \end{align*} where the supremum is taken over all Hilbert spaces $K$, all contractions \[\norm{U_A}_{S_\infty(O_AKE_A, I_AE_A)}, \norm{U_B}_{S_\infty(O_BE_B, KI_BE_B)} \leq 1,\] and unit vectors $\ket{\xi_A} \in E_A, \ket{\xi_B} \in E_B$. It can be seen that we may express the above norm as \[
\om_{A \to B}(\mathcal  G) = \sup \sum_t \Big\|\sum_j (\bra{\g_t^j} \otimes \id_{K})(V_B \otimes \id_{O_A})(V_A \otimes \id_{I_B})\ket{\psi_j\eta}\Big\|_{K}^2,
        \]  where the supremum is taken over all Hilbert spaces $K$, all contractions \[\norm{V_A}_{S_\infty(I_AK, O_AK)}, \norm{V_B}_{S_\infty(I_BK, O_BK)} \leq 1,\] and vectors $\norm{\ket{\eta}}_{K} \leq 1$. \newline Considering the norm of the half-tensor, we write \begin{align*}
            \norm{\eta_{\mathcal  G}}_{C_D(S_1(I_B,O_B) \otimes_h S_1(I_A,O_A))} &= \norm{\eta_{\mathcal  G}}_{C_D \otimes_{min}(S_1(I_B,O_B) \otimes_h S_1(I_A,O_A))} 
            \\&= \sup \norm{(\id_{C_D} \otimes T_B \odot T_A)(\eta_{\mathcal  G})}_{C_D \otimes_{min} S_\infty(E)},
        \end{align*} where the supremum is taken over all complete contractions $T_A: S_1(I_A,O_A) \to S_\infty(E), T_B: S_1(I_B,O_B) \to S_\infty(E)$. Note we have used $T_B \odot T_A$ to denote the linear mapping \[
T_B \odot T_A: S_1(I_B,O_B) \otimes S_1(I_A,O_A) \to S_\infty(E), \, x \otimes y \mapsto T_B(x) \cdot T_A(y).
        \] As we have done before, for each $t$, we define the operator $A_t:= \sum_{j} T_B \odot T_A(\bra{\oline{\g}_t^j} \otimes \ket{\oline{\psi}_j})$, and thus we may express the above norm as \[
\norm{\eta_{\mathcal  G}}_{C_D \otimes_{min} (S_1(I_B,O_B) \otimes_h S_1(I_A,O_A))} ^2= \sup \sum_t \norm{A_t \ket{\psi}}_E^2,
        \] where the supremum is taken over Hilbert spaces $E$, all complete contractions $T_A: S_1(I_A,O_A) \to S_\infty(E), T_B: S_1(I_B,O_B) \to S_\infty(E)$ and all unit vectors $\ket{\varphi} \in E$. Employing our identification Equation~\eqref{eq: useful identification, T to V} we express the above as \[
\sup \sum_t \Big\|\sum_j (\bra{\g_t^j} \otimes \id_{E})(V_B \otimes \id_{O_A})(V_A \otimes \id_{I_B})\ket{\psi_j \vp}\Big\|_E^2,
        \] where the supremum runs over all Hilbert spaces $E$, and all contractions $V_A: I_AE \to O_AE$ and $V_B: I_BE \to O_BE$ and all unit vectors $\ket{\varphi} \in E$. This finishes the proof. \qedhere
\end{proof}

We conclude this section with the corresponding full-tensor calculation: 

\begin{theorem}\label{thm: one-way value, full tensor}
    Given a two-player quantum game $\mathcal  G(\ket{\psi}, P)$, then the one-way value may be expressed as \[
\om_{A \to B}(\mathcal  G)= \norm{\xi_\mathcal  G}_{S_1(I_A; S_\infty(O_A; S_1(I_B; S_\infty(O_B))))}.
    \]
\end{theorem}

\begin{remark}\label{Rem: norm_p_positive}

In the next proof, we will need to deal with the norm of positive elements in the space
\[
S_\infty(I_A; S_1(O_A; S_\infty(I_B; S_1(O_B)))).
\]

Note \cite[Eq. (3.23)]{Dev06} that for every positive element \(y\in S_1(J; S_\infty(I_B; S_1(O_B)))\), Eq.~(\ref{Ec:S_1(X)})  has the form 
\begin{align}\label{Eq: (1,infty,1)}
\|y\|_{S_1(O_A; S_\infty(I_B; S_1(O_B)))}
    &=
    \inf \Bigl\{
        \|B\|^2_{S_2(O_A)}\, \|W\|_{S_\infty(O_AI_B; S_1(O_B))}
    \Bigr\},
\end{align}where the infimum runs over all positive definite matrices \(B \in S_\infty(O_A)\) and all positive semidefinite matrices $W$ such that
\(y = (B \otimes 1)\, W\, (B \otimes 1)\). In fact, this expression is equivalent to  
\begin{align*}
\|y\|_{S_1(O_A; S_\infty(I_B; S_1(O_B)))}
    &=
    \inf \Bigl\{
        \|(B^{-1/2}\otimes 1)y(B^{-1/2}\otimes 1)\|_{S_\infty(O_AI_B; S_1(O_B))}
    \Bigr\},
\end{align*}where the infimum is over all positive definite matrices \(B \in S_\infty(O_A)\) with $\tr(B)=1$.

Equation (\ref{Eq: (1,infty,1)}) allows us to prove a Cauchy-Schwarz inequality for the space $S_1(O_A; S_\infty(I_B; S_1(O_B)))$ using the same argument as in \cite[Lemma 9]{Dev06}. That is, we have 
\begin{align}\label{CS-inequality}
\|x^*y\|_{S_1(O_A; S_\infty(I_B; S_1(O_B)))}\leq \Big(\|x^*x\|_{S_1(O_A; S_\infty(I_B; S_1(O_B)))}\|y^*y\|_{S_1(O_A; S_\infty(I_B; S_1(O_B)))}\Big)^\frac{1}{2}.
\end{align}

Finally, using Equation~(\ref{CS-inequality}), one easily obtains that for every positive element $x$,
\[
\|x\|_{S_\infty(I_A; S_1(O_A; S_\infty(I_B; S_1(O_B)))}
=
\sup \Bigl\{
\|(A \otimes 1)\, x\, (A \otimes 1)\|_{S_1(I_A O_A; S_\infty(I_B; S_1(O_B)))}
\Bigr\},
\]
where the supremum is taken over all positive elements \(A\) such that \(\|A\|_{S_2(I_A)} \le 1\).
\end{remark}

\begin{proof}
    For notational ease we will denote the operator space $$\mathcal  X=S_1(I_A; S_\infty(O_A; S_1(I_B; S_\infty(O_B)))).$$ Using the identifications (\ref{Haagerup K}) and (\ref{Haagerup T}), together with the associative property of the Haagerup norm, we obtain the identification \[
S_1(I_A; S_\infty(O_A; S_1(I_B; S_\infty(O_B)))) = (R_{I_A} \otimes_h C_{O_A} \otimes_h R_{I_B} \otimes_h C_{O_B}) \otimes_h (R_{O_B} \otimes_h C_{I_B} \otimes_h R_{O_A} \otimes_h C_{I_A}),
\] completely isometrically. Writing $\xi_{\mathcal  G}$ as the product of half-tensors, $\xi_{\mathcal  G} = \eta_{\mathcal  G}^*\eta_{\mathcal  G}$ and applying Ec. (\ref{Ec. Haagerup_norm}) for the case $n=1$ and Theorem~\ref{thm: one-way value A to B, half-tensor} one can obtain \[
\|\xi_{\mathcal  G}\|_{\mathcal  X} \leq \|\eta_{\mathcal  G}\|_{C_D \otimes_h (R_{O_B} \otimes_h C_{I_B} \otimes_h R_{O_A} \otimes C_{I_A})}^2 = \|\eta_{\mathcal  G}\|_{C_D \otimes_h (S_1(I_B,O_B) \otimes_h S_1(I_A,O_A))}^2 = \om_{A \to B}(\mathcal  G).
\] 

Our task will be now in proving $\om_{A \to B}(\mathcal  G) \leq \norm{\xi_{\mathcal  G}}_{\mathcal  X}$. To this end, consider a one-way strategy $\Phi \in \mathcal {OW}_{A \to B}(I_AI_B, O_AO_B)$, and let $C_{\Phi} \in S_1(O_AO_B) \otimes S_\infty(I_AI_B)$ denote the Choi matrix of $\Phi$. Via Proposition~\ref{prop: equivalence of QNS B not to A} there exists a positive semidefinite  matrix $\wh{P} \in S_1(O_A) \otimes S_\infty(I_A)$ such that $\tr_{S_1(O_A)} \wh{P} = 1_{S_\infty(I_A)}$ and $\tr_{S_1(O_B)} C_{\Phi} = \hat{P} \otimes 1_{S_\infty(I_B)}$. Since $\mathcal  X^*=S_\infty(I_A; S_1(O_A; S_\infty(I_B; S_1(O_B))))$,  due to the positivity of $C_{\Phi}$ we have (see Remark \ref{Rem: norm_p_positive}) \[
\norm{C_{\Phi}}_{\mathcal  X^*} = \sup \{\norm{(A \otimes \id)C_{\Phi}(A \otimes \id)}_{S_1(I_AO_A; S_\infty(I_B; S_1(O_B)))}\},
\] where the supremum is taken over all positive elements $A$ such that  $\norm{A}_{S_2(I_A)} \leq 1$.
For any such $A \in S_2(I_A)$, set $Y_A:= (A \otimes \id)C_{\Phi}(A \otimes \id)$. According to Remark \ref{Rem: norm_p_positive}, the positivity of $Y_A$ allows us to write 
\[
\inf \Big\{ \norm{ (C^{-\frac{1}{2}} \otimes \id)Y_A (C^{-\frac{1}{2}} \otimes \id)}_{S_\infty(I_AO_AI_B; S_1(O_B))}\Big\},
\] where the infimum runs over all positive elements $C \in S_\infty(I_AO_A)$ such that $\tr(C) = 1$. Using that $\norm{x}_{S_\infty^n(S_1^m)} = \norm{(\id \otimes \tr)(x)}_{S_\infty^n}$ for every positive element $x \in M_n \otimes M_m$  (see \cite[Eq. (3.20)]{Dev06}),  we calculate 
\begin{align*}
    &\norm{(C^{-\frac{1}{2}} \otimes \id)Y_A (C^{-\frac{1}{2}} \otimes \id)}_{S_\infty(I_AO_AI_B; S_1(O_B))} = \norm{\tr_{S_1(O_B)}(C^{-\frac{1}{2}} \otimes \id)Y_A (C^{-\frac{1}{2}} \otimes \id)}_{S_\infty(I_AO_AI_B)} \\&= \norm{(C^{-\frac{1}{2}} \otimes \id)(A \otimes \id)\tr_{S_1(O_B)}C_{\Phi}(A \otimes \id)(C^{-\frac{1}{2}} \otimes \id)}_{S_\infty(I_AO_AI_B)}.
\end{align*} Applying the condition that $\tr_{S_1(O_B)}C_{\Phi} = \wh{P} \otimes 1_{S_\infty(I_B)},$ the above norm becomes \begin{align*}
    &\norm{(C^{-\frac{1}{2}} \otimes \id)(A \otimes \id)(\wh{P} \otimes 1_{S_\infty(I_B)})(A \otimes \id)(C^{-\frac{1}{2}} \otimes \id)}_{S_\infty(I_AO_AI_B)}\\ &= \norm{(C^{-\frac{1}{2}} \otimes \id)(A \otimes \id)\wh{P}(A \otimes \id)(C^{-\frac{1}{2}} \otimes \id)}_{S_\infty(I_AO_A)}.
\end{align*} We have \begin{align*}
    &\sup_A \inf_C \norm{(C^{-\frac{1}{2}} \otimes \id)(A \otimes \id)\wh{P}(A \otimes \id)(C^{-\frac{1}{2}} \otimes \id)}_{S_\infty(I_AO_A)} 
    \\&= \norm{\wh{P}}_{S_\infty(I_A; S_1(O_A))} = \norm{\tr_{S_1(O_A)}\wh{P}}_{S_\infty(I_A)} = \norm{1_{S_\infty(I_A)}}_{S_\infty(I_A)} = 1.
\end{align*} This proves that for every one-way strategy $\Phi \in \mathcal{OW}_{A \to B}(I_{AB}, O_{AB})$, it holds that $C_{\Phi } \in B_{X^*}$. Hence, the desired inequality follows \[
\om_{A \to B}(\mathcal  G) \leq \norm{\xi_{\mathcal  G}}_{\mathcal  X}.
\]
\qedhere 
\end{proof}

\begin{remark}\label{Rem: comments_to_the_proof}
Note that, in order to prove the inequality 
\[
\|\xi_{\mathcal G}\|_{\mathcal X} 
\leq 
\|\eta_{\mathcal G}\|_{C_D \otimes_h (R_{O_B} \otimes_h C_{I_B} \otimes_h R_{O_A} \otimes_h C_{I_A})}^2,
\]
in the previous proof, we only used that 
$\xi_{\mathcal G} = \eta_{\mathcal G}^* \eta_{\mathcal G}$, 
but we didn't need any particular (game) structure of the tensor.

Note also that our proof above for the inequality 
\[
\omega_{A \to B}(\mathcal G) \leq \|\xi_{\mathcal G}\|_{\mathcal X}
\]
does not use the particular structure of the game $\mathcal G$, 
since it only considers the elements in the dual space. 
In particular, the previous proof shows that for every tensor $R$ we have
\[
\sup_{\Phi \in \mathcal{OW}_{A \to B}(I_{AB}, O_{AB})} 
|\langle C_\Phi, R \rangle| \leq \|R\|_{\mathcal X}.
\]
\end{remark}

Theorem~\ref{thm: one-way value, full tensor} generalizes the corresponding result in \cite{amr2019optimal} in which it is proven for a classical game $\mathcal  G$, that \[
\om_{A \to B}(\mathcal  G) = \norm{\xi_{\mathcal  G}}_{\ell_1(I_A, \ell_\infty(O_A, \ell_1(I_B, \ell_\infty(O_B))))}.
\]

\section{sub-no-signalling values and the symmetrized Haagerup norm}

In Section~\ref{sec: ns strategies and symmetrized one-way} we proved the equivalence of no-signalling strategies with the symmetrized one-way communication channels. Using these results, in Subsection~\ref{sec: one-way values of games} we proved that the one-way value of quantum games is obtained via calculating the norm of full tensor in the space $S_1(I_A; S_\infty(O_A; S_1(I_B; S_\infty(O_B))))$, or equivalently, via calculating the norm of the half-tensor in the space $C_D(S_1(I_B,O_B)\otimes_h S_1(I_A,O_A))$. Recall that such one-way strategies $\Phi: S_1(I_AI_B) \to S_1(O_AO_B)$ are quantum channels which factorize as $\Phi = (\id_{S_1(O_A)} \otimes \Phi_B)\circ (\id_{S_1(I_B)} \otimes \Phi_A)$, where \[
\Phi_A: S_1(I_A) \to S_1(O_AK), \Phi_B: S_1(KI_B) \to S_1(O_B),
\] are quantum channels. We now wish to consider the analogues of no-signalling strategies and one-way strategies where we consider decompositions of completely positive (not necessarily trace-preserving) operators by completely positive complete contractions. To this end, we will introduce the so-called \emph{sub}-regime of games and strategies, which appeared in the context of classical games (see \cite{amr2019optimal,lancien2015parallel}). 

\begin{definition}\label{defn: sub-one-way operators, A to B}
    Given a completely positive complete contraction $\Phi: S_1(I_AI_B) \to S_1(O_AO_B)$, we say that $\Phi$ is a \emph{sub-one-way} strategy from Alice to Bob if there exists a Hilbert space $K$, and completely positive complete contractions \[
\Phi_A: S_1(I_A) \to S_1(O_AK), \Phi_B: S_1(KI_B) \to S_1(O_B),
    \] such that $\Phi = (\id_{S_1(O_A)} \otimes \Phi_B) \circ (\id_{S_1(I_B)} \otimes \Phi_A)$. Let $\mathcal {SOW}_{A \to B}(I_{AB}, O_{AB})$ denote the set of completely positive maps which admit such a factorization. Switching the roles of $A$ and $B$, we similarly define the \emph{sub-one-way} strategies from Bob to Alice, and we will denote such strategies as $\mathcal {SOW}_{B \to A}(I_{AB}, O_{AB})$. We say $\Phi$ is \emph{symmetrized sub-one-way} if \[
\Phi \in \mathcal {SOW}(I_{AB}, O_{AB}):= \mathcal {SOW}_{A \to B}(I_{AB}, O_{AB}) \cap \mathcal {SOW}_{B \to A}(I_{AB}, O_{AB}).
    \]
\end{definition}

\begin{definition}\label{defn: sub ns, B to A}
    Given a completely positive but not necessarily trace preserving operator $\Phi: S_1(I_AI_B) \to S_1(O_AO_B)$, let $C_{\Phi^\dag}$ denote the Choi matrix of the adjoint $\Phi^\dag: S_\infty(O_AO_B) \to S_\infty(I_AI_B)$. We say $\Phi$ is \emph{sub-no-signalling} from Bob to Alice and write $\Phi \in \mathcal {SNOS}_{B \not\to A}(I_{AB}, O_{AB})$ if there exists a unital completely positive map $\Phi_A^\dag: S_\infty(O_A) \to S_\infty(I_A)$, with Choi matrix $C_{\Phi_A^\dag}$ such that \[
\tr_{S_\infty(O_B)} C_{\Phi^\dag} \leq C_{\Phi_A^\dag} \otimes 1_{S_\infty(I_B)}.
    \] We define analogously the set $\mathcal {SNOS}_{A\not\to B}(I_{AB}, O_{AB})$. We say the operator $\Phi: S_1(I_AI_B) \to S_1(O_AO_B)$ is \emph{sub-no-signalling} if \[
\Phi \in \mathcal {SNOS}(I_{AB}, O_{AB}):= \mathcal {SNOS}_{A \not\to B}(I_{AB}, O_{AB}) \cap \mathcal {SNOS}_{B\not\to A}(I_{AB}, O_{AB}).
    \]
\end{definition}

\begin{remark}\label{Rem: SNSQ and unit ball}
One can easily check that the proof of Theorem \ref{thm: one-way value, full tensor} can be adapted to show that every element $\Phi\in \mathcal{SNOS}_{A\not\to B}$ belongs to the unit ball of $\mathcal X^*$. Indeed, by following exactly the same argument we arrive at the point $$\|(C^{-\frac{1}{2}}\otimes \text{Id})(A\otimes \text{Id})\tr_{S_1(O_B)}C_{\Phi^\dag}(A\otimes \text{Id})(C^{-\frac{1}{2}}\otimes \text{Id})\|_{S_\infty(I_AO_AI_B)}.$$Since $A$ and $C$ are positive elements and $\tr_{S_\infty(O_B)}C_{\Phi^\dag}\leq C_{\Phi_A^\dag}\otimes 1_{S_\infty(I_B)}$, we immediately obtain that the previous norm is upper bounded by 
$$\|(C^{-\frac{1}{2}}\otimes \text{Id})(A\otimes \text{Id})(C_{\Phi_A^\dag}\otimes 1_{S_\infty(I_B)})(A\otimes \text{Id})(C^{-\frac{1}{2}}\otimes \text{Id})\|_{S_\infty(I_AO_AI_B)}$$and the rest of the proof follows in the same way as in the proof of Theorem \ref{thm: one-way value, full tensor}.

According to Remark \ref{Rem: comments_to_the_proof} this implies that for every tensor $R$ we have
\[
\sup_{\Phi \in \mathcal{SOW}_{A \to B}(I_{AB}, O_{AB})} 
|\langle C_\Phi, R \rangle| \leq \|R\|_{\mathcal X}.
\]
\end{remark}

In a similar fashion to our results from Section~\ref{sec: ns strategies and symmetrized one-way}, we now prove the correspondence between sub-one-way strategies and sub-no-signalling strategies. 

\begin{prop}\label{prop: equiv_SNSQ_SOW}
    The following equalities hold: \begin{align*}
        \mathcal {SOW}_{A \to B}(I_{AB}, O_{AB}) &= \mathcal {SNOS}_{B \not\to A}(I_{AB}, O_{AB}), \\ \mathcal {SOW}_{B \to A}(I_{AB}, O_{AB}) &= \mathcal {SNOS}_{A \not\to B}(I_{AB}, O_{AB}). 
    \end{align*} In particular, \begin{align*}
        \mathcal {SOW}(I_{AB}, O_{AB}) = \mathcal {SNOS}(I_{AB},O_{AB}).
    \end{align*}
\end{prop}

In the following, given a bipartite quantum game $\mathcal G$, we will denote
\begin{align*}
&\omega_{NOS} (\mathcal G)=\sup_{\varphi\in \mathcal {NOS}(I_{AB},O_{AB})}\langle \mathcal G,C_\varphi\rangle,\\ &\omega_{SNOS} (\mathcal G)=\sup_{\varphi\in \mathcal {SNOS}(I_{AB},O_{AB})}\langle \mathcal G,C_\varphi\rangle.
\end{align*}

\begin{proof}
    We will prove the coincidence \[
\mathcal {SOW}_{A \to B}(I_{AB}, O_{AB}) = \mathcal {SNOS}_{B \not\to A}(I_{AB}, O_{AB}),
    \] and the case interchanging the roles of Alice and Bob follows similarly. To this end, suppose $\Phi \in \mathcal {SOW}_{A \to B}(I_{AB}, O_{AB})$. Thus, we have a Hilbert space $K$ and completely positive complete contractions \[
\Phi_A: S_1(I_A) \to S_1(O_AK), \Phi_B: S_1(KI_B) \to S_1(O_B),
    \] such that \[
\Phi = (\id_{S_1(O_A)} \otimes \Phi_B)\circ (\id_{S_1(I_B)} \otimes \Phi_A).
    \] It then follows that the corresponding adjoint maps \[
\Phi_A^\dag: S_\infty(O_AK) \to S_\infty(I_A), \Phi_B^\dag: S_\infty(O_B) \to S_\infty(KI_B),
    \] satisfy \[\Phi^\dag = (\Phi_A^\dag \otimes \id_{S_\infty(I_B)}) \circ (\id_{S_\infty(O_A)} \otimes \Phi_B^\dag).\] The fact that both $\Phi_A^\dag,\Phi_B^\dag$ are completely positive complete contractions implies the elements \[
X:= 1_{S_\infty(I_A)} - \Phi_A^\dag(1_{S_{\infty}(O_AK)}), Y:= 1_{S_\infty(KI_B)} - \Phi_B^\dag(1_{S_\infty(O_B)}),
    \] are necessarily positive. Define the operators $\Psi_A: S_\infty(O_AK) \to S_\infty(I_A), \Psi_B: S_\infty(O_B) \to S_\infty(KI_B)$, by \[
\Psi_A(\cdot):= \Phi_A^\dag(\cdot) + \frac{\tr(\cdot)}{\dim O_AK}X, \Psi_B(\cdot):= \Phi_B^\dag(\cdot) + \frac{\tr(\cdot)}{\dim O_B}Y.
    \] Note both $\Psi_A$ and $\Psi_B$ are unital completely positive, and satisfy \[\Phi_A^\dag \preceq_{c.p.} \Psi_A, \Phi_B^\dag \preceq_{c.p.} \Psi_B.\] Setting \[
    \Psi:= (\Psi_A \otimes \id_{S_\infty(I_B)}) \circ (\id_{S_\infty(O_A)} \otimes \Psi_B): S_\infty(O_AO_B) \to S_\infty(I_AI_B),
    \] we obtain a unital completely positive map $\Psi$ such that  $\Phi^\dag \preceq_{c.p.} \Psi$. Equivalently, one has $C_{\Phi^\dag} \preceq C_{\Psi}$, where $C_{\Phi^\dag}, C_{\Psi}$ are the corresponding Choi matrices. Via Proposition~\ref{prop: equivalence of QNS B not to A} it follows $\tr_{S_1(O_B)} C_{\Psi} = \wh{P} \otimes 1_{S_\infty(I_B)}$, for some positive semidefinite  $\wh{P} \in S_1(O_A) \otimes S_\infty(I_A)$ such that $\tr_{S_1(O_A)}\wh{P} = 1_{S_\infty(I_A)}$, whose corresponding operator $\psi_A: S_\infty(O_A) \to S_\infty(I_A)$ is therefore unital completely positive. It follows \[
\tr_{S_1(O_B)}C_{\Phi^\dag} \preceq \tr_{S_1(O_B)}C_{\Psi} = \wh{P} \otimes 1_{S_\infty(I_B)}.
    \] This proves $\Phi \in \mathcal {SNOS}_{B \not\to A}(I_{AB}, O_{AB})$ proving \[
\mathcal {SOW}_{A \to B}(I_{AB}, O_{AB}) \subset \mathcal {SNOS}_{B \not\to A}(I_{AB}, O_{AB}).
    \]
In proving that $\mathcal {SNOS}_{B \not\to A}(I_{AB}, O_{AB}) \subset \mathcal {SOW}_{A \to B}(I_{AB}, O_{AB})$, we will follow a similar proof method to that done in Theorem~\ref{thm: QNS equals QOW}. To this end, let $\Phi: S_1(I_AI_B) \to S_1(O_AO_B)$ be a completely positive complete contraction, and let $C_{\Phi^\dag}$ denote the Choi matrix of the adjoint $\Phi^\dag: S_\infty(O_AO_B) \to S_\infty(I_AI_B)$. Suppose  $\Phi_A: S_1(I_A) \to S_1(O_A)$ is a quantum channel inducing the unital completely positive map $\Phi_A^\dag: S_\infty(O_A) \to S_\infty(I_A)$, with Choi matrix $C_{\Phi_A^\dag} \in S_1(O_A) \otimes S_\infty(I_A)$, satisfying \[
\tr_{S_1(O_B)} C_{\Phi^\dag} \preceq C_{\Phi_A^\dag}\otimes 1_{S_\infty(I_B)}.
\] Using Stinespring's theorem, there exist Hilbert spaces $H_{AB}, H_A$, a contraction $V_{AB}: I_AI_B \to O_AO_BH_{AB}$, and an isometry $V_A: I_A \to O_AH_A$, such that \[
\Phi(\rho) = \tr_{AB} V_{AB}\rho V_{AB}^\dag, \Phi_A(\s) = \tr_A V_A \s V_A^\dag, \forall \rho \in S_1(I_AI_B), \forall \s \in S_1(I_A).
\] 

Let us write \[
V_{AB}(\ket{xy}) = \sum_{ab} \ket{ab}\ket{\eta_{x,y}^{a,b}}, \{\ket{\eta_{x,y}^{a,b}}\} \subset H_{AB},\,\,\text{and}\,\,V_A(\ket{x}) = \sum_a \ket{a} \ket{\xi_x^a},
\{\ket{\xi_x^a}\} \subset H_A,
\] for every $x,y,a,b$. 

By setting  $\ket{\psi}:= \sum_{x,y,a} \l_{x,y}^a \ket{xya} \in I_AI_BO_A$ and using that $\tr_{S_1(O_B)} C_{\Phi^\dag} \preceq C_{\Phi_A^\dag}\otimes 1_{S_\infty(I_B)}$, we can deduce that 
\begin{align*}
\sum_{x,x^\prime,y,y^\prime,a, a^\prime, b} \l_{x,y}^a\oline{\l_{x^\prime,y^\prime}^{a^\prime}} \ip{\eta_{x^\prime,y^\prime}^{a^\prime,b}}{\eta_{x,y}^{a,b}} &=\bra{\psi} \tr_{S_1(O_B)} C_{\Phi^\dag} \ket{\psi} \\&\leq \bra{\psi}C_{\Phi_A^\dag} \otimes 1_{S_\infty(I_B)} \ket{\psi} \\&= \sum_{x,x^\prime,y,a,a^\prime} \l_{x,y}^a\oline{\l_{x^\prime,y}^{a^\prime}} \ip{\xi_{x^\prime}^{a^\prime}}{\xi_x^a}.
\end{align*}
Define the Hilbert space $K:= \Span\{\ket{\xi_x^a}\}_{x,a} \subset H_{A}$, and consider a map $\theta_B:KI_B \to O_BH_{AB}$ defined, by linearity, on simple tensors as \[
\ket{\xi_x^ay} \mapsto \sum_b \ket{b\eta_{x,y}^{a,b}},
\] for every $x,y,b$. Via our deduction that \[
\sum_{x,x^\prime,y,y^\prime,a, a^\prime, b} \l_{x,y}^a\oline{\l_{x^\prime,y^\prime}^{a^\prime}} \ip{\eta_{x^\prime,y^\prime}^{a^\prime,b}}{\eta_{x,y}^{a,b}} \leq \sum_{x,x^\prime,y,a,a^\prime} \l_{x,y}^a\oline{\l_{x^\prime,y}^{a^\prime}} \ip{\xi_{x^\prime}^{a^\prime}}{\xi_x^a},
\] it necessarily follows for all $(\l_{x,y}^a) \subset \mathbb C$ that \begin{align}
    \norm{ \theta_B(\sum_{x,y,a} \l_{x,y}^a \ket{\xi_x^ay})}_{H_{AB}O_B}^2 &\leq  \norm{\sum_{x,y,a}\l_{x,y}^a \ket{\xi_x^a y}}_{KI_B}^2,
\end{align} implying that $\theta_B$ is a well-defined linear contraction. 
    We trivially obtain \[
V_{AB} = (\id_{O_A} \otimes \theta_B) \circ (V_A \otimes \id_{I_B}).
    \]
At this point we define the completely positive complete contractions \begin{align*}
    &\tilde{\Phi}_A: S_1(I_A) \to S_1(O_AK), \rho \mapsto V_A \rho V_A^\dag, \\
    &\tilde{\Phi}_B: S_1(KI_B) \to S_1(O_B), \s \mapsto \tr_{H_{AB}}\theta_B \s \theta_B^\dag. 
\end{align*} We thus obtain the desired decomposition \[
\Phi = (\id_{S_1(O_A)} \otimes \tilde{\Phi}_B)\circ (\tilde{\Phi}_A \otimes \id_{S_1(I_B)}),
\] proving $\Phi \in \mathcal {SOW}_{A \to B}(I_{AB}, O_{AB})$. This finishes the proof. \qedhere 
\end{proof}

Our next theorem will be to describe the sub-no-signalling value of a two-player quantum game.
\begin{definition}
    Given a two-player quantum game $\mathcal  G$, then we define the \emph{sub-no-signalling value} $\om_{SNOS}(\mathcal  G)$ of the game as \begin{align}
\omega_{SNOS}(G)=\sup\left\{\tr\big(P\cdot (\varphi \otimes \text{Id}_{S_1(R)})(\psi)\big) \right\},
\end{align} where the supremum is taken over all sub-no-signalling strategies $\varphi \in\mathcal{SNOS}(I_{AB}, O_{AB})$. 
\end{definition}

As in previous sections, we will prove the norm equalities with respect to both the full tensor and the half tensor of a two-player game.  It is interesting to note that a direct consequence of Lemma \ref{lem: replace cc with cptp, two-player one-way} is that $$\omega_{A\to B}(G)=\sup\left\{\tr\big(P\cdot (\varphi \otimes \text{Id}_{S_1(R)})(\psi)\big) \right\},$$where the supremum runs over all $\mathcal {SOW}_{A\rightarrow B}(I_{AB}, O_{AB})$. That is, the one-way value of a two-prover quantum game coincides with its sub-one-way value. However, this equality is not clear at all when comparing the no-signalling value and the sub-no-signalling one. That is, when considering the correlations in the intersection of the spaces \[\mathcal {SOW}_{A\rightarrow B}(I_{AB}, O_{AB})\cap \mathcal {SOW}_{B\rightarrow A}(I_{AB}, O_{AB}).\]

In order to characterize the sub-no-signalling value of a two-prover quantum game, let us start by proving the following key lemma.

\begin{lemma}\label{lem: decomposition of cb from intersection operator space}
    Let $\mathcal  X_1, \mathcal  X_2$ be two operator spaces embedded (via continuous injections) into a linear space $\mathcal  Y$. Let $T: \mathcal  X_1 \cap \mathcal  X_2 \to \ell_2$ be a linear map such that $\cbnorm{T: \mathcal  X_1 \cap \mathcal  X_2 \to C} \leq \l$. Then we may write \[
T = \l_1W_1T_1 + \l_2W_2T_2,
    \] where $T_i: \mathcal  X_i \to K_i$ are linear maps for some Hilbert spaces $K_i$ such that $\cbnorm{T_i: \mathcal  X_i \to (K_i)_c} \leq 1,i=1,2,W_i: K_i \to \ell_2$ are linear maps such that $\norm{(W_1,W_2): K_1 \oplus_2 K_2 \to \ell_2} \leq 1,$ and $\l_i \geq 0$ such that $\l_1^2 + \l_2^2 \leq \l^2.$
\end{lemma}

\begin{proof}
    By homogeneity, we may assume the linear operator $T: \mathcal  X_1 \cap \mathcal  X_2 \to C$ is a complete contraction. Let $H_i$ be Hilbert spaces such that $j_i:\mathcal  X_i \hookrightarrow \mathcal B(H_i)$, completely isometrically. Then \[
j_1\otimes j_2:\mathcal  X_1 \cap \mathcal  X_2 \hookrightarrow \mathcal B(H_1) \oplus_\infty \mathcal B(H_2) \subset \mathcal B(H_1 \oplus_2 H_2),
    \]completely isometrically, where the first inclusion is the diagonal map $x \mapsto (j_1(x),j_2(x))$. Injectivity of the column operator space $C$ implies there exists an extension $\wt{T}: \mathcal B(H_1 \oplus_2 H_2) \to \ell_2$ satisfying \[
\cbnorm{\wt{T}: \mathcal B(H_1 \oplus_2 H_2) \to C} \leq 1.
    \] We then know (see for instance \cite[Theorem 14.1 and Remark 14.3]{pisier2012grothendieck}) there exists a state $\vp$ on $\mathcal B(H_1 \oplus_2 H_2)$ such that \begin{align}\label{eq: cc upperbounded by state}
\cbnorm{\wt{T}(z)}^2 \leq \vp(z^*z), 
    \end{align} for every $z \in \mathcal B(H_1 \oplus_2 H_2)$.
    
    Restricting to $\mathcal B(H_1) \oplus_\infty \mathcal B(H_2)$ we can find nonnegative $\l_i$ and states $\vp_i \in \mathcal B(H_i)^*$ such that $\l_1 + \l_2 =1$ and $\vp = \l_1 \vp_1 + \l_2\vp_2.$ At the same time we may define the closed subspace $A:= \{z \in \mathcal B(H_1 \oplus_2 H_2): \vp(z^*z) = 0\}$ and consider the completion of the quotient  \[
    L_2(\vp):= \oline{\mathcal B(H_1 \oplus_2 H_2)/A}^{\norm{\cdot}_{\vp}}, \] where $\norm{\cdot}_{\vp}$ is determined via the inner product \[\bra{[y_1]}\,\ket{[y_2]}:= \vp(y_1^*y_2).\] It follows by Equation~\eqref{eq: cc upperbounded by state} that the mapping \[
W: L_2(\vp) \to \ell_2, [x] \mapsto T(x),
    \] is well defined linear such that $\norm{W:L_2(\vp) \to \ell_2} \leq 1$. \newline For each $i =1,2$, define analogously $L_2(\vp_i):= \oline{\mathcal B(H_i)/A_i}^{\norm{\cdot}_{\vp_i}}$. Then we obtain an isometric inclusion \[
\iota_{\l_1,\l_2}: L_2(\vp_1) \oplus_2 L_2(\vp_2) \to L_2(\vp), ([y_1],[y_2]) \mapsto \Big[\Big(\frac{y_1}{\sqrt{\l_1}},\frac{y_2}{\sqrt{\l_2}}\Big)\Big]. 
    \] This follows by observing
    \begin{align*}
        \norm{\Big[\Big(\frac{y_1}{\sqrt{\l_1}},\frac{y_2}{\sqrt{\l_2}}\Big)\Big]}_{L_2(\vp)}^2 &= \vp\Big(\Big(\frac{y_1}{\sqrt{\l_1}},\frac{y_2}{\sqrt{\l_2}})^*(\frac{y_1}{\sqrt{\l_1}},\frac{y_2}{\sqrt{\l_2}}\Big)\Big)
        \\ &= \vp\Big( \frac{y_1^*y_1}{\l_1}, \frac{y_2^*y_2}{\l_2}\Big) \\&= \vp_1(y_1^*y_1) + \vp(y_2^*y_2) 
        \\&= \norm{[y_1]}_{L_2(\vp_1)}^2 + \norm{[y_2]}_{L_2(\vp_2)}^2
    \end{align*}

Denote $K_i:= L_2(\vp_i)$ and define the map $S_i: \mathcal B(H_i) \to K_i$ such that $S_i(y_i) = [y_i]_{\vp_i},i=1,2$. We observe \[
\norm{S_i(y_i)}^2 = \norm{[y_i]}_{L_2(\vp_i)}^2 = \vp_i(y_i^*y_i),
\] which implies that the mappings $S_i$ are well defined linear maps such that \[
\cbnorm{S_i: \mathcal B(H_i) \to (K_i)_c} \leq 1.
\] Let us now define $W_i: K_i \to \ell_2$ by $W_i:= W\circ \iota_i: K_i \to \ell_2$, where each $\iota_i: K_i \to L_2(\vp)$ denotes the corresponding inclusion $\iota_1: K_1 \to L_2(\vp), [y_1] \mapsto \frac{1}{\sqrt{\l_1}}[y_1,0], \iota_2: K_2 \to L_2(\vp), [y_2] \mapsto \frac{1}{\l_2}[0,y_2].$ The contractive inclusion \[
\iota_{\l_1,\l_2}: L_2(\vp_1) \oplus_2 L_2(\vp_2) \to L_2(\vp),
\] defined previously, together with the contractivity of the mapping $W: L_2(\vp) \to \ell_2$, imply the induced mapping satisfies \[
\norm{(W_1,W_2): K_1 \oplus_2 K_2 \to \ell_2} \leq 1. 
\]It is then readily checked that \[
T(x) = \wt{T}(j_1(x),j_2(x)) = \wt{T}(j_1(x),0) + \wt{T}(0,j_2(x)) = \sqrt{\l_1}W_1S_1(j_1(x)) + \sqrt{\l_2}W_2S_2(j_2(x)).
\] That is, $T_1:=S_1\circ j_1$ and $T_2:=S_2\circ j_2$ are the desired maps. This finishes the proof. \qedhere 
\end{proof}

\begin{theorem}\label{thm: sub-no-signalling, half tensor}
    Given a two-player quantum game $\mathcal  G$, we have \[
\om_{SNOS}(\mathcal  G) = \norm{\eta_{\mathcal  G}}_{C_D (S_1(I_A,O_A) \otimes_\mu S_1(I_B,O_B))}^2.
    \]
\end{theorem}

\begin{proof}
    First suppose $\norm{\eta_{\mathcal  G}}_{C_D (S_1(I_A,O_A) \otimes_\mu S_1(I_B,O_B))}^2 \leq \d$, with $\d >0$. Since the $\min$ norm is commutative, we can see the tensor  $\eta_{\mathcal G}$ in $(S_1(I_A, O_A) \otimes_\mu S_1(I_B, O_B)) \otimes_{\min} C_D$ by simply flipping it correspondingly, and the norm does not change. Moreover, we know (see Eq. (\ref{Eq: min_cb})) that $$\norm{\eta_{\mathcal  G}}_{(S_1(I_A,O_A) \otimes_\mu S_1(I_B,O_B))^* \otimes_{min}C_D}=\|T_{\eta_{\mathcal  G}}:(S_1(I_A,O_A) \otimes_\mu S_1(I_B,O_B))^*\to C_D\|_{cb},$$where $T_{\eta_{\mathcal  G}}$ is the linear map associated to the tensor $\eta_{\mathcal  G}$. Set $\mathcal  X_1:= S_\infty(I_B,O_B) \otimes_h S_\infty(I_A,O_A)$ and $\mathcal  X_2:= S_\infty(I_A,O_A) \otimes_h S_\infty(I_B,O_B)$. Note that $\mathcal  X_1 =\mathcal  X_2$ algebraically. Using Eq. (\ref{mu_quotient}) one can see that the $\d$-contractivity of the half tensor $\eta_{\mathcal  G}$ is equivalent to the induced operator $T_{\eta_{\mathcal  G}}: \mathcal  X_1 \cap \mathcal  X_2 \to C_D$ satisfying \[
\cbnorm{T_{\eta_{\mathcal  G}}: \mathcal  X_1 \cap \mathcal  X_2 \to C_D} \leq \sqrt{\d}.
    \] 
    
    Applying Lemma~\ref{lem: decomposition of cb from intersection operator space} there exist Hilbert spaces $K_i, i=1,2$, linear maps $T_i: \mathcal  X_i \to K_i, W_i: K_i \to \ell_2^D$ satisfying 
    \[
\cbnorm{T_i: \mathcal  X_i \to (K_i)_c} \leq 1,\,\,\norm{(W_1,W_2): K_1\oplus_2 K_2 \to \ell_2^D} \leq 1,
    \] and nonnegative numbers $\l_i$ satisfying $\l_1^2 + \l_2^2 \leq \d$, such that we may write $T_{\eta_{\mathcal  G}}$ as  \[
T_{\eta_{\mathcal  G}} = \l_1 W_1 T_1 + \l_2 W_2T_2.
    \] 
For each $i$, set $\eta_i \in \mathcal  X_i^* \otimes K_i$ to be the corresponding tensor of the operator $T_i: \mathcal  X_i \to K_i$. It then follows that we may write the half tensor $\eta_{\mathcal  G}$ as \[
\eta_{\mathcal  G} = \l_1(\id_{\mathcal  X_1^*} \otimes W_1)(\eta_1) + \l_2(\id_{\mathcal  X_2^*} \otimes W_2)(\eta_2).
\] Set $\wt{W}_i:= \id \otimes W_i$. The expression for our half tensor becomes \[
\eta_{\mathcal  G}  = \l_1\wt{W}_1(\eta_1) + \l_2\wt{W}_2 (\eta_2). 
\] The norm bound $\norm{(W_1,W_2): K_1 \oplus_2 K_2 \to \ell_2} \leq 1$ implies \begin{align*}
    \xi_{\mathcal G}=\eta_{\mathcal  G}^*\eta_{\mathcal  G} &= \begin{pmatrix}
        \l_1\eta^* & \l_2\eta^*
    \end{pmatrix} \begin{pmatrix}
        \wt{W}_1^* \\ \wt{W}_2^*
    \end{pmatrix} \begin{pmatrix}
        \wt{W}_1 & \wt{W}_2
    \end{pmatrix} \begin{pmatrix}
        \l_1\eta_1 \\ \l_2\eta_2
    \end{pmatrix} \\ &\leq \begin{pmatrix}
        {\l_1}\eta^*_1 & {\l_2}\eta^*_2
    \end{pmatrix}\begin{pmatrix}
        \l_1\eta_1 \\ \l_2\eta_2
    \end{pmatrix} \\ &= \l_1^2\eta_1^*\eta_1 + \l_2^2 \eta_2^*\eta_2.
\end{align*}

Since $C_\varphi\geq 0$ for every $\varphi \in \mathcal{SNOS}(I_{AB}, O_{AB})$, we know that 
\begin{align*}
\om_{SNOS}(\mathcal  G)&=\sup \{\langle \xi_G, C_\varphi\rangle\}=\sup \{\langle \eta^*\eta, C_\varphi\rangle\}\leq \sup  \{\lambda_1^2\langle \eta_1^*\eta_1, C_\varphi\rangle+\lambda_2^2\langle \eta_2^*\eta_2, C_\varphi\rangle\}\\&\leq \lambda_1^2\sup  \{\langle \eta_1^*\eta_1, C_\varphi\rangle\}+\lambda_2^2\sup\{\langle \eta_2^*\eta_2, C_\varphi\rangle\},
\end{align*}where all the supremums are taken when $\varphi \in \mathcal{SNOS}(I_{AB}, O_{AB})$. It is then clear that
\begin{align*}
\om_{SNOS}(\mathcal  G)\leq \lambda_1^2\sup  \{\langle \eta_1^*\eta_1, C_\varphi\rangle:\, \varphi\in \mathcal {SOW}_{A \to B}\}+\lambda_2^2\sup\{\langle \eta_2^*\eta_2, C_\varphi\rangle\, \varphi\in \mathcal {SOW}_{B \to A}\}.
\end{align*} Now, as we have seen in the proof of Theorem \ref{thm: one-way value, full tensor}, $\varphi\in \mathcal {SOW}_{A \to B}$ implies that $C_{\varphi}$ is in the unit ball of $S_\infty(I_A; S_1(O_A; S_\infty(I_B; S_1(O_B))))$. Therefore, we can upper bound 
\begin{align*}
\om_{SNOS}(\mathcal  G)&\leq \lambda_1^2\|\eta_1^*\eta_1\|_{S_1(I_A; S_\infty(O_A; S_1(I_B; S_\infty(O_B))))}+\lambda_2^2\|\eta_2^*\eta_2\|_{S_1(I_B; S_\infty(O_B; S_1(I_A; S_\infty(O_A))))}\\&\leq \lambda_1^2\|\eta_1\|^2_{C_D \otimes_h (S_1(I_B,O_B) \otimes_h S_1(I_A,O_A))}+\lambda_2^2\|\eta_2\|^2_{C_D \otimes_h (S_1(I_A,O_A) \otimes_h S_1(I_B,O_B))}\\&\leq \lambda_1^2+ \lambda_2^2\leq \delta,
\end{align*}where the second inequality follows from the first part of the proof of Theorem \ref{thm: one-way value, full tensor} (see Remark \ref{Rem: comments_to_the_proof}).

This finishes our claim that  \[
\om_{SNOS}(\mathcal  G) \leq \norm{\eta_{\mathcal  G}}_{C_D \otimes_{h} (S_1(I_A, O_A) \otimes_\mu S_1(I_B,O_B))}.\] 

For the reverse inequality we follow a similar line of thought as that of Theorem~\ref{thm: one-way value A to B, half-tensor}. To this end, suppose $\om_{SNOS}(\mathcal  G) \leq \d^2$. Given two complete contractions $T_A: S_1(I_A,O_A) \to S_\infty(\h), T_B: S_1(I_B, O_B) \to S_\infty(\h)$ with commuting ranges, we obtain \[
\norm{\id_{C_D} \otimes T_A \odot T_B(\eta_{\mathcal  G})}_{C_D \otimes_{min} S_\infty(\h)}^2 = \sup_{\ket{\vp} \in H} \sum_t \norm{A_t \ket{\vp}}_H^2,
\] where for each $t \in D$ we set \[A_t:= \sum_j T_A \odot T_B( \bra{\oline{\g}_t^j} \otimes \ket{\oline{\psi}_j}),\] and where the supremum is taken over all unit vectors $\ket{\vp} \in H$. As in Theorem~\ref{thm: one-way value A to B, half-tensor}, for each complete contraction $T_A, T_B$, we associate contractions \[
V_A: I_AH \to O_AH, V_B: I_BH \to O_BH,
\] so that the above supremum may be expressed as \[
\sum_t \Big\|\sum_j (\bra{{\g}_t^j} \otimes \id_E)((V_B \otimes \id_{O_A})(V_A \otimes \id_{I_B})(\ket{\psi_j\vp})\Big\|_H^2.
\] Commutativity of $T_A$ and $T_B$ imply \[
(\id_{O_A} \otimes V_B)(V_A \otimes \id_{I_B}) =(V_A \otimes \id_{I_B})(\id_{O_A} \otimes V_B). 
\] The previous expression coincides with \[
\tr (P \cdot (\id_{S_1(R)}\otimes \id_{S_1(O_A)} \otimes \vp_B) \circ (\id_{S_1(R)}\otimes \vp_A \otimes \id_{S_1(I_B)})(\op{\psi}{\psi})),
\] where the operators $\vp_A: S_1(I_A) \to S_1(O_AH), \vp_B: S_1(I_BH) \to S_1(O_B)$ are defined as \[
\vp_A(\s) = V_A(\s \otimes \op{\vp}{\vp})V_A^\dag,\,\, \vp_B(\rho) = \tr_H V_B \rho V_B^\dag.
\] Applying the commutativity relation between $V_A$ and $V_B$ enables us to write the strategy \[
(\id_{S_1(R)}\otimes \id_{S_1(O_A)}\otimes \vp_B) \circ (\id_{S_1(R)}\otimes \vp_A \otimes \id_{S_1(I_B)}),
\] as \[
(\id_{S_1(R)}\otimes \wt{\vp}_A \otimes \id_{S_1(O_B)}) \circ (\id_{S_1(R)}\otimes \id_{S_1(I_A)} \otimes\wt{\vp}_B),
\] where \[
\wt{\vp}_B(\s) = V_B(\s \otimes \op{\vp}{\vp})V_B^\dag,\,\, \wt{\vp}_A(\rho) = \tr_H(V_A \rho V_A^\dag).
\] We thus have the desired sub-no-signalling strategy which implies \[
\norm{\id_{C_D} \otimes T_A \odot T_B(\eta_{\mathcal  G})}_{C_D \otimes_{min} S_\infty(\h)}^2 \leq \d^2. 
\] This finishes the reverse inequality, and finishes the proof. \qedhere
\end{proof} 

As a corollary to Theorem~\ref{thm: sub-no-signalling, half tensor} we express the sub-no-signalling value of the game via the full game tensor: 

\begin{corollary}\label{cor: sub-no-signalling, full tensor}
    Given a two-player quantum game $\mathcal  G$ with full tensor $\xi_{\mathcal  G}$, then \[
\om_{SNOS}(\mathcal  G) = \inf \big\{ \norm{\xi_1}_{S_1(I_A; S_\infty(O_A; S_1(I_B; S_\infty(O_B))))} + \norm{\xi_2^T}_{S_1(I_B; S_\infty(O_B; S_1(I_A; S_\infty(O_A))))}\big \},
    \]where the infimum runs over all positive elements $\xi_1$, $\xi_2$ such that $\xi_G\leq \xi_1+\xi_2$ and the superscript $T$ refers to the suitable rearrangment of the tensor.
\end{corollary}

\begin{proof}
    Given two positive tensors $\xi_1,\xi_2$ such that $\xi_{\mathcal  G} \leq \xi_1 + \xi_2$, and using Remark~\ref{Rem: comments_to_the_proof} and the fact that elements in SQNS are completely positive maps, it necessarily follows
    \begin{align*}
\om_{SNOS}(\mathcal  G) &\leq \om_{SNOS}(\xi_1) + \om_{SNOS}(\xi_2) \leq \om_{A \to B}(\xi_1) + \om_{B \to A}(\xi_2)
\\&\leq  \norm{\xi_1}_{S_1(I_A; S_\infty(O_A; S_1(I_B; S_\infty(O_B))))} + \norm{\xi_2}_{S_1(I_B; S_\infty(O_B; S_1(I_A; S_\infty(O_A))))}.
    \end{align*} 
    
    Conversely, suppose $\om_{SNOS}(\mathcal  G) \leq \d^2$. Writing $\xi_{\mathcal  G} = \eta_{\mathcal  G}^* \eta_{\mathcal  G}$ with \[\norm{\eta_{\mathcal  G}}_{C_D (S_1(I_A,O_A) \otimes_\mu S_1(I_B,O_B))}^2 \leq \d^2, 
\] we are able to follow the analogous argument as in Theorem~\ref{thm: sub-no-signalling, half tensor}. We conclude \[
\xi_{\mathcal  G} \leq \l_1^2\eta_1^*\eta_1 + \l_2^2\eta_2^*\eta_2,
\] where $\|\eta_1\|_{C_D\otimes_{min}(S_1(I_B,O_B) \otimes_h S_1(I_A,O_A))}\leq 1$, $\|\eta_2\|_{C_D\otimes_{min}(S_1(I_A,O_A) \otimes_h S_1(I_B,O_B))}\leq 1$ and $\l_1^2+\l_2^2\leq \delta^2$. Defining $\xi_i:= \l_i^2\eta_i^*\eta_i$, we obtain the upper bound for $\xi_{\mathcal  G}\leq \xi_1+\xi_2$ such that \begin{align*}
&\norm{\xi_1}_{S_1(I_A; S_\infty(O_A; S_1(I_B; S_\infty(O_B))))} + \norm{\xi_2}_{S_1(I_B; S_\infty(O_B; S_1(I_A; S_\infty(O_A))))} \\&\leq  \l_1^2\|\eta_1\|^2_{C_D\otimes_{min}(S_1(I_B,O_B) \otimes_h S_1(I_A,O_A))}+ \l_2^2\|\eta_2\|^2_{C_D\otimes_{min}(S_1(I_A,O_A) \otimes_h S_1(I_B,O_B))} \\&\leq \d^2.
\end{align*} This finishes the proof. \qedhere 
\end{proof}

\begin{corollary}\label{cor: SNOS full tensor_II}
Let $\mathcal G$ be a quantum game with full tensor $\xi_G$. Then, 
\begin{align*}
\omega_{SNOS}(\mathcal G)=\sup\{\langle \xi_\mathcal G, \hat{\Phi}\rangle\},
\end{align*}where the sup runs over all positive elements $\hat{\Phi}$ such that $$\max\big\{\|\hat{\Phi}\|_{S_\infty(I_A; S_1(O_A; S_\infty(I_B; S_1(O_B))))}\leq 1, \, \|\hat{\Phi}^T\|_{S_\infty(I_B; S_1(O_B; S_\infty(I_A; S_1(O_A))))}\big\}\leq 1.$$
\end{corollary}
\begin{proof}
According to Lemma \ref{lem: duality of game tensors and strategies}, for every $\Phi\in \mathcal {SNOS}(I_{AB}, O_{AB})=\mathcal {SOW}_{A\rightarrow B}(I_{AB}, O_{AB})\cap \mathcal {SOW}_{B\rightarrow A}(I_{AB}, O_{AB})$, we have $$\tr\big(P\cdot (\text{Id}_{S_1(R)}\otimes \Phi)(\psi)\big)=\langle \xi_G, C_{\Phi}\rangle,$$where $C_{\Phi}$ is the Choi matrix of $\Phi$. According to Remark \ref{Rem: SNSQ and unit ball}, this implies 
\begin{align*}
\|C_{\Phi}\|_{S_\infty(I_A; S_1(O_A; S_\infty(I_B; S_1(O_B))))}\leq 1\text{ and }\|C_{\Phi}^T\|_{S_\infty(I_B; S_1(O_B; S_\infty(I_A; S_1(O_A))))}\leq 1.
\end{align*} Hence, one immediately obtains inequality $\leq$.

On the other hand, according to Corollary \ref{cor: sub-no-signalling, full tensor}, for every game $\mathcal G$ such that $\omega_{SNOS}(\mathcal G)<1$, there exist positive tensors $\xi_1$, $\xi_2$ such that $\xi_\mathcal G\leq \xi_1+\xi_2$ and 

\begin{align*}
\|\xi_1\|_{S_1(I_A; S_\infty(O_A; S_1(I_B; S_\infty(O_B)))}+\|\xi_2^T\|_{S_1(I_B; S_\infty(O_B; S_1(I_A; S_\infty(O_A)))}<1.
\end{align*}  

Now, if $C_{\Phi}$ is a positive element such that $$\max\{\|C_{\Phi}\|_{S_\infty(I_A; S_1(O_A; S_\infty(I_B; S_1(O_B))))}\leq 1, \, \|C_{\Phi}^T\|_{S_\infty(I_B; S_1(O_B; S_\infty(I_A; S_1(O_A))))}\}\leq 1,$$we have
\begin{align*}
\langle \xi_G, C_{\Phi}\rangle&\leq \langle \xi_1, C_{\Phi}\rangle+ \langle \xi_2, C_{\Phi}\rangle\\&\leq \|\xi_1\|_{S_1(I_A; S_\infty(O_A; S_1(I_B; S_\infty(O_B)))}+\|\xi_2^T\|_{S_1(I_B; S_\infty(O_B; S_1(I_A; S_\infty(O_A)))}\\&<1.
\end{align*}

We conclude the proof by homogeneity.
\end{proof}

It is clear that $\mathcal{NOS}(I_{AB}, O_{AB})\subseteq \mathcal{SNOS}(I_{AB}, O_{AB})$, which immediately implies that for every quantum game $\mathcal G$ one has $$\om_{NOS}(\mathcal G)\leq \om_{SNOS}(\mathcal G).$$  The sub-no-signalling correlations have been already studied in the context of classical games (see for instance \cite{amr2019optimal, lancien2015parallel}) and in that case, it can be proved that for every $\Phi\in \mathcal{SNOS}(I_{AB}, O_{AB})$ there exists a $\Phi'\in \mathcal{NOS}(I_{AB}, O_{AB})$ such that $\Phi\leq \Phi'$, from where one can easily deduce that $$\om_{NOS}(\mathcal G)= \om_{SNOS}(\mathcal G)$$  for every bipartite classical game. However, the proof of the existence of such a $\Phi'$ fails dramatically in the case of quantum games and at the moment of this writing we don't know if the no-signalling and the sub-no-signalling value of a bipartite quantum games are the same. Nevertheless, it can be proved that 
\begin{align}\label{Eq: SQNS=SNS}
\omega_{SNOS}(\mathcal G) \leq 8\,\omega_{NOS}(\mathcal G)
\end{align}
for every bipartite quantum game $\mathcal G$. 
That is, the quantities $\omega_{SNOS}(\mathcal G)$ and $\omega_{NOS}(\mathcal G)$ 
are equivalent up to a multiplicative constant factor. 

The comparison between \(\omega_{NOS}(G)\) and \(\omega_{SNOS}(G)\) is closely tied to the problem of lifting families of commuting partial isometries and commuting contractions, and the proof of (\ref{Eq: SQNS=SNS}) is technically involved. Since none of the results in this work rely on this estimate, and including the proof would significantly increase the length of the paper, we have chosen to omit it here. A complete exposition will be provided in a forthcoming work.

\section{Grothendieck's Theorem for Operator Spaces}\label{Sec: Grothendieck's Theorem for Operator Spaces}
Our goal in this section is to apply the theory developed in the previous sections to derive several consequences related to Grothendieck's Theorem. We begin with a consequence of our results on the sub-no-signalling and entangled values of games established earlier.
\begin{theorem}\label{thm: comparison of subns and entangled}
    The following assertions are equivalent: \begin{enumerate}
        \item There exists a constant $c_1$ such that for every bipartite quantum game $\mathcal  G$ one has \[
\om_{SNOS}(\mathcal  G) \leq c_1\om^*(\mathcal  G).
        \]
        \item There exists a constant $c_2$ such that for every tensor $t\in C\otimes S_1\otimes S_1$ that is obtained as a half tensor of a quantum game it holds that \[
\norm{t}_{C(S_1 \otimes_{\mu} S_1)} \leq c_2 \norm{t}_{C(S_1 \otimes_{min} S_1)}.
        \] If $c_1$ is the best constant satisfying (1) and $c_2$ is the best item satisfying (2), then \[
c_1 = c_2^2.
        \]
    \end{enumerate}
\end{theorem}

\begin{proof}
    First suppose there exists a universal constant $c_1$ such that \[
\om_{SNOS}(\mathcal  G) \leq c_1 \om^*(\mathcal  G),
    \] for all quantum games $\mathcal  G$, and let $\eta \in C_D \otimes S_1(I_A,O_A) \otimes S_1(I_B,O_B)$ be a tensor satisfying \[
\norm{\eta}_{C_D(S_1(I_A,O_A) \otimes_{min} S_1(I_B,O_B))}^2 \leq 1.
    \] Via Theorem~\ref{thm: entangled value, half-tensor} we then have the quantum game $\mathcal  G = \eta^*\eta$ satisfies \[
\om^*(\mathcal  G) \leq 1.
    \] 
Our assumption then implies $\om_{SNOS}(\mathcal  G) \leq c_1$. Consequently via Theorem~\ref{thm: sub-no-signalling, half tensor} we have \[\norm{\eta}_{C_D (S_1(I_A,O_A) \otimes_{\mu} S_1(I_B,O_B))} \leq \sqrt{c_1}.
\] This finishes the claim that (1) implies (2). \newline Suppose now Item (2) holds. Let $\mathcal  G$ be an arbitrary quantum game and assume $\om^*(\mathcal  G)=1$. Writing $\mathcal  G = \eta_{\mathcal  G}^*\eta_{\mathcal  G}$ with $\eta_{\mathcal  G} \in C_D \otimes S_1(I_A,O_A) \otimes S_1(I_B,O_B)$, the corresponding half tensor, we have via Theorem~\ref{thm: entangled value, half-tensor} that  \[
\norm{\eta_{\mathcal  G}}_{C_D(S_1(I_A,O_A) \otimes_{min} S_1(O_A,O_B))} = 1
\] Our assumption implies \[
\norm{\eta_{\mathcal  G}}_{C_D (S_1(I_A,O_A) \otimes_{\mu} S_1(I_B,O_B))} \leq {c_2},
\] for some universal constant $c_2$. Applying Theorem~\ref{thm: sub-no-signalling, half tensor} it follows \[
\om_{SNOS}(\mathcal  G) \leq c_2^2.
\] This proves Item (1), which finishes the equivalence. \qedhere 

\end{proof}

\begin{remark}
It is plausible that the previous theorem could be formulated in a more general setting, one in which the tensor \(t\) is not required to possess the specific structure arising from a nonlocal game, but instead where the statement holds for any tensor \(t \in C \otimes S_1 \otimes S_1\). However, establishing such a result would require modifying several of the arguments developed in this work in order to adapt them to this broader framework. On the other hand, in this section we present a counterexample to item (2) even within the more restrictive setting considered here, so there is no compelling reason to pursue this generalization.
\end{remark}

Remarkably, statement (2) in the previous theorem can be viewed as a restricted version of the Grothendieck problem for operator spaces. This problem, first formulated in \cite{Blecher92}, asks whether there exists a universal constant \( K \) such that  
\[
\|z\|_{S_\infty \otimes_{\min} (S_1 \otimes_{\mu} S_1)} 
   \leq K \|z\|_{S_\infty \otimes_{\min} (S_1 \otimes_{\min} S_1)}
\]
for every \( z \in S_\infty \otimes S_1 \otimes S_1 \). We refer the reader to \cite{Ara, pisier2012grothendieck} for a detailed discussion of this problem and its consequences.  

A negative answer to item~(2) above would immediately yield a negative answer to this question. This issue was resolved in the recent paper~\cite{Ara}, which was partly motivated by results in quantum information theory. However, the deeper analysis developed in the present work---particularly Theorem \ref{thm: comparison of subns and entangled}---provides a much clearer explanation of the counterexample exhibited in~\cite{Ara}.  It should be emphasized, however, that the main result in~\cite{Ara} is considerably stronger than the one obtained in this work, as it provides a counterexample to item~(2) above in the case where the spaces $S_1$ are replaced by their commutative analogue~$\ell_1$---a situation to which the approach developed in the present paper does not apply.

In fact, we know examples of even classical games for which item~(1) of Theorem~ \ref{thm: comparison of subns and entangled} fails, thereby implying the failure of item~(2) as well. Thus, the counterexample presented in~\cite{Ara} can be precisely derived from certain games already known in the field of quantum information theory. Although this connection was briefly mentioned in~\cite{Ara}, Theorem~\ref{thm: comparison of subns and entangled} fully clarifies the situation. Let us now look into this connection in detail in the next section.

\subsubsection{An explicit counterexample to item (2) in Theorem \ref{thm: comparison of subns and entangled}}

In the work~\cite{bavarian15}, the following two-prover (classical) game was considered. 
Let us fix $n$, which we assume to be a prime number (or a power of a prime number), and let $\mathbb{Z}_n$ denote the cyclic group of order~$n$. 
In the $\mathcal G=\mathrm{CHSH}_n$ game, the referee chooses questions $(x,y)$ uniformly at random from $\mathbb{Z}_n \times \mathbb{Z}_n$ and sends question~$x$ to Alice and question~$y$ to Bob. 
Each party must produce an output $a,b \in \mathbb{Z}_n$, respectively. 
The players win the game if their outputs satisfy the equation $a + b = xy$, where both the addition and the product are taken modulo~$n$. 
This means that the game is defined by the probability distribution $\pi(x,y) = 1/n^2$ for every $x,y$, and the predicate function is $V(a,b \mid x,y) = \delta_{a + b = xy}$. 
It was proved in~\cite{bavarian15} that $\omega^*(\mathcal G) \leq C / \sqrt{n}$ for some universal constant~$C$, while it is easy to see that $\omega_{\mathrm{NOS}}(\mathcal G) = 1$. 
According to the explanation given in the introduction of the present work, if we consider the tensor $$z=\frac{1}{n^2}\sum_{\substack{x,y,a,b\in \mathbb Z_n
    \\ \hspace{-0.2 cm}a+b=xy}}(e_x\otimes e_a)\otimes (e_y\otimes e_b)\in \ell_1^n(\ell_\infty^n)\otimes \ell_1^n(\ell_\infty^n),$$one has
    \begin{align*}
    \|z\|_{\ell_1^n(\ell_\infty^n)\otimes_{min} \ell_1^n(\ell_\infty^n)}=\omega^*(\mathcal G)\leq  C\frac{1}{\sqrt{n}},\hspace{0.4 cm}\|z\|_{NOS}=\omega_{NOS}(\mathcal G)=1, \end{align*}
    where the second norm is defined in (\ref{Eq: NS_Classical_Games}).
    
 Now, in order to produce our counterexample, we need to compute the half-tensor of the game. To do so, we first reinterpret our previous classical game as a quantum game. This can be achieved by embedding the previous $\ell_p$ spaces as the diagonal of $S_p$. 

In other words, we consider the full tensor
\begin{align}\label{Eq: xi_G_CHSH_n}
\xi_\mathcal G = \frac{1}{n^2}\sum_{\substack{x,y,a,b\in \mathbb{Z}_n \\[-0.2em] a+b=xy}}
  (|x\rangle\langle x|\otimes |a\rangle\langle a|)\otimes (|y\rangle\langle y|\otimes |b\rangle\langle b|)
  \;\in\; S_1^n(S_\infty^n)\otimes S_1^n(S_\infty^n),
\end{align}
and we obtain
\begin{align}\label{classic_to_quantum_ex}
\|\xi_\mathcal G \|_{S_1^n(S_\infty^n)\otimes_{\min} S_1^n(S_\infty^n)} = \omega^*(\mathcal G)
  \;\leq\; C\,\frac{1}{\sqrt{n}}, 
  \qquad 
  \|\xi_\mathcal G \|_{\mathrm{SNOS}} = \omega_{NOS}(\mathcal G)= 1,
\end{align}
where the second norm is defined as in Corollary~\ref{cor: sub-no-signalling, full tensor}.  In the above expression, $|p\rangle\langle p|$ denotes the matrix whose only nonzero entry is $1$ in the $(p,p)$ position. Note that we have used the fact that, for bipartite classical games, the no-signalling and sub-no-signalling values coincide (see \cite{amr2019optimal, lancien2015parallel}). 
\begin{remark}
It is very clear that the inclusion
\[
\ell_1^n(\ell_\infty^n)\otimes_{min} \ell_1^n(\ell_\infty^n)
\subset
S_1^n(S_\infty^n)\otimes_{\min} S_1^n(S_\infty^n)
\]
via the embedding of each of the commutative $\ell_p$ spaces into the diagonal of the non-commutative ones defines a linear isometry.
Although we have not shown that a similar statement holds for the sub--no-signalling value, it is obvious from its very definition that the sub--no-signalling value of a classical game, when viewed as a quantum game, is at least the (sub--)no-signalling value of the game when viewed as a classical game. This explians the second equality in the Equation (\ref{classic_to_quantum_ex}) above.
\end{remark}

We find it interesting to explain how this game should be understood in terms of the input, following the framework introduced in Section~\ref{Sec: Games and Tensors}. 
To do so, we must describe the input state $\ket{\psi} \in R I_A I_B$ and the projection $P \in S_\infty(R O_A O_B)$ that define the quantum game~$\mathcal{G}$. Let us consider $I_A=I_B=O_A=O_B=\C^n$ and $R=\C^n\otimes\C^n$. Then, we consider $$\ket{\psi}=\frac{1}{n}\sum_{x,y\in\Z_n}|xy\rangle_R|x\rangle_A|y\rangle_B\in RI_AI_B,$$and 
 $$P=\sum_{\substack{x,y,a,b\in \mathbb{Z}_n \\[-0.2em] a+b=xy}}|x y\rangle\langle x y|_R\otimes |a\rangle\langle a|_A\otimes |b\rangle\langle b |_B\in P(RO_AO_B).$$  Here, $\ket{xy}_R = \ket{x}_R \otimes \ket{y}_R \in \mathbb{C}^n \otimes \mathbb{C}^n$ and $\ket{xy}\!\bra{xy}_R = \ket{x}\!\bra{x}_R \otimes \ket{y}\!\bra{y}_R$. 
Thus, the notation $xy$ here is purely a shorthand for the tensor product and does not denote any multiplication. According to these definitions we see that 
\begin{align*}
\xi_{\mathcal  G} &= \tr_R\left(\left(\psi \otimes 1_{S_\infty(O)} \right)\left(P^{T_O} \otimes 1_{S_1(I)} \right)  \right) \\&=  \frac{1}{n^2}\sum_{\substack{x,y,a,b\in \mathbb{Z}_n \\[-0.2em] a+b=xy}}
  |xy\rangle\langle xy|\otimes |ab\rangle\langle ab|\in S_1(I_AI_B)\otimes S_\infty(O_AO_B).\end{align*}
  
  If we want to realize this element as a tensor in $S_1(I_A,S_\infty(O_A))\otimes S_1(I_B,S_\infty(O_B))$, we obtain
  \begin{align*} 
\xi_{\mathcal  G}=  \frac{1}{n^2}\sum_{\substack{x,y,a,b\in \mathbb{Z}_n \\[-0.2em] a+b=xy}}
  (|x\rangle\langle x|\otimes |a\rangle\langle a|)\otimes (|y\rangle\langle y|\otimes |b\rangle\langle b|)
  \;\in\; S_1(I_A,S_\infty(O_A))\otimes S_1(I_B,S_\infty(O_B)),
    \end{align*}as we anticipated in Eq.~\eqref{Eq: xi_G_CHSH_n}.

The next step is to find the half-tensor $\eta_{\mathcal{G}}$ such that 
\[
\xi_\mathcal G = \eta_{\mathcal{G}}^* \eta_{\mathcal{G}}.
\]
Let us consider the element
\[
\eta_{\mathcal{G}}
  = \frac{1}{n}\sum_{\substack{x,y,a,b\in \mathbb{Z}_n \\[-0.2em] a+b=xy}}
  |1\rangle \langle xy|\otimes |xyab\rangle\langle ab|\in S_1(I_{AB})\otimes S_\infty(O_AO_B, D),
\]where, again, notation $xy$, $xyab$ and $ab$ here is purely a shorthand for the tensor product and $D$ is the Hilbert spaces spanned by the elements $\text{span}\{ |xyab\rangle:\, x,y,a,b\in \mathbb{Z}_n, a+b=xy\}$. Then, it is easy to check that 
 \begin{align*} 
\eta_{\mathcal{G}}^* \eta_{\mathcal{G}}=\frac{1}{n^2}\sum_{\substack{x,y,a,b\in \mathbb{Z}_n \\[-0.2em] a+b=xy}}
  |xy\rangle\langle xy|\otimes |ab\rangle\langle ab|=\xi_{\mathcal  G} \in S_1(I_{AB})\otimes S_\infty(O_AB).
    \end{align*}

Then, according to the identification explained in Eq. (\ref{Eq._Half_Tensor}) we must consider the element 
\[
\eta_{\mathcal{G}}
  = \frac{1}{n}\sum_{\substack{x,y,a,b\in \mathbb{Z}_n \\[-0.2em] a+b=xy}}
  |xyab\rangle \otimes \langle ab|\otimes |xy\rangle\in C_D\otimes R_{O_AO_B}\otimes C_{I_AI_B}.\]
  
  If we want to realize this element as a tensor in $C_D\otimes S_1(I_A,O_A)\otimes S_1(I_B,O_B)$ we must write
  \[
\eta_{\mathcal{G}}
  = \frac{1}{n}\sum_{\substack{x,y,a,b\in \mathbb{Z}_n \\[-0.2em] a+b=xy}}
  |xyab\rangle \otimes |a\rangle\langle x| \otimes |b\rangle\langle y| \in C_D\otimes S_1(I_A,O_A)\otimes S_1(I_B,O_B).\]

Hence, according to Theorem~\ref{thm: comparison of subns and entangled}, this tensor satisfies
\[
\|\eta_{\mathcal{G}}\|_{C^{n^4}(S_1^n\otimes_{\min} S_1^n)}
  \;\leq\; C'\, n^{-1/4},
  \qquad
  \|\eta_{\mathcal{G}}\|_{C^{n^4}(S_1^n\otimes_{\mu} S_1^n)}
  \;=\; 1,
\]
which provides the desired counterexample.

We note that this is exactly the same tensor, up to the normalization constant~$n$, that was used in the proof of~\cite[Theorem~3.1]{Ara} to construct a counterexample to Grothendieck's  theorem for operator spaces.

\subsubsection{An upper bound}
We finish this section by studying the comparison of $c_1, c_2$ in Theorem \ref{thm: comparison of subns and entangled} quantitatively. In particular, we wish to express the constant $c_2$ as a function of the rank of the game tensor. 

First, given linear spaces $X$, $Y$ and $Z$ and  a tensor $z \in X \otimes Y \otimes Z$, we define the \emph{rank} of $z$ as
\[
\operatorname{rank}(z)
   = \min \Bigl\{ r \in \mathbb{N} :
      \operatorname{rank}(z_{X|YZ}),
      \operatorname{rank}(z_{Y|XZ}),
      \operatorname{rank}(z_{Z|XY})  \Bigr\},
\]where $\operatorname{rank}_{\mathrm{mat}}(z_{X|YZ})$ denotes the rank of the tensor $z$
when viewed as a bipartite tensor in $X \otimes (Y \otimes Z)$, and similarly for the other bipartitions.

The main result of this section is the following.
\begin{theorem}\label{thm: upper_bound_mu_min}
    Given any two-prover quantum game $\mathcal G(\ket{\psi}, P)$ with corresponding half tensor $\eta_\mathcal G$ we have the following inequality: \[
\norm{\eta_{\mathcal G}}_{C_D (S_1(I_A,O_A) \otimes_\mu S_1(I_B,O_B))} \leq K\sqrt{\operatorname{rank}(\eta_\mathcal G)}\norm{\eta_\mathcal G}_{C_D(S_1(I_A,O_A) \otimes_{min} S_1(I_B,O_B))}, 
    \]where $K$ is a universal constant.
\end{theorem}

In \cite{Crann}, the rank of a game $\mathcal{G}$ is defined as the rank of the projection $P$ that defines the game, which coincides with the dimension of the column space $C_D$ introduced above. Let us denote this rank by $\operatorname{rank}_{\mathcal{G}}$.

Then, in \cite[Corollary 5.18]{Crann}, it is shown that
\[
\omega_{qc}(\mathcal{G}) \leq 4 \, \operatorname{rank}_{\mathcal{G}} \, \omega^*(\mathcal{G}),
\]
where $\omega_{qc}(\mathcal{G})$ is a value associated with the game $\mathcal{G}$, which is known to satisfy $\omega_{qc}(\mathcal{G}) \leq \omega_{NOS}(\mathcal{G})$. Since $\operatorname{rank}_{\mathcal{G}} \leq \operatorname{rank}(\eta_\mathcal{G})$, Theorem \ref{thm: upper_bound_mu_min} improves the upper bound shown in \cite{Crann} in two different ways.

On the one hand, our theorem implies that
\[
\omega_{NOS}(\mathcal{G}) \leq K^2 \, \operatorname{rank}_{\mathcal{G}} \, \omega^*(\mathcal{G}),
\]
so the previously shown upper bound also holds for the quantity $\omega_{NOS}(\mathcal{G})$. On the other hand, Theorem \ref{thm: upper_bound_mu_min} is stated in terms of a much more general definition of rank, not just the dimension of $C_D$, but also the dimension of the spaces $S_1$ as well as any possible subspace of them.

In fact, we will show in a forthcoming paper that the previous theorem is optimal up to logarithmic factors. Namely, for every (high enough) $n$ there exist tensors $z \in C \otimes S_1 \otimes S_1$ of rank $\operatorname{rank}(z) = n$ such that 
\[
\frac{\|\eta_{\mathcal{G}}\|_{C(S_1 \otimes_\mu S_1)}}{\|\eta_{\mathcal{G}}\|_{C (S_1 \otimes_{\min} S_1)}} \geq K \frac{\sqrt{n}}{\log^\beta n},
\]
for some universal positive constants $K$ and $\beta$.

In order to prove Theorem \ref{thm: upper_bound_mu_min}, we will need some technical results. First, recall that, given a linear map $T : X \to Y$ between operator spaces the \emph{$\gamma_{R \oplus C}$-norm} of $T$ is defined by
\[
\gamma_{R \oplus C}(U)
=\inf\left\{\, \|a\|_{\mathrm{cb}}\,\|b\|_{\mathrm{cb}} :
T = b \circ a,\;
a : X \to R \oplus C,\;
b : R \oplus C \to Y \,\right\},
\]
where the infimum is taken over all factorizations of $T$ through the direct sum of the row and column operator spaces $R$ and $C$. Here, the operator space $R \oplus C$ is defined as the direct sum
\[
R \oplus C = \{ (x, y) : x \in R,\, y \in C \},
\]
equipped with the matrix norms
\[
\|(x, y)\|_{M_n(R \oplus C)} = \max\{\, \|x\|_{M_n(R)},\, \|y\|_{M_n(C)} \,\},
\quad x \in M_n(R),\; y \in M_n(C).
\]

We shall also make use of the notion of a \textit{$2$-summing operator} and some of its basic properties. 
Let $T:X\rightarrow Y$ be a linear operator between Banach spaces. 
We say that $T$ is \textit{$2$-summing} whenever
\begin{align}\label{Equation p-summing}
\pi_2(T):=\|id\otimes T:\ell_2\otimes_\epsilon X\rightarrow \ell_2(X)\|<\infty.
\end{align}
It can be seen that $\pi_2$ defines a norm on the collection of all $2$-summing operators. 

The \textit{factorization theorem} for such operators \cite[Theorem 2.13]{DiJaTo95} states that 
$T:X\rightarrow Y$ is $2$-summing if and only if there exist a regular Borel probability measure 
$\mu$ on the unit ball of the dual space $B_{X^*}$, and a linear map 
$u:L_2(C(B_{X^*}), \mu)\rightarrow Y$ with $\|u\|=\pi_2(T)$ such that the following diagram commutes:
$$
\xymatrix@R=0.6cm@C=1.5cm{
  {C(B_{X^*})}\ar[r]^{i} &
  {L_2(\mu)\ar[d]^{u}} \\
  {X}\ar[u]^{j}\ar[r]^{T} & {Y}}
$$
Here, $j:X\hookrightarrow C(B_{X^*})$ denotes the canonical embedding, and 
$i:C(B_{X^*})\rightarrow L_p(C(B_{X^*}), \mu)$ is the identity map. 

From this factorization result one can easily deduce that $2$-summing operators enjoy the 
\textit{extension property} \cite[Theorem 4.15]{DiJaTo95}: 
given a $2$-summing operator $T:X\rightarrow Y$ and any isometric embedding $j:X\hookrightarrow \tilde{X}$, 
there exists an extension $\tilde{T}:\tilde{X}\rightarrow Y$ (so that $T=\tilde{T}\circ j$) satisfying 
$\pi_2(T)=\pi_2(\tilde{T})$.

A classical result, known as the \textit{little Grothendieck theorem}, asserts that every bounded linear map 
$T:A\rightarrow L_2(\mu)$, where $A$ is a commutative C$^*$-algebra, satisfies 
$\pi_2(T)\leq K_{LG}\|T\|$. 
Here, $K_{LG}=\sqrt{\pi/2}$ in the real case, and $K_{LG}=2/\sqrt{\pi}$ in the complex case.

\begin{lemma}\label{lemma: 2-sum}
	Let $H$ be a Hilbert space endowed with one of the following operator space structures: $R$, $C$ or $R\cap C$ and let $A$ be a commutative $C^*$-algebra. Then, for any linear map $T:A\rightarrow H$ we have $$\|T\|_{cb}=\pi_2(T)\leq K_{GL}\|T\|.$$
	\end{lemma}

There also exists a noncommutative analogue of this theorem, first proved in \cite{Pisier78} and later refined in 
\cite{Haagerup85}. 
This result, commonly referred to as the \textit{noncommutative little Grothendieck theorem}, 
implies, in particular, that for any (non-necessarily commutative) C$^*$-algebra $A$ and any Hilbert space $H$, 
every bounded linear operator $T:H\rightarrow A^*$ satisfies
\begin{align}\label{Eq: NC_LG}
\|T:R\cap C\rightarrow A^*\|_{cb}\leq 2\|T\|.
\end{align}

The reader may find a more detailed discussion on $p$-summing operators and the Little Grothendieck Theorem, in the form in which we are using them, in \cite[Section 2]{junge22}.

The first lemma we need to prove is the following.
\begin{lemma}
Let $E$ be an $n$-dimensional operator space and let $j:E\iny S_1$ be a complete isometry. Then, there exists a linear map $v:S_1\to S_1$ such that $\gamma_{R\oplus C}(v)\leq C\sqrt{n}$ and $v\circ j=j$.
\end{lemma}
\begin{proof}
According to John's Theorem \cite[Proposition 3.8]{Pis89} there exist linear maps $a:E\to \ell_2^n$ such that $\pi_2(a)\leq \sqrt{n}$ and $\|b:\ell_2^n\to E\|\leq 1$ satisfying $b\circ a=\id_E$.

Now, because of the extension property of the $2$-summing operators, we can extend the map $a$ to $\tilde{a}:S_1\to \ell_2^n$ so that $\pi_2(\tilde{a})\leq \sqrt{n}$ (note that $\tilde{a}\circ j=a$). Moreover, since $2$-summing operators factor through an $L_\infty$-space, one can deduce that $\|\tilde{a}:S_1\to R_n\cap C_n\|_{cb}\leq C\sqrt{n}$ for a given universal constant $C$. 

Since $j:E\iny S_1$ is a complete isometry,  according to Eq. (\ref{Eq: NC_LG}), $$\|j\circ b:R_n\cap C_n\to  S_1\|_{cb}\leq C'\|j\circ b:\ell_2^n\to  S_1\|\leq C'$$ for a certain universal constant $C'$. Moreover, the fact that $R_n\cap C_n$ is an exact operator space guarantees (\cite[Corollary 0.7]{Pi02}) that $$\gamma_{R\oplus C}(j\circ b:R_n\cap C_n\to  S_1)\leq C''\|j\circ b:R_n\cap C_n\to  S_1\|_{cb}\leq C'''.$$ That is, there exist linear maps $u$ and $v$ such that  $\|u:R_n\cap C_n\to  R\oplus C\|_{cb}\leq C'''$, $\|v:R\oplus C\to S_1 \|_{cb}\leq 1$ and such that $j\circ b=v\circ u$. Hence, by defining $T=u\circ\tilde{a}$ and $S=v$ we have that $v=S\circ T:S_1\to S_1$ is a linear map which factors through $R\oplus C$, $\|S:R\oplus C\to S_1 \|_{cb}\leq 1$, $\|T:S_1\to R\oplus C\|_{cb}\leq \|u:R_n\cap C_n\to  R\oplus C\|_{cb}\|\tilde{a}:S_1\to R_n\cap C_n\|_{cb}\leq C''''\sqrt{n}$ and $$v\circ j=S\circ T\circ j=v \circ u\circ a=j\circ b\circ a=j\circ \id_E=j.$$  This concludes the proof.
\end{proof}

\begin{prop}\label{prop: min_mu_cb}
Let $E$ be an $n$-dimensional operator space, let $j:E\iny S_1$ be a complete isometry and let $X$ be any operator space. Then, $$\|j\otimes \id_X:E\otimes_{min} X\to S_1\otimes_{\mu} X\|_{cb}\leq C\sqrt{n},$$where $C$ is a universal constant.
\end{prop}

\begin{proof}
First of all, let us note that Eqs.~(\ref{Eq: indetity_min_h_1}) and (\ref{Eq: indetity_min_h_2}) imply that 
\[
C \otimes_{\min} X = C \otimes_{\mu} X \quad \text{and} \quad X \otimes_{\min} R = X \otimes_{\mu} R
\] 
completely isometrically. Moreover, since these two norms are commutative, it follows that 
\begin{align}\label{Eq: min_mu_R}
R \otimes_{\min} X = R \otimes_{\mu} X
\end{align} 
completely isometrically as well.

Now, we want to prove that 
$$\|j\otimes \id:M_d(E\otimes_{min} X)\to M_d(S_1\otimes_{\mu} X)\|\leq C\sqrt{n}$$ for every $d$, where $\id=\id_{M_d}\otimes \id_X$. To this end, let us consider an element $t\in M_d(E\otimes_{min} X)$ and show that $$\|(j\otimes \id)(t)\|_{M_d(S_1\otimes_{\mu} X)}\leq C\sqrt{n}\|t\|_{M_d(E\otimes_{min} X)}.$$ According to the previous lemma,   $$\|(j\otimes \id)(t)\|_{M_d(S_1\otimes_{\mu} X)}=\|((v\circ j)\otimes \id)(t)\|_{M_d(S_1\otimes_{\mu} X)},$$ where $v:S_1\to S_1$ is such that $\gamma_{R\oplus C}(v)\leq C\sqrt{n}$.

Now, the estimate $\gamma_{R\oplus C}(v)\leq C\sqrt{n}$ implies that $v=v_1+v_2$ such that $$\gamma_{R}(v_1:S_1\rightarrow S_1)+\gamma_{C}(v_1:S_1\rightarrow S_1)\leq C\sqrt{n}.$$ Therefore, 
$$\|((v\circ j)\otimes \id)(t)\|_{M_d(S_1\otimes_{\mu} X)}\leq \|((v_1\circ j)\otimes \id)(t)\|_{M_d(S_1\otimes_{\mu} X)}+\|((v_2\circ j)\otimes \id)(t)\|_{M_d(S_1\otimes_{\mu} X)}.$$

Let us first prove that 
\begin{align}\label{Eq: v__com_j}
\|((v_1\circ j)\otimes \id)(t)\|_{M_d(S_1\otimes_{\mu} X)}\leq \gamma_{R}(v_1:S_1\rightarrow S_1)\|t\|_{M_d(E\otimes_{min} X)},
\end{align}which will follow from the estimate $$\|(v_1\circ j)\otimes \id:M_d(E\otimes_{min} X)\to M_d(S_1\otimes_{\mu} X)\|\leq  \gamma_{R}(v_1:S_1\rightarrow S_1).$$ Let us write $v_1=w_2\circ w_1$, where $w_1$ and $w_2$ are linear maps satisfying  $$\|w_1:S_1\to R\|_{cb}\|w_2:R\to S_1\|_{cb}\leq \gamma_{R}(v_1:S_1\rightarrow S_1)$$ so that we can decompose $$(v_1\circ j)\otimes \id= (w_2\otimes \id)\circ(w_1\otimes \id)\circ (j\otimes \id).$$ We know that 
\begin{align*}
&\|j\otimes \id:M_d(E\otimes_{min} X)\to M_d(S_1\otimes_{min} X)\|\leq 1,\\&
\|w_1\otimes \id:M_d(S_1\otimes_{min} X)\to M_d(R\otimes_{\mu} X)\|\leq \|w_1:S_1\to R\|_{cb},\\&
\|w_2\otimes \id:M_d(R\otimes_{\mu} X)\to M_d(S_1\otimes_{\mu} X)\|\leq \|w_2:R\to S_1\|_{cb},
\end{align*}where in the second estimate we have used Eq. (\ref{Eq: min_mu_R}). It follows that 

\begin{align*}
\|(v_1\circ j)\otimes \id:M_d(E\otimes_{min} X)\to M_d(S_1\otimes_{\mu} X)\|&\leq \|w_1:S_1\to R\|_{cb}\|w_2:R\to S_1\|_{cb}\\&\leq \gamma_{R}(v_1:S_1\rightarrow S_1).
\end{align*} So we obtain the desired estimate (\ref{Eq: v__com_j}).

A completely analogous proof shows that $$\|((v_2\circ j)\otimes \id)(t)\|_{M_d(S_1\otimes_{\mu} X)}\leq \gamma_{C}(v_2:S_1\rightarrow S_1)\|t\|_{M_d(E\otimes_{min} X)}.$$

Hence, we conclude that
\begin{align*}
\|((v\circ j)\otimes \id)(t)\|_{M_d(S_1\otimes_{\mu} X)}&\leq \Big(\gamma_{R}(v_1:S_1\rightarrow S_1)+ \gamma_{C}(v_2:S_1\rightarrow S_1)\Big)\|t\|_{M_d(E\otimes_{min} X)}\\&\leq C\sqrt{n}\|t\|_{M_d(E\otimes_{min} X)},
\end{align*}as we wanted.

\end{proof}

\begin{remark}
Note that the previous result is not true if one replaces the $\mu$ norm by the Haagerup tensor norm. Indeed, in that case, if we set $E=R_n$ and $X=C_n$, the injectivity of the Haagerup tensor norm implies
\begin{align*}
\|j\otimes \id_{C_n}:R_n\otimes_{min} C_n\to S_1^n\otimes_{h} C_n\|_{cb}&=\|\id_{R_n}\otimes \id_{C_n}:R_n\otimes_{min} C_n\to R_n\otimes_{h} C_n\|_{cb}\\&\geq\|\id_{R_n}\otimes \id_{C_n}:R_n\otimes_{min} C_n\to R_n\otimes_{h} C_n\|\geq n,
\end{align*}where the last estimate follows from the isometric identifications $R_n\otimes_{min} C_n=C_n\otimes_{min} R_n=S_\infty^n$ and $R_n\otimes_{h} C_n=S_1^n$.
\end{remark}

Finally, Theorem \ref{thm: upper_bound_mu_min} immediately follows from the next corollary.
\begin{cor}
Let $z\in C\otimes S_1\otimes S_1$ be a  tensor of rank equal $n$. Then, $$\|z\|_{C(S_1\otimes_\mu S_1)}\leq C\sqrt{n}\|z\|_{C (S_1\otimes_{min} S_1)},$$where $C$ is a universal constant. 
\end{cor}
\begin{proof}
Let us start assuming that $z=(j\otimes \id)(t)$ for a certain $t\in C\otimes E\otimes S_1$, where $E$ is an $n$-dimensional operator space,  $j:E\iny S_1$ is complete isometry and $\id=\id_C\otimes \id_{S_1}$. According to  Proposition \ref{prop: min_mu_cb} we have 
\begin{align*}
&\|z\|_{C\otimes_{min} (S_1\otimes_\mu S_1)}\leq C\sqrt{n}\|t\|_{C\otimes_{min} (E\otimes_{min} S_1)}=C\sqrt{n}\|z\|_{C\otimes_{min} (S_1\otimes_{min} S_1)},
\end{align*}where in the last equality we have used the injectivity of the minimal norm and the fact that $j$ is a complete isometry.

Of course this argument works in the same way for $t\in C\otimes S_1\otimes E$.

Let us now assume that $u=(j\otimes \id)(t)$ for a certain $t\in E\otimes S_1\otimes S_1$, where $E$ is an $n$-dimensional operator space,  $j:E\iny C$ is complete isometry and $\id=\id_{S_1}\otimes \id_{S_1}$. Since $C$ is an homogeneous operator space (see \cite[Section 9.2]{Pi03}), we can assume that $u\in C_n\otimes S_1 \otimes S_1$ and we must show that 
$$\|u\|_{C_n\otimes_{min} (S_1\otimes_\mu S_1)}\leq C\sqrt{n}\|t\|_{C_n\otimes_{min} (S_1\otimes_{min} S_1)}.$$

Now, it is well known and easy to see that $\|\id:C_n\to \ell_\infty^n\|_{cb}=\|\id:\ell_2^n\to \ell_\infty^n\|=1$ and $\|\id:\ell_\infty^n\to C_n\|_{cb}=\sqrt{n}$. Then, 
\begin{align*}
\|u\|_{C_n\otimes_{min} (S_1\otimes_\mu S_1)}&\leq \sqrt{n}\|u\|_{\ell_\infty^n\otimes_{min} (S_1\otimes_\mu S_1)}=\sqrt{n}\|u\|_{\ell_\infty^n(S_1\otimes_\mu S_1)}\\&\leq 2\sqrt{n}\|u\|_{\ell_\infty^n(S_1\otimes_{min} S_1)}= C\sqrt{n}\|u\|_{\ell_\infty^n\otimes_{min} (S_1\otimes_{min} S_1)}\\&\leq 2\sqrt{n}\|u\|_{C_n\otimes_{\min} (S_1\otimes_{min} S_1)}.
\end{align*}

Here, we have used that $\ell_\infty\otimes_{min} X=\ell_\infty(X)$ isometrically and also Eq. (\ref{Eq: GT_OS}).

This concludes the proof.
\end{proof}

\end{document}